\documentclass[11pt]{article}

\usepackage[utf8]{inputenc}
\usepackage{amsmath}
\usepackage{amsfonts}
\usepackage{amssymb}
\usepackage{amsthm}
\usepackage{thmtools}
\usepackage{xcolor}
\usepackage{lmodern} %

\usepackage{array}
\usepackage{verbatim}
\usepackage{float}
\usepackage{multirow}
\usepackage{enumitem}
\usepackage[procnumbered,ruled,vlined,linesnumbered, algo2e]{algorithm2e}
 
\usepackage[T1]{fontenc}

\DontPrintSemicolon
\SetKw{KwAnd}{and}
\SetProcNameSty{textsc}
\SetFuncSty{textsc}

\definecolor{ForestGreen}{rgb}{0.1333,0.5451,0.1333}
\definecolor{DarkRed}{rgb}{0.8,0,0}
\definecolor{Red}{rgb}{1,0,0}
\usepackage[linktocpage=true,
pagebackref=true,colorlinks,
linkcolor=DarkRed,citecolor=ForestGreen,
bookmarks,bookmarksopen,bookmarksnumbered]
{hyperref}

\usepackage{cleveref}
\usepackage{geometry}
\global\long\def\vol{\mathrm{vol}}
\global\long\def\prune{\mathtt{Prune}}
\global\long\def\round{\mathtt{Round}}
\newcommand{\polylog}{\operatorname{polylog}}
\newcommand{\poly}{\operatorname{poly}}

\newcommand{\dist}{\operatorname{dist}}

\newcommand{\deltamax}{\Delta_{\max}}
\newcommand{\deltamin}{\Delta_{\min}}
\newcommand{\sampleb}{\textrm{{\sc SampleVertex}}}
\newcommand{\Otil}{\tilde{O}}
\newcommand{\inc}{\textsc{Inc}}
\newcommand{\tschedule}{T_{\schedule}}

\global\long\def\A{\mathcal{A}}
\global\long\def\B{\mathcal{B}}
\global\long\def\C{\mathcal{C}}
\global\long\def\P{\mathcal{P}}
\global\long\def\S{\mathcal{S}}

\global\long\def\F{\mathcal{F}}

\global\long\def\R{\mathbb{R}}

\global\long\def\property{\eqref{con:identity}, \eqref{con:union}, \eqref{con:contraction}, \eqref{con:increment}, and \eqref{con:nested}}

\declaretheorem[numberwithin=section]{theorem}
\declaretheorem[numberlike=theorem]{lemma}
\declaretheorem[numberlike=theorem]{proposition}
\declaretheorem[numberlike=theorem]{corollary}
\declaretheorem[numberlike=theorem]{definition}
\declaretheorem[numberlike=theorem]{claim}
\declaretheorem[numberlike=theorem]{algorithm}

\declaretheorem[numberlike=theorem,refname={Fact,Facts},Refname={Fact,Facts},name={Fact}]{fact}
\declaretheorem[numberlike=theorem]{observation}

\providecommand{\tabularnewline}{\\}
\floatstyle{ruled}
\newfloat{algorithm}{tbp}{loa}
\providecommand{\algorithmname}{Algorithm}
\floatname{algorithm}{\protect\algorithmname}
\crefname{algorithm}{Algorithm}{Algorithms}

\makeatletter
\renewcommand{\paragraph}{%
	\@startsection{paragraph}{4}%
	{\z@}{1.25ex \@plus 1ex \@minus .2ex}{-1em}%
	{\normalfont\normalsize\bfseries}%
}
\makeatother

\ifdefined\ShowComment
	\newcommand{\jan}[1]{\textcolor{blue}{Jan: #1}}
	\newcommand{\he}[1]{\textcolor{red}{He: #1}}
	\newcommand{\thatchaphol}[1]{\textcolor{purple}{TS: #1}}
	
	\newcommand{\mpg}[1]{\textcolor{olive}{Max: #1}}

    \newcommand{\sidford}[1]{\textcolor{DarkRed}{Sidford: #1}}

 \newcommand{\danupon}[1]{\textcolor{orange}{Danupon: #1}}

\else
	\newcommand{\jan}[1]{}
	\newcommand{\thatchaphol}[1]{}
	\newcommand{\he}[1]{}
	\newcommand{\mpg}[1]{}
    \newcommand{\sidford}[1]{}
    \def\danupon#1{}
\fi

\newcommand{\schedule}{\mathrm{schedule}}
\newcommand{\rhoValue}{\frac{2^{16}(\alpha+1)\log n \Delta_{max} }{\Delta_{min}^2\phi^2}}

\title{Fully-Dynamic Graph Sparsifiers Against an Adaptive Adversary}

\usepackage{authblk}
\makeatletter
\renewcommand\AB@affilsepx{, \protect\Affilfont}
\makeatother

\author[1]{Aaron Bernstein}
\author[2]{Jan van den Brand}
\author[3]{Maximilian Probst Gutenberg}
\author[2]{Danupon Nanongkai}
\author[4]{Thatchaphol Saranurak\thanks{Work partially done while at KTH}}
\author[5]{Aaron Sidford}
\author[6]{He Sun}

\affil[1]{Rutgers University}
\affil[2]{KTH Royal Institute of Technology}
\affil[3]{University of Copenhagen}
\affil[4]{TTIC}
\affil[5]{Stanford University}
\affil[6]{University of Edinburgh}

\date{}

\begin{document}

\begin{titlepage}
	\maketitle
	\pagenumbering{gobble}
	\ifdefined\ShowComment
	\begin{center}
		{\centering\huge\textcolor{red}{DEBUG VERSION}}
	\end{center}
	\fi
	\begin{abstract}
Designing efficient dynamic graph algorithms against an {\em adaptive} adversary is a major goal in the field of dynamic graph algorithms. Compared to most graph primitives (e.g.~spanning trees and matchings), designing such algorithms for {\em graph spanners} and (more broadly) {\em graph sparsifiers} posts a unique challenge due to the inherent need of randomness for static computation and the lack of a way to adjust the output slowly (known as ``small recourse/replacements'').  
	
This paper presents the first non-trivial efficient adaptive algorithms for many sparsifiers: against an adaptive adversary. Specifically, we present algorithms that maintain  
\begin{enumerate}[noitemsep]
		\item  a $\polylog(n)$-spanner of size $\Otil(n)$ in $\polylog(n)$ amortized update time,
		\item an $O(k)$-approximate cut sparsifier of size $\Otil(n)$ in $\tilde{O}(n^{1/k})$ amortized update time, and 
		\item a $\polylog(n)$-approximate spectral sparsifier in $\polylog(n)$ amortized update time.
\end{enumerate}
Our bounds are the first non-trivial ones even when only the recourse is concerned. Our results hold even against a stronger adversary, who can access the random bits previously used by the algorithms. Our spanner result resolves an open question by Ahmed~et~al.~(2019). The amortized update time of all algorithms can be made worst-case by paying sub-polynomial factors. Our results and techniques also imply improvements over existing results, including (i) answering open questions about decremental single-source shortest paths by Chuzhoy and Khanna (STOC'19) and Gutenberg and Wulff-Nilsen (SODA'20), implying a nearly-quadratic time algorithm for approximating minimum-cost unit-capacity flow and (ii) de-amortizing the result of Abraham et al.~(FOCS'16) for dynamic spectral sparsifiers.

Our results are based on two novel techniques. 
The first technique is a generic
	black-box reduction that allows us to assume that the graph is initially an expander with almost uniform-degree and, more importantly, stays as an almost uniform-degree expander while undergoing only edge deletions.
	The second technique is called \emph{proactive
		resampling}:
	here we constantly \emph{re-sample} parts of the input graph so that, independent of an adversary's  computational power, a desired structure of the underlying graph can be  always maintained. 
	Despite its simplicity, the analysis of this sampling scheme is far from trivial,  because the adversary can potentially create dependencies between the random choices used by the algorithm.
	We believe these two techniques could be useful for developing other adaptive algorithms. 
\end{abstract}

	\newpage 
	
	\setcounter{tocdepth}{2}
	\tableofcontents
	
\end{titlepage}
\pagenumbering{arabic}

\part{Background}
\section{Introduction}\label{sec:intro}

Dynamic graph algorithms maintain information in an input graph undergoing edge {\em updates}, which typically take the form of edge insertions and deletions. Many efficient algorithms have been developed in this setting, such as those for maintaining a minimum spanning tree, maximum matching, shortest distances, and sparsifiers. However, many of these algorithms are randomized and, more importantly, make the so-called
{\em oblivious adversary assumption}, which assumes that each update given to the algorithm cannot depend on earlier query-answers of the algorithms. In other words, the whole update sequence is fixed by some adversary in advance, and then each update is given to the algorithm one by one. 
This assumption is crucial for many recent advances in the design of efficient randomized algorithms for dynamic problems
(e.g.  \cite{KapronKM13,HenzingerKN14_focs,chechik,BernsteinPW19,BaswanaKS12,AbrahamDKKP16,Solomon16,BaswanaGS18,BhattacharyaCHN18}).

The oblivious-adversary assumption significantly limits the use of dynamic algorithms in certain interactive environments and, in particular, the setting where these dynamic algorithms are employed as \emph{subroutines} for other algorithms. 
For example, \cite{ChuGPSSW18} pointed out that their goal of computing a (static) short cycle decomposition could have been achieved easily using existing dynamic spanner algorithms, 
if such algorithms worked without the oblivious-adversary assumption.
In addition, if the partially-dynamic single-source shortest paths algorithm of \cite{HenzingerKN14_focs} worked without this assumption, we would have an almost-linear time approximate min-cost flow and balanced separator algorithm in the static setting   (see, e.g., \cite{BernsteinC16,Bernstein17,gutenberg2020decremental,gutenberg2020deterministic}, for recent developments in this direction). Because of this,  designing dynamic algorithms 
\emph{without}  the oblivious adversary assumption has become  a major goal in the field of dynamic graph algorithms in recent years. We call such algorithms {\em adaptive} and say that they work against an {\em adaptive adversary}.

Thanks to the recent efforts in developing adaptive algorithms, such algorithms now exist for maintaining a number of graph primitives, such as minimum spanning trees with  bounded worst-case update time \cite{ChuzhoyGLNPS19,NanongkaiSW17,NanongkaiS17,Wulff-Nilsen17}, partially-dynamic single-source shortest paths \cite{EvenS, ChuzhoyK19,BernsteinC16, BernsteinChechikSparse, Bernstein17, gutenberg2020deterministic, gutenberg2020decremental, GutenbergWW20}, and fully-dynamic matching \cite{BhattacharyaHI15,BhattacharyaHN16,BhattacharyaHN17,BhattacharyaK19,Wajc19}. %
This line of research on dynamic graph algorithms has also brought new insights on algorithm design in the static setting (e.g. flow, vertex connectivity, matching, and traveling salesman problem \cite{Madry10,ChuzhoyK19,BrandLNPSSSW20,ChekuriQ2017near,ChekuriQ2018fast}).
One very recent exciting application is the use of an adaptive dynamic algorithm called {\em expander decomposition} to compute maximum-weight matching and related problems in nearly-linear time on moderately dense graphs \cite{BrandLNPSSSW20}.

\paragraph{Graph sparsifiers.} Despite the fast recent progress, very little was known for certain important primitives, like maintaining {\em graph sparsifiers} against an adaptive adversary.
To formalize our discussion, we say that a sparsifier of a graph $G=(V, E)$ is a sparse graph $H=(V, E')$ that approximately preserves properties of $G$, such as  all cuts ({\em cut sparsifiers}), all-pairs distances ({\em spanners}), and spectral properties ({\em spectral sparsifiers}). 
For any integer $\alpha\geq 1$, an {\em $\alpha$-spanner} of graph $G=(V,E)$ is a subgraph $H$ such that for any pair of nodes $(u,v)$, the distance between $u$ and $v$ in $H$ is at
most $\alpha$ times their distance in $G$. 
An {\em $\alpha$-cut sparsifier} of $G$ is a sparse graph $H$ that
preserves all cut sizes up to an $\alpha$ factor: that is, $\delta_H(S)\in [\delta_G(S), \alpha \delta_G(S)]$ for every $S\subseteq V$, where $\delta_G(S)$ (respectively $\delta_H(S)$) is the total weight of edges  between $S$ and $V\setminus S$ in $G$ (respectively $H$) for any $S\subset V$.
An $\alpha$-spectral sparsifier is a stronger object than an $\alpha$-cut sparsifier, but we defer the formal definition to \Cref{sec:preliminiaries}.

A dynamic algorithm for maintaining a spanner or a cut sparsifier is given a weighted undirected $n$-node graph $G$ to preprocess, and returns a spanner or cut-sparsifier $H$ of $G$. After this, it must process a sequence of {\em updates}, each of which  is an edge insertion or deletion of $G$. After each update, the algorithm outputs edges to be inserted and deleted to $H$ so that the updated $H$ remains an $\alpha$-spanner of the updated $G$.
The  algorithm's performance is measured by the {\em preprocessing time} (the time to preprocess $G$ initially); the {\em update time} (the time to process each update); the {\em stretch} (the value of $\alpha$); and the {\em size} of the spanner (the number of edges).
The update time is typically categorized into two types: \emph{amortized case} update time and \emph{worst case} update time. The more desirable one is the {\em worst-case} update time  which holds for every single update. This is in contrast to an {\em amortized} update time which holds ``on average''; i.e., for any $t$, an algorithm is said to have an amortized update time of $t$ if, for any $k$, the total time it spends to process the first $k$ updates is at most $kt$.

Spanners and cut sparsifers are fundamental objects that have been studied extensively in various settings (e.g., \cite{AlthoferDDJS93,RodittyZ11,BaswanaS07,RodittyTZ05,DerbelMZ10,GrossmanP17,GhaffariK18,BenczurK15,FungHHP11,AhnGM12,AhnGM13}). 
In the dynamic setting,  they have been actively studied  since 2005 (e.g. \cite{AusielloFI06,FeigenbaumKMSZ05,Baswana08,Elkin11,BernsteinR11,BaswanaKS12,BodwinK16,BernsteinFH19}). A fairly tight algorithm with amortized update time for maintaining dynamic spanners was   known in 2008 due to Baswana~et~al.~\cite{BaswanaKS12}. For any $k\geq 1$, their algorithm   maintains, with high probability, a $(2k-1)$-spanner of size  $\tilde{O}(kn^{1+1/k})$ in $O(k^{2}\log^{2}n)$  amortized update time\footnote{Throughout, $\tilde O$ hides $O(\polylog(n))$. With high probability (w.h.p.) means with probability at least $1-1/n^c$ for any constant $c>1$.}. 
The stretch and size tradeoff is almost tight assuming  Erd\H{o}s' girth conjecture, which implies that a $(2k-1)$-spanner must contain $\Omega(n^{1+1/k})$ edges. 
Recently, Bernstein et al. \cite{BernsteinFH19} showed how to ``de-amortize'' the result of Baswana~et~al.~\cite{BaswanaKS12}, giving an algorithm that  in $O(1)^k \log^3(n)$ {\em worst-case} update time maintains, w.h.p., a $(2k-1)$-spanner of size $\tilde O(n^{1+1/k})$.
We refer the reader to 
 \Cref{tab:spanner} for other related results and clear comparison. 
For dynamic cut sparsifiers, the only result we are aware of is  \cite{AbrahamDKKP16}, which maintains a $(1+\epsilon)$-cut sparsifier in polylogarithmic worst-case update time. \cite{AbrahamDKKP16} can also maintain a $(1+\epsilon)$-spectral sparsifier within the same update time, but this holds only for the amortized update time.

Similar to other dynamic algorithms, most existing dynamic spanner and cut sparsifier algorithms are {\em not} adaptive. The exceptions are the algorithms of \cite{AusielloFI06}, which can maintain 
a $3$-spanner (respectively a $5$-spanner) of size $O(n^{1+1/2})$ (respectively  $O(n^{1+1/3})$) in $O(\Delta)$ time, where $\Delta$ is the maximum degree. These algorithms are deterministic, and thus work against an adaptive adversary. Since $\Delta$ can be as large as $\Omega(n)$, their update time is  rather inefficient as typically $\poly\log(n)$ or $n^{o(1)}$ update times are desired. Designing dynamic algorithms for this low update time is one of the major objectives of our paper.

\global\long\def\H{\mathcal{H}}

\paragraph{Challenges.} Developing efficient adaptive algorithms for maintaining graph sparsifiers posts great challenges in the general research program towards adaptive dynamic algorithms. 
First, computing most sparsifiers inherently relies on the use of randomness. Even in the static setting, existing fast algorithms for constructing cut and spectral sparsifiers are all randomized, and  known deterministic algorithms   require  $\Omega(n^4)$ time~\cite{BatsonSS14,zouzias2012matrix}. In fact, a nearly-linear time deterministic algorithm for a certain cut sparsifier would resolve a major open problem about computing the minimum cut deterministically \cite{KawarabayashiT2015deterministic,GawrychowskiMW20,MukhopadhyayN2019weighted}.
Thus, in contrast to other primitives,  such as the minimum spanning tree or approximate maximum matching, for which  efficient deterministic algorithms exist in the static setting, there is little chance to dynamically  maintain sparsifiers {\em deterministically}. (Deterministic dynamic algorithms always work against adaptive adversaries.) 

Secondly, even if we allowed infinite update time and focused on the strictly simpler objective of minimizing the changes in the maintained sparsifier (the so-called {\em recourse} or {\em replacements} in online algorithms), it is entirely unclear from existing techniques whether it is possible to maintain such a sparsifier against an adaptive adversary while only making (amortized) $\polylog(n)$ changes to the sparsifier per update to the input graph. For example, an $O(\log n)$-spanner of $\tilde O(n)$ edges can be easily maintained with $\tilde O(n)$ recourse per update by replacing the entire spanner by a new one after every update. Is it possible that an adversary who can see the output spanner can make a few changes to the graph so that a new $O(\log n)$-spanner has to change completely? An answer to this question is unclear. This is in contrast to most dynamic graph primitives where bounding the changes is obvious even against adaptive adversaries. For example, it can be easily shown that the minimum spanning tree requires at most one edge insertion and one edge deletion after each update to the input graph.

Designing algorithms with low recourse is a prerequisite for fast dynamic algorithms, and there are several graph problems where low-recourse algorithms were the crucial bottleneck, e.g.~maximal independent set \cite{Censor-HillelHK16,AssadiOSS18sublinear_m,ChechikZ19,BehnezhadDHSS19,Monemizadeh19}, planar embeddings \cite{ChechikZ19,BehnezhadDHSS19}, and topological sorting \cite{BernsteinC18cycle}. In particular, the lack of recourse-efficient algorithms makes it very challenging to maintain sparsifiers against an adaptive adversary.

\subsection{Our Results}
We show how to dynamically maintain both spanners, cut sparsifiers, and spectral sparsifiers against an adaptive adversary in \emph{poly-logarithmic} update time and recourse. 
We summarize these results as follows:

\begin{theorem}
[Adaptive Spanner]\label{thm:main spanner}There is a randomized
\emph{adaptive} algorithm that, given an $n$-vertex graph undergoing
edge insertions and deletions, with high probability, explicitly maintains
a $\polylog(n)$-spanner of size $\Otil(n)$ using $\polylog(n)$
amortized update time.
\end{theorem}

\begin{theorem}
[Adaptive Cut Sparsifier]\label{thm:main cut}There is a  randomized
\emph{adaptive} algorithm that, given an $n$-vertex graph undergoing
edge insertions and deletions and a parameter $k\ge1$, with high
probability, maintains an $O(k)$-cut sparsifier of size $\tilde{O}(n)$
using $\tilde{O}(n^{1/k})$ amortized update time, which is $\polylog(n)$
time when $k=\log n$.
\end{theorem}

\begin{theorem}
[Adaptive Spectral Sparsifier]\label{thm:main spectral adaptive}There is a  randomized
\emph{adaptive} algorithm that, given an $n$-vertex graph undergoing
edge insertions and deletions, with high
probability, maintains an $\polylog(n)$-spectral sparsifier of size $\tilde{O}(n)$
using $\polylog(n)$ amortized update time.
\end{theorem}

All results above hold even against a stronger adversary, called {\em randomness-adaptive} in \cite{NanongkaiS17}. This adversary can access the random bits previously used by our algorithms (but not the future random bits).
\Cref{thm:main spanner} is the first algorithm with $o(n)$ update time against an adaptive adversary, and answers the open problem in \cite{AhmedB_survey}. The only previous adaptive algorithm is by \cite{AusielloFI06} which can take $O(n)$ update time. 
No non-trivial dynamic adaptive algorithm for
cut sparsifiers and spectral sparsifiers is known before \Cref{thm:main cut,thm:main spectral adaptive}. 

Compared to results assuming the oblivious-adversary assumption (e.g.~\cite{BaswanaKS12,BernsteinFH19,AbrahamDKKP16}), our bounds are not as tight. For example,  \Cref{thm:main spanner} does not achieve the standard $(2k-1)$-spanner with $O(n^{1+1/k})$ edges. 
One reason for this limitation is that it is not clear if such trade-off is possible even when we focus on the {\em recourse}, as discussed above.
Maintaining spanners or other sparsifiers against adaptive adversaries with tight trade-offs and polylogarithmic recourse is a challenging barrier that is beyond the scope of this paper. 
Additionally, the sparse-spanner regime studied in this paper is generally the most useful for applications to other problems (see discussion in \Cref{sec:intro:app}); getting a sharper trade-off would not lead to significant improvements for most of these applications.

All the above results can be \emph{deamortized}.\footnote{To do this, we use, e.g., the sparsification technique~\cite{EppsteinGIN97} and a sophisticated dynamic expander decomposition; see \Cref{sec:overview_worstcase} for an overview.}
For example, a $2^{O(\sqrt{\log n}\log\log(n))}$-spanner
of size $\Otil(n)$ can be maintained in $2^{O(\sqrt{\log n}\log\log(n))}$
\emph{worst-case} update time. Also, for any $k$, a $2^{O(k\polylog(k))}$-cut
sparsifier of size $\Otil(n)$ can be maintained in $\tilde{O}(n^{1/k})$
worst-case time. 
In particular, we can maintain an $O(\log^{*}n)$-cut
sparsifier and an $O(1)$-cut sparsifier in $n^{o(1)}$ and $n^{\epsilon}$ time
for any constant $\epsilon$, respectively.

Our deamortization technique also implies, as a side result, the first non-trivial algorithm with \emph{worst-case} update time against an oblivious adversary for maintaining spectral sparsifiers. 

\begin{theorem}
[Oblivious Spectral Sparsifier]\label{cor:main spectral}There is
a randomized algorithm against an oblivious adversary that, given
an $n$-vertex graph undergoing edge insertions and deletions and
$\epsilon\ge1/\polylog(n)$, with high probability, maintains a $(1+\epsilon)$-spectral
sparsifier of size $n\cdot 2^{O(\sqrt{\log n})}$ using $2^{O(\log^{3/4}n)}$
worst-case update time.
\end{theorem}

\noindent The previous algorithm by Abraham et al.~\cite{AbrahamDKKP16} maintains a $(1+\epsilon)$-spectral sparsifier
of size $\Otil(n)$ using $\polylog(n)$ amortize update time. It
was asked in the same paper if the update time can be made worst-case.
\Cref{cor:main spectral} almost answers this open question
modulo $n^{o(1)}$ factors.

\subsection{Applications} \label{sec:intro:app}
Our results imply several interesting applications.
Our first set of applications are for the decremental $(1+\epsilon)$-approximate single-source shortest paths (SSSP) problem. 
There has been a line of work \cite{BernsteinC16,Bernstein17,BernsteinChechikSparse,gutenberg2020deterministic} on fast adaptive algorithms for solving this problem.
Although all these algorithms are adaptive, they share a drawback that they cannot return the shortest path itself; they can only maintain distance estimates.
Very recently, Chuzhoy and Khanna~\cite{ChuzhoyK19} showed a partial fix to this issue for some algorithms \cite{BernsteinC16,Bernstein17} and consequently obtained impressive applications to static flow algorithms. 
Unfortunately, this fix only applies in a more restricted setting, and moreover it is not clear how the technique from \cite{ChuzhoyK19} can be used to fix the same issue in  other algorithms (e.g.   \cite{BernsteinChechikSparse,gutenberg2020deterministic}).

We show that our main result from \Cref{thm:main spanner} can be employed  in a consistent and simple way, such that  the path-reporting issue in all previous  algorithms in \cite{BernsteinC16,Bernstein17,BernsteinChechikSparse,gutenberg2020deterministic} can be fixed. This resolves an open question posed in multiple papers \cite{Bernstein17,gutenberg2020deterministic,ChuzhoyK19}. We summarize these applications below:

\begin{restatable}[Fixing the path-reporting issue of \cite{BernsteinChechikSparse,gutenberg2020deterministic}]{corollary}{gutenbergExtension}\label{thm:sparseGraphSSSPWithPathQueryMainResult}
For any decremental unweighted graph $G=(V,E)$, fixed source $s$, and constant $\epsilon > 0$, there is an adaptive algorithm $\mathcal{B}$ that maintains  the $(1+\epsilon)$-approximate distances from vertex $s$ to every vertex $v \in V$ \emph{and supports corresponding shortest path queries}. The algorithm $\mathcal{B}$ has expected total update time $mn^{0.5+o(1)}$, distance estimate query time $O(\log \log n)$ and shortest path query time $\tilde{O}(n)$. 
\end{restatable}

\Cref{thm:sparseGraphSSSPWithPathQueryMainResult} gives the first adaptive algorithm without the path-reporting issue that can take $o(n^2)$ total update time.
The next algorithm works on weighted graphs and is near-optimal on dense graphs:

\begin{restatable}[Fixing the path-reporting issue of \cite{BernsteinC16,Bernstein17}]{corollary}{bernsteinExtension}\label{thm:denseGraphSSSPWithPathQueryMainResult}
For any decremental weighted graph $G=(V,E,w)$ with $W$ being the ratio between maximum and minimum edge weight, fixed source $s$, and $\epsilon > 0$, there is an adaptive algorithm $\mathcal{A}$ that maintains  the $(1+\epsilon)$-approximate distances from vertex $s$ to every vertex $v \in V$ \emph{and supports corresponding shortest path queries}. The algorithm $\mathcal{A}$ has expected total update time $\tilde{O}(n^2 \log W)$, distance estimate query time $O(\log \log (nW))$ and shortest path query time $\tilde{O}(n \log W)$. \end{restatable}

\Cref{thm:denseGraphSSSPWithPathQueryMainResult} can be compared to two previous results  \cite{ChuzhoyK19,ChuzhoyS20} 
In \cite{ChuzhoyK19}, their algorithm requires slower $n^{2+o(1)}\log W$ total update time, and needs to assume that the input graph undergoes only \emph{vertex} deletions which is more restrictive. So \Cref{thm:denseGraphSSSPWithPathQueryMainResult} strictly improves the algorithm by \cite{ChuzhoyK19}.
In \cite{ChuzhoyS20}, they use very different techniques than ours and show an algorithm with   $n^{2+o(1)}\log W$ total update time, distance query time $O(\log \log (nW))$, shortest path query time $O(|P|n^{o(1)})$ when a path $P$ is returned, and is deterministic. This algorithm is incomparable to \Cref{thm:denseGraphSSSPWithPathQueryMainResult}. Our result has slightly faster total update time, but their algorithm is deterministic and guarantees faster shortest path query time. 

By plugging \Cref{thm:denseGraphSSSPWithPathQueryMainResult} into the standard multiplicative weight update framework (e.g. \cite{GargK07,Fleischer00,ChuzhoyK19,Madry10}), we get the following:

\begin{corollary}\label{thm:intro:flow applications}
There exist $(1+\epsilon)$-approximate algorithms with expected running time $\tilde{O}(n^2)$ for the following problems: 
\begin{enumerate}[noitemsep]
    \item minimum-cost maximum $s$-$t$ flow in undirected vertex-capacitated graphs, and
    \item minimum-cost maximum $s$-$t$ flow in undirected unit-edge-capacity graphs. 
\end{enumerate}
\end{corollary}

\Cref{thm:intro:flow applications} slightly improves the results from \cite{ChuzhoyK19,ChuzhoyS20} with running time $n^{2+o(1)}$ to $\tilde{O}(n^2)$.\footnote{We note that the result of \cite{ChuzhoyS20} is deterministic.}

The second set of applications are faster algorithms for variants of multi-commodity flow problems using \Cref{thm:main spanner}.
For example, we achieve a static
$\tilde{O}((n+k)n)$-time $\polylog(n)$-approximation algorithm for the \emph{congestion minimization} problem with $k$ demand pairs
on unweighted vertex-capacitated graphs. This improves the $\tilde{O}((m+k)n)$-time
$O(\log(n)/\log\log(n))$-approximation algorithms implied by Karakostas~\cite{Karakostas08} in terms of the running time at the cost of a worse approximation ratio. 
See \Cref{sec:multiflow,sec:multiflow_proof} for more detail. 

Finally,  we apply \Cref{cor:main spectral} to the problem of  maintaining \emph{effective resistance}. Durfee et~al.~\cite{DurfeeGGP18,DurfeeGGP18-older} presented a dynamic algorithm with
$\tilde{O}(n^{6/7})$ amortized update time for $(1+\epsilon)$-approximately maintaining  the effective resistance between a fixed pair of nodes. Plugging our result in the algorithm
of Durfee et al.~leads to an $n^{6/7+o(1)}$ worst-case update time. (Both of these results assume an oblivious
adversary.)

\subsection{Techniques}

To prove the above results, the first key tool is the black-box reduction
in \Cref{thm:meta} that allows us to focus on {\em almost-uniform-degree expanders}\footnote{Expanders are graphs with high conductance (see \Cref{sec:overview}). Intuitively, they are ``robustly connected'' graphs.}.
\Cref{thm:meta} works for a large class of problems satisfying natural properties
(defined in \Cref{sec:overview}) which includes spanners, cut sparsifiers,
and spectral sparsifiers.
Hence the theorem uses the term ``$\alpha$-approximate sparsifier'' without defining the exact type of sparsifier. 
\begin{theorem}
[Informal Blackbox Reduction, 
see \Cref{def:amortized:decremental_algorithm} and \Cref{thm:amortized:fully_dynamic_weighted} for the formal statements]\label{thm:meta}Assume that there is an algorithm ${\cal A}$
that can maintain an $\alpha$-approximate sparsifier on an $n$-vertex
graph $G$ with the following promises:
\begin{itemize}[noitemsep]
\item $G$ undergoes batches of edge deletions\footnote{%
That means, in each iteration the algorithm is given a set $D \subset E$ of edges that are to be deleted.} (isolated nodes are automatically removed from $G$),
\item $G$ is unweighted, 
\item after each batch of deletions, $G$ is an expander graph, and
\item after each batch of deletions, $G$ has almost uniform degree, i.e.~the maximum degree
$\Delta_{\max}$ and the minimum degree $\Delta_{\min}$ are within a
$\polylog(n)$ factor. 
\end{itemize}
Then, there is another algorithm $\cal B$ with essentially the same amortized
update time for maintaining an $\alpha$-approximate sparsifier of essentially
the same size on a \emph{general weighted} graph undergoing both edge
insertions and deletions.
If $\cal A$ is adaptive or deterministic,
then so is $\cal B$. 
\end{theorem}

As it is well-known that many problems become
much easier on expanders (e.g.,~\cite{Trevisan05,SpielmanT04,Sherman13,JambulapatiS18,ChangS19,LSZ}), we believe that this
reduction will be useful for future developments of dynamic algorithms. For example, if one can come up with
an adaptive algorithm for maintaining $(1+\epsilon)$-cut sparsifiers on expanders,
then one can immediately obtain the same result on general graphs. 

Our second technique is a new sampling scheme called \emph{proactive resampling}:
here we constantly re-sample parts of the input graph so that, independent of an adversary's computational power, a desired structure of the underlying graph can be always maintained; see \Cref{sec:overview} for a high-level discussion of this technique. Since there are still few known tools for designing algorithms that work against an adaptive adversary, we expect that our technique will prove useful for the design of other adaptive algorithms in the future.

We further extend the black-box reduction from \Cref{thm:meta} to algorithms with worst-case update time, which allows us to deamortize both \Cref{thm:main spanner} and \Cref{thm:main cut} with slightly worse guarantees. It also easily implies \Cref{cor:main spectral}. %

\subsection{Organization}

We give a high-level overview of our approach in \Cref{sec:overview}, and list the preliminaries in \Cref{sec:preliminiaries}.

The goal of \Cref{part:reduction} is to prove the blackbox reduction summarised in \Cref{thm:meta}. 
As warm-up we start in \Cref{sec:simple_amortized_reduction} with a simpler variant that does not provide the uniform degree guarantee.
We then extend the result in \Cref{sec:uniform_degree_reduction} to further provide the uniform degree guarantee.
Both of these results yield dynamic algorithms with amortized update time.
In order to obtain worst-case update time, we must create further tools.
The first tool is presented in \Cref{sec:pruning}, where we discuss the expander pruning algorithm with worst-case update time, which  generalizes and improves the one in \cite{NanongkaiS17} for maintaining an expander under updates. 
In \Cref{sec:sparsification} we study the extension of Eppstein et al.  sparsification technique which is crucial for us to study the worst-case update time. 
Finally, we combine these tools and prove the blackbox reduction to expander graphs for worst-case update time in \Cref{sec:worst_case_blackbox}.

The goal of \Cref{part:algorithm} is to develop fast dynamic algorithms on expanders which are, thanks to  the blackbox reduction, sufficient for us to obtain dynamic algorithms for general graphs.  Specifically, in \Cref{sec:adaptive} we develop the \emph{proactive resampling} technique and present an adaptive algorithm for cut-sparsifiers in \Cref{thm:main cut}. As a cut sparsifier of an expander is also a spanner and a spectral sparsifier, we can immediately prove 
\Cref{thm:main spanner,thm:main spectral adaptive} in \Cref{sec:adaptive_more}. \Cref{secss} proves a side-result on constructing spectral sparsifiers against an oblivious adversary, which corresponds to 
\Cref{cor:main spectral}. 
 
Finally, in \Cref{part:application} we show applications to decremental shortest paths in \Cref{sec:SSSP}, dynamic effective resistance in \Cref{sec:effective_resistance}, and variants of multi-commodity flow problems in \Cref{sec:multiflow}.

\section{Overview}
\label{sec:overview}

All our algorithms use a common framework based on expanders,
which results in a reduction from fully dynamic algorithms on general graphs
to the special case of decremental algorithms on expander graphs.
The reduction holds for a general class of graph problems that satisfy some criteria. 
These criteria are satisfied for spectral-sparsifiers, cut-sparsifiers and spanners.
In this section, we define the abstract criteria needed for our reduction (See Conditions \ref{con:identity}-\ref{con:nested} below), so that we only need  to prove our algorithm  once 
and  apply it to all these types of sparsifiers.

The overview is split into three parts. In \Cref{sec:reduction_overview} we show the reduction for amortized update time. In \Cref{sec:adaptive_overview} we show how to take advantage of the reduction by designing efficient algorithms on expanders.   Finally, in \Cref{sec:overview_worstcase} we finish the overview with a sketch of how to extend the reduction to worst-case update-time algorithms.

\subsection{Reduction to Expanders: Amortized Update Time}
\label{sec:reduction_overview}

We now outline our black-box reduction, which can preserve several nice properties of the algorithms. 
That is, given an algorithm with property $x$ running on expander, 
we obtain another algorithm with property $x$ with essentially the same running time and approximation guarantee, where
the property $x$ can be ``deterministic'', 
``randomized against an adaptive adversary'', 
or ``worst-case update time''. In this subsection, we focus on amortized update time: see \Cref{sec:overview_worstcase} for an overview of how to extend the reduction to apply to worst-case algorithms.

The reduction holds for any graph problem that satisfies a small number of conditions.
We formalize a graph problem as a function $\H$ 
that maps $(G,\epsilon)$ for a graph $G$ and parameter $\epsilon>0$ to a set of graphs.
We say a dynamic algorithm $\A$ solves $\H(\epsilon)$ 
if for every input graph $G$, algorithm $\A$ maintains/computes a graph $H \in \H(G,\epsilon)$. For example we could define $\H(G,\epsilon)$ to be the set of all $(1+\epsilon)$-cut sparsifiers.
So then saying ``data structure $\A$ solves $\H(\epsilon)$'' means that $\A$ maintains for any input graph an $(1+\epsilon)$-cut sparsifier.

\paragraph{Pertubation Property}

The first property required by our reduction allows us to slightly perturb the edges, 
i.e. scale each edge $\{u,v\}$ by some small factor $f_{u,v}$ bounded by $1 \le f_{u,v} \le e^\epsilon$. Define $\zeta \cdot G$ to be the graph $G$ with all edge-weights multiplied by $\zeta$.
\begin{align}
\text{Let $G'$ be $G$ with edges scaled by up to $e^{\epsilon}$, then $G' \in \H(G,\epsilon)$ and $e^{\epsilon} \cdot G \in \H(G',\epsilon)$.}\label{con:identity}
\end{align}
Property \eqref{con:identity} implies that $G \in \H(G, \epsilon)$ for all $\epsilon > 0$. For example any graph is a (potentially dense) spectral approximation or spanner of itself. 

The property is also useful when we want to discretize the edge weights. 
A common technique is to round edge weights to the nearest power of $e^{\epsilon}$ 
in order to discretize the set of possible edge weights
without changing graph properties such as the spectrum or distances too much.
Combined with the following \emph{union} property, 
this also allows us to generalize algorithm for unweighted graphs to support weighted graphs.

\paragraph{Union Property}

Say that $G = \bigcup_{i=1}^k G_i$ for some $k$ and that $s_1,...,s_k \in \R$.
Then the union property is defined as follows:
\begin{align}
\text{If }H_i \in \H(G_i,\epsilon) \text{ and $0\le s_i$, then }
\bigcup_{i} s_i\cdot H_i \in \H\left( \bigcup_{i} s_i \cdot G_i, \epsilon \right). \label{con:union}
\end{align}
Combining this property with the previous \emph{pertubation} property \eqref{con:identity} gives us the   following reduction.
Given a graph $G$ with real edge weights from $[1,W]$,
one can decompose $G$ into graphs $G_1,...,G_k$,
such that each $G_i$ contains edges with weights in $[e^{(i-1)\epsilon}, e^{i\epsilon})$.
One can then use any algorithm that solves $\H$ on \emph{unweighted} graphs
to obtain $H_i \in \H(G'_i, \epsilon)$ for all $i=1,...,k$, 
where $G'_i$ is the graph $G_i$ when ignoring the edge weights.
Then $\bigcup_i e^{i\epsilon} \cdot H_i \in \H(\bigcup_i e^{i\epsilon}\cdot G_i, \epsilon) \subset \H(G, \epsilon)$ 
by combining property \eqref{con:union} and \eqref{con:identity}.
Thus one obtains an algorithm that solves $\H$ on \emph{weighted} graphs.

\paragraph{Reduction for Amortized Update Time (\Cref{sec:amortized:reduction})}

Loosely speaking, our black-box states the following. Say that we have a data structure $\A_X$ on a graph $X$ that at all times maintains a sparisifer in $\H(X,\epsilon)$ with amortized update time $T(\A_X)$, but assumes the following restricted setting: {\bf 1)} Every update to $X$ is an edge \emph{deletion} (no insertion), and {\bf 2)} $X$ is always an expander. We claim that $\A_X$ can be converted into a \emph{fully dynamic} algorithm $\A$ that works on any graph $G$, and has amortized update time $T(\A) = \tilde{O}(T(\A_X))$.

We first outline this black-box under the assumption that we have a dynamic algorithm 
that maintains a decomposition of $G = \bigcup_i G_i$ 
into edge disjoint expander graphs $G_1,G_2,...$.
This dynamic algorithm will have the property 
that whenever the main graph $G$ is updated by an adversarial edge insertion/deletion, each expander $G_i$ receives only edge deletions,
though occasionally a new expander $G_j$ is added to the decomposition.
Thus one can simply initialize $\A_X$ on the expander $G_j$ to obtain some $H_j \in \H(G, \epsilon)$, 
when $G_j$ is added to the decomposition.
Then whenever an edge deletion is performed to $G_j$, 
we simply update the algorithm $\A_X$ to update the graph $H_j$.
By the union property \eqref{con:union} we then have that 
$$H := \bigcup_i H_i \in \H\left(\bigcup_i G_i, \epsilon\right) = \H(G, \epsilon).$$
So we obtain an algorithm $\A$ that is able to maintain a sparsifier $H$ of $G$.
We are left with proving how to obtain this dynamic algorithm 
for maintaining the expander decomposition of $G$.

\paragraph{Dynamic Expander Decomposition (\Cref{sec:amortized:dynamic_decomposition})}

The idea is based on the expander decomposition and expander pruning of \cite{SaranurakW19}.
Their expander decomposition splits $V$ into disjoint node sets $V_1,V_2,...$, 
such that the induced subgraphs $G[V_i]$ on each $V_i$ are expanders, 
and there are only $o(m)$ edges between these expanders. 
In \Cref{sec:amortized:dynamic_decomposition} we show that by recursively applying this decomposition on the subgraph 
induced by the inter-expander edges, 
we obtain a partition of the \emph{edges} of $G$ into a union of expanders.
This means, we can decompose $G$ into subgraphs $G_1,G_2,...$, 
where each $G_i$ is an expander and $\bigcup_i G_i = G$.
We show in \Cref{sec:amortized:dynamic_decomposition} that the time complexity 
of this decomposition algorithm is $\tilde{O}(m)$. 

We now outline how we make this decomposition dynamic in \Cref{sec:amortized:dynamic_decomposition}.
Assume for now, that we have a decomposition of $G$ into $G = \bigcup_i G^{(i)}$,
where for all $i$ the graph $G^{(i)}$ has at most $2^i$ edges, 
but each $G^{(i)}$ is not necessarily an expander.
Further, each $G^{(i)}$ is decomposed into expanders $G^{(i)} = \bigcup_j G^{(i)}_j$.
To make this decomposition dynamic, we will first consider edge insertions. Every adversarial edge insertion is fed into the graph $G^{(1)}$. Now, when inserting some edges into some $G^{(i)}$, there are two cases: 
(i) the number of edges in $G^{(i)}$ remains at most $2^i$. 
In that case recompute the expander decomposition $G^{(i)} = \bigcup_j G^{(i)}_j$ of $G^{(i)}$.
Alternatively we have case (ii) where $G^{(i)}$ has more than $2^i$ edges.
In that case we set $G^{(i)}$ to be an empty graph and insert all the edges that previously belonged to $G^{(i)}$ into $G^{(i+1)}$.
Note that on average it takes $2^{i-1}$ adversarial insertions until $G^{(i)}$ is updated,
and we might have to pay $\tilde{O}(2^i)$ to recompute its decomposition,
so the amortized update time for insertions is simply $\tilde{O}(1)$.

For edge deletions we use the expander pruning technique %
based on \cite{SaranurakW19} 
(refining from \cite{NanongkaiSW17,NanongkaiS17,Wulff-Nilsen17}).
An over-simplified description of this technique is that, 
for each update to the graph, 
we can repeatedly prune (i.e. remove) 
some $\tilde{O}(1)$ edges from the graph, 
such that the remaining part is an expander.
So whenever an edge is deleted from $G$,
we remove the edge from its corresponding $G^{(i)}_j$,
and remove/prune some $\tilde{O}(1)$ many extra edges from $G^{(i)}_j$,
so that it stays an expander graph. 
These pruned edges are immediately re-inserted into $G^{(1)}$
to guarantee that we still have a valid decomposition $G = \bigcup_i G^{(i)}$.
In summary, we are able to dynamically maintain 
a decomposition $G = \bigcup_{i,j} G^{(i)}_j$ of $G$ into expander graphs.
Further, the decomposition changes only by creating new expanders 
and removing edges from existing expanders (via pruning),
so we can run the decremental expander algorithm $\A_X$ on each $G^{(i)}_j$.

\paragraph{Contraction Property and Reduction to Uniform Degree Expanders (\Cref{sec:uniform_degree_reduction})}

Many problems are easier to solve on graphs with (near) uniform degree. Thus, we strengthen our reduction to work even if the decremental algorithm $\A_X$ assumes that graph $X$ has near-uniform degree. On its own, the expander decomposition described above is only able to guarantee that for each expander the minimum degree is close to the average degree; the maximum degree could still be quite large. In order to create a (near) uniform degree expander,
we split these high degree nodes into many smaller nodes of smaller degree.
In order to perform this operation, we need the condition that whichever graph problem $\H$ we are trying to solve
must be able to handle the reverse operation, 
i.e. when we contract many small degree nodes into a single large degree node.
\begin{align}
\begin{array}{l}
\text{When contracting $W \subset V$ in both $G$ and $H\in\H(G,\epsilon)$,}\\
\text{let $G'$ and $H'$ be the resulting graphs, then $H' \in \H(G',\epsilon)$.}
\end{array} \label{con:contraction}
\end{align}

All in all, our black-box reduction shows that in order to solve a sparsification problem $\H$ in the fully dynamic model on general graphs, we need to {\bf 1)} show that $\H$ satisfies the perturbation, union, and contraction properties above (Properties \ref{con:identity}-\ref{con:contraction})
AND {\bf 2)} Design an algorithm $\A_X$ for $\H$ in the simpler setting where the dynamic updates are purely decremental (only edge deletions), and where the dynamic graph $G$ is always guaranteed to be a near-uniform degree expander.

We now present the second main contribution of our paper, which is a new adaptive algorithm $\A_X$ on expanders. We conclude the overview with a discussion of the worst-case reduction (Section \ref{sec:overview_worstcase}), for which we will need two additional properties of the problem $\H$. 

\subsection{Adaptive Algorithms on Expanders}
\label{sec:adaptive_overview}

We showed above that maintaining a sparsifier in general graphs can be reduced to the same problem in a near-uniform-degree expander. Thus, for the rest of this section we assume that $G = (V,E)$ is \emph{at all times} a $\phi$-expander with max degree $\deltamax$ and min-degree $\deltamin$, and that $G$ is only subject to edge deletions. Let $n = |V|, m = |E|$. In this overview, we assume that $\phi$ and $\deltamax/\deltamin$ are $O(\polylog n)$, and we assume $\deltamin \gg 1/\phi$. Define $\inc_G(v)$ to the edges incident to $v$ in $G$.

We now show how to maintain a $O(\log(n))$-approximate cut-sparisifier $H$ in $G$ against an adaptive adversary (see \Cref{thm:amortizedCutSparsifierMainResult} and \Cref{thm:worstCaseCutSparsifierMainResult}); it is not hard to check that $H$ is also a spanner of stretch $\Otil(1/\phi)$, because a cut-sparisifer of a $\phi$-expander is itself a $\tilde{\Omega}(\phi)$-expander, and hence has diameter $\Otil(1/\phi)$. See Section \ref{sec:adaptive} for details. 

\paragraph{Static Expander Construction}
We first show a very simple \emph{static} construction of $H \subseteq G$. Define $\rho = \tilde{\Theta}\left(\frac{\deltamax}{\deltamin^2 \phi^2}\right) = \tilde{\Theta}\left(\frac{1}{\deltamin}\right)$, with a sufficiently large polylog factor. Now, every edge is independently sampled into $H$ with probability $\rho$, and if sampled, is given weight $1/\rho$. To see that $H$ is a cut spatsifier, consider any cut $X,\bar{X}$, with $|X| \leq n/2$. We clearly have $\mathbb{E}[|E_H(X,\bar{X})|] = \rho |E_G(X,\bar{X})|$, so since every edge in $H$ has weight $1/\rho$, we have the same weight \emph{in expectation}. For a high probability bound, want to show that $\Pr[|E_H(X,\bar{X})| \sim \rho |E_G(X,\bar{X})|] \geq 1 - n^{-2|X|}$; we can then take a union bound over the $O(n^{|X|})$ cuts of size $|X|$.   

Since the graph is an expander, we know that $|E_G(X,\bar{X})| \geq \vol_G(X) \cdot \phi \geq |X| \cdot \deltamin \cdot \phi = \tilde{\Omega}(|X|\deltamin)$. Thus, by our setting of $\rho = \tilde{\Theta}(\deltamin)$, we have $\mathbb{E}[|E_H(X,\bar{X})|] \geq |X|\log^2(n)$. Since each edge is sampled independently, a chernoff bound yields the desired concentration bound for $|E_H(X,\bar{X})|$.

\paragraph{Naive Dynamic Algorithms}
The most naive dynamic algorithm is: whenever the adversary deletes edge $(u,v)$, resample all edges in $\inc_G(u)$ and $\inc_G(v)$: that is, include each such edge in $H$ with probability $\rho$. Efficiency aside, the main issue with this protocol is that the adversary can cause some target vertex $x$ to become \emph{isolated} in $H$, which clearly renders $H$ not a cut sparisifer. To see this, let $y_1, \ldots, y_k$ be the neighbors of $x$. The adversary then continually deletes arbtirary edges $(y_1,z) \neq (y_1,x)$, which has the effect of resampling edge $(x,y_1)$ each time. With very high probability, the adversary can ensure within $\log(n)$ such deletions $(y_1,z)$ that $(x,y_1)$ is NOT included in $H$; the adversary then does the same for $y_2$, then $y_3$, and so on. 

\paragraph{Slightly Less Naive Algorithm}
To fix the above issue, we effectively allow vertices $u$ and $v$ to have separate copies of edge $(u,v)$, where $u$'s copy can only be deleted if $u$ itself is resampled. Formally, every vertex $v$ will have a corresponding set of edges $S_v$ and we will always have $H = \bigcup_{v \in V} S_v$, where all edges in $H$ have weight $1/\rho$. We define an operation $\sampleb(v)$ that \emph{independently} samples each edge in $\inc_G(v)$ into $S_v$ with probability $\rho$. The naive implementation of $\sampleb(v)$ takes time $O(\deg_G(v)) = O(\deltamax)$ time, but an existing technique used in \cite{knuth1997seminumerical, devroye2006nonuniform, bringmann2012efficient} allows us to implement $\sampleb(v)$ in time $O(\rho \deltamax \log(n)) = \Otil(1)$. (The basic idea is that the sampling can be done in time proportional to the number of edges successfully chosen, rather than the number examined.)

The dynamic algorithm is as follows.  At initialization, construct each $S_v$ by calling $\sampleb(v)$, and then set $H = \bigcup_{v \in V} S_v$. Whenever the adversary deletes edge $(u,v)$, replace $S_u$ and $S_v$ with new sets $\sampleb(u)$ and $\sampleb(v)$, and modify $H = \bigcup_{v \in V} S_v$ accordingly. By the above discussion, the update time is clearly $\Otil(1)$. We now show that this algorithm effectively guarantees a good \emph{lower bound} on the weight of each cut in $H$, but might still lead to an overly high weight. Consider any cut $(X,\bar{X})$. By the expansion of $G$,
the \emph{average} vertex $x \in X$ has $\inc_G(x) \cap E_G(X,\bar{X}) \geq \phi \deltamin$. For simplicity, let us assume that \emph{every} vertex $x \in X$ has $\inc_G(x) \cap E_G(X,\bar{X}) = \tilde{\Omega}(\phi\deltamin) = \tilde{\Omega}(\deltamin)$, as we can effectively ignore the small fraction of vertices for which this is false. Now, say that an operation $\sampleb(x)$ succeeds if it results in $|S_x \cap E_G(X,\bar{X})| \sim \rho |\inc_G(x) \cap E_G(X,\bar{X})|$. Because of our setting for $\rho$ and our assumption that $\inc_G(x) \cap E_G(X,\bar{X}) = \tilde{\Omega}(\deltamin)$, a Chernoff bound guarantees that each $\sampleb(x)$ succeeds with probability $1 - n^{-10}$. Now, since the adversary makes at most $m$ updates before the graph is empty, each $\sampleb(x)$ is called at most $n^2$ times, so there is a $1 - n^{-8}$ probability that \emph{every} call $\sampleb(x)$ is successful; we call such vertices always-successful. A simple probability calculation shows that  $\Pr[$at least $|X|/2$ vertices in $X$ are always-successful$]$ $\geq  1 - n^{-2|X|}$, which allows us to union bound over all cuts of size $X$. Thus, at all times, half the vertices in $X$ have $|S_x \cap E_G(X,\bar{X})| \sim \rho |\inc_G(x) \cap E_G(X,\bar{X})|$; assuming for simplicity that this is an ``average" half of vertices, i.e. that these vertices have around half of the edges crossing the cut, we have $|E_H(X,\bar{X})| \geq |\bigcup_{x \in X} S_x \cap E_G(X,\bar{X})| \gtrsim \rho |E_G(X,\bar{X})| /2$. 

The above idea already implies that we can maintain a sparse graph $H$ where each cut is expanding, i.e.~a sparse expander, against an adaptive adversary. As an expander has low diameter, $H$ is a spanner. Therefore, we are done if our goal is a dynamic spanner algorithm.

Unfortunately, this algorithm is not strong enough for maintaining cut sparsifiers, as the algorithm may result in $|E_H(X,\bar{X})| \gg \rho |E_G(X,\bar{X})|$. Let $\deltamax \sim \sqrt{n}$, and consider the following graph $G$. There is a set $X$ of size $\sqrt{n}$ such that $G[X]$ is a clique and $G[\bar{X}]$ is $\sqrt{n}$-degree-expander. There is also a $\sqrt{n}$-to-1 matching from $X$ to $\bar{X}$: so every vertex in $y \in \bar{X}$ has \emph{exactly one} edge $e_y$ crossing the cut. It is easy to check that $G$ is an expander. The adversary then does the following. For each $y \in \bar{X}$, it keeps deleting edges in $E({y},\bar{X})$ until $e_y$ is sampled into $S_y$; with high probability, this occurs within $O(\log(n)/\rho)$ deletions for each vertex $y$. Thus, at the end, $H \supseteq \bigcup_{y \in \bar{X}} S_y$ contains all of $E_G(X,\bar{X})$.

\paragraph{Better Algorithm via Proactive Sampling}
We now show how to modify the above algorithm to ensure that with high probablity, $|E_H(X,\bar{X})| = \Otil(\rho |E_G(X,\bar{X})|)$; we later improve this to $|E_H(X,\bar{X})| = O(\rho \log(n) |E_G(X,\bar{X})|)$. We let time $t$ refer to the $t$th adversarial update. As before, we always have $H = \bigcup_{v \in V} S_v$, and if the adversary deletes edge $(u,v)$ at time $t$, the algorithm immediately calls $\sampleb(u)$ and $\sampleb(v)$. The change is that the algorithm also calls $\sampleb(u)$ and $\sampleb(v)$ at times $t+1, t+2, t+4, t+8,t+16,\ldots$; we call this \emph{proactive sampling}. The proof that $|E_H(X,\bar{X})| \gtrsim \rho |E_G(X,\bar{X})|/2$ remains basically the same as before. We now upper bound  $|E_H(X,\bar{X})|$.

The formal analysis is somewhat technical, but the crux if the following {\bf key claim:} for any $(u,v) \in G$, we have that after $t$ adversarial updates, $\Pr[(u,v) \in H \ \textrm{at time t}] \leq 2\rho\log(t) \leq 2\rho\log(m)$. We then use the key claim as follows: consider any cut $(X,\bar{X})$. If every edge in $E_G(X,\bar{X})$ was \emph{independently} sampled into $H$ with probability at most $2\rho\log(m)$, then a Chernoff bound would show that  $|E_H(X,\bar{X})| \leq 4\rho\log(m) |E_G(X,\bar{X})|$ with probability at least $1 - n^{-2|X|}$, as desired. Unfortunately, even though every individual edge-sampling occurs with probability $\rho$, independent of everything that happened before, it is NOT the case that event $e \in H$ is independent from event $e' \in H$: the adversary is adaptive, so its sampling strategy for $e'$ can depend on whether or not $e$ was successfully sampled into $H$ at an earlier time. Nonetheless, we show in the full proof that these dependencies can be disentangled.

Let us now sketch the proof for the key claim. The edge $(u,v)$ can appear in $H$ because it is in $S_u$ or $S_v$ at the time $t$. Let us bound the probability that $u \in S_u$ at time $t$. Let $\tschedule(u)$ be all times before $t$ for which $\sampleb(u)$ has been scheduled by proactive sampling: so whenever the adversary updates an edge $(u,v)$ at time $t'$, times $t', t'+1,t'+2,t'+4,t'+8,\ldots$ are added to $\tschedule(u)$. Let $\tschedule^{t'}(u) \subset [t',t]$ be the state of $\tschedule(u)$ at time $t'$. Now, we say that a call to $\sampleb(u)$ at time $t'$ is \emph{relevant} if $\tschedule^{t'+1}(u) = \emptyset$. Observe that for $(u,v)$ to be in $S_u$ at time $t$, it must have been added during some relevant call $\sampleb(u)$, because every non-relevant call is followed by another call before time $t$ which invokes again $\sampleb(u)$ and thereby deletes the previously sampled set $S_u$ and replaces it by a new one. We complete the proof by claiming that there are at most $\log(t)$ relevant calls $\sampleb(u)$. This is because if a relevant call occurs at $t'$, then proactive sampling adds some time $t^*$ to $\tschedule(u)$ such that $(t' + t) / 2 \leq t^* \leq t$; thus, there can be no relevant calls in time interval $[t',(t' + t)/2]$. So each relevant call halves the possible time interval for other relevant calls, so there are at most $\log(t)$ relevant calls.

We now briefly point out why this modified algorithm achieves $|E_H(X,\bar{X})| = \Otil(\rho |E_G(X,\bar{X})|)$, rather than the desired  $|E_H(X,\bar{X})| = O(\rho \log(n) |E_G(X,\bar{X})|)$. Consider again the graph $G$ consisting of a vertex set $X$ of size $\sqrt{n}$ such that $G[X]$ is a complete graph, and let $\overline{X}$ be a $\sqrt{n}$-degree expander in $G[\overline{X}]$. Additonally, we have a $\sqrt{n}$-to-one matching, i.e. every vertex in $X$ is matched to $\sqrt{n}$ vertices in $\overline{X}$. The graph is still an expander as argued before. 

Now observe that $\deltamax = 2\sqrt{n}$ and we obtain a first sparsifier $H$ at time $0$ of $G$ where we have weight on the cut, i.e. $|E_H(X,\bar{X})|/\rho$, of size $\sim n$ (which is the number of edges crossing). In particular, the weight on edges in $|E_H(X,\bar{X}) \cap \bigcup_{x \in \overline{X}} S_x|/\rho \sim n$, i.e. the vertices in $\overline{X}$ carry half the weight of the cut in the sparsifier. But over the course of the algorithm, the adversary can delete edges in the cut $(X, \overline{X})$ that are in $G \setminus H$. Observe that the resampling events do not affect edges in $E_H(X,\bar{X}) \cap \bigcup_{x \in \overline{X}} S_x$ since none of the deleted edges is incident to any such edge (recall the $\sqrt{n}$-to-one matching). The adversary can continue until the cut only has weight in $G$ of $|X|\deltamin \phi$ without violating the expander and min-degree guarantees. But then the weight in $H$ on the cut is still $\sim n$ while the weight in $G$ is only $\sim n (\deltamax/\deltamin) \phi$. Thus, we only obtain a $\sim (\deltamax/\deltamin) \phi$-approximation (plus a $\log n$-factor from proactive sampling might appear).

\paragraph{Final Algorithm}
To resolve the issue above, we would like to ensure that the edges in $E_H(X,\bar{X})$ are resampled when $|E_G(X,\bar{X})|$ changes by a large amount. We achieve this with one last modification to the algorithm: for every $v \in V$, whenever $\deg_G(v)$ decreases by $\zeta = \phi \deltamin$, we run $\sampleb(w)$ for every edge $(v,w) \in G$. It is not hard to check that each vertex will only resampled a total of $O(\deltamax^2 / \zeta) = \Otil(\deltamax)$ additional times as a result of this change which is subsumed by $\tilde{O}(m)$ when summing over the vertices. (A naive implementation of the above modification only leads to small \emph{amortized} update time, but this can easily be worst-case by staggering the work over several updates using round-robin scheduling.) We leave the analysis of the approximation ratio for Section \ref{sec:adaptive}.

By using the convenient lemma which says that any cut sparsifier on an expander is also a spectral spectral (see \Cref{cor: cut on expander is spectral}), we also obtain an adaptive algorithm for spectral sparsifier.

\subsection{Reduction to Expanders: Worst-Case Update Time}
\label{sec:overview_worstcase}

We now outline how to extend our black-box reduction to work with worst-case update time.
We again assume there exists some algorithm $\A_X$ that maintains for any graph $G$ a sparsifier $H \in \H(G, \epsilon)$,
provided that $G$ stays a uniform degree expander throughout all updates,
all of which are only edge deletions.

The condition that $G$ always remains an expander is too strong, but we can use expander pruning to maintain the property from the perspective of $\A_X$. Consider some deletion in $G$: although $G$ may not be an expander, we can us pruning to find a subset of edges $P \subset E(G)$, such that $G \setminus P$ is an expander. We then input all edges in $P$ as deletions to $\A_X$, so the graph $G \setminus P$ in question is still an expander: $\A_X$ thus returns $H \in \H(G\setminus P, \epsilon)$.
Then based on property \eqref{con:identity} and \eqref{con:union} it can be shown
that $H \cup P \in \H(G, \epsilon)$.
So by taking all the pruned edges together with the sparsifier $H$ of $G \setminus P$
we obtain a sparsifier of $G$, even when $G$ itself is no longer an expander.

Unfortunately this dynamic sparsifier algorithm has two downsides: 
(i) The maintained sparsifier is only sparse for a short sequence of updates, 
as otherwise the set of pruned edges becomes too large 
and thus the output $H \cup P$ becomes too dense. 
(ii) The algorithm only works on graphs that are initially an expander.

\paragraph{Extending the algorithm to general graphs}

To extend the previous algorithm to work on general graphs, 
we run the static expansion decomposition algorithm (\Cref{sec:uniform:decomposition}).
As outlined before, we can decompose $G$ into subgraphs $G_1,G_2,...$, 
where each $G_i$ is an expander and $\bigcup_i G_i = G$.
We can then run the algorithm, outlined in the previous paragraph, on each of these expanders 
and the union of all the obtained sparsifiers will be a sparsifier of the original input $G$.

Similar as before, one downside of this technique is 
that the size of the sparsifier will increase with each update, 
because more and more edges will be pruned. 
Thus, the resulting dynamic algorithm can only maintain a sparsifier 
for some limited number of updates.

\paragraph{Extending the number of updates}

A common technique for dynamic algorithms, 
which only work for some limited number of updates (say $k$ updates),
is to reset the algorithm after $k$ updates. 
If the algorithm has preprocessing time $p$ and update time $u$, 
then one can obtain an algorithm with amortized update time $O(p/k + u)$. 
However, the worst-case complexity would be quite bad,
because once the reset is performed, 
the old sparsifier (from before the reset) must be replaced by the new one.
Listing all edges of the new sparsifier within a single update would be too slow.
There is a standard technique for converting such an amortized bound
to an equivalent worst-case bound.
The idea is to slowly translate from the old sparsifier to the new one,
by only listing few edges in each update.
For this we require another property for $\H$ that guarantees 
that the sparsifier stays valid, even when removing a few of its edges.

\paragraph{Transition Property}
Consider some $H_1, H_2 \in \H(G, \epsilon)$,
and we now want to have a slow transition from $H_1$ to $H_2$,
by slowly removing edges from $H_1$ from the output (and slowly inserting edges of $H_2$).
The exact property we require is as follows:
\begin{align}
\text{Let $H_1,H_2 \in \H(G, \epsilon)$ and $H \subset H_1$, then $(e^{\delta}-1) H \cup H_2 \in \H(G, \epsilon + \delta)$.} \label{con:increment}
\end{align}
Here $H \subset H_1$ represents the remaining to be removed edges 
(or alternatively $H_1 \setminus H$ are the remaining to be inserted edges).
Exploiting this property we are able to obtain a $O(p/k+u)$ worst-case update time.

As the output grows with each update, we must perform the reset after $k = O(n)$ updates,
otherwise the output becomes too dense.
This is unfortunate as the preprocessing time is $p=\Omega(m)$, 
because one must read the entire input, 
which is too slow to obtain a subpolynomial update time.
This issue can be fixed via a sparsification technique based on \cite{EppsteinGIN97}, 
presented in \Cref{sec:sparsification}.
By using this technique, 
we can make sure that $m = O(n^{1+o(1)})$ 
and thus the preprocessing time will be fast enough 
to allow for $O(n^{o(1)})$ update time.

For this sparsification technique we require the following transitivity property.

\paragraph{Transitivity Property}

\begin{align}
\text{If }H \in \H(G, \epsilon) \text{, then }
\H\left( H, \delta \right) 
\subset 
\H\left( G, \delta+\epsilon \right). \label{con:nested}
\end{align}
Intuitively this means that an approximation $H$ of $G$ and an approximation $H'$ of $H$,
then $H'$ is also a (slightly worse) approximation of $G$.

Properties \ref{con:identity}-\ref{con:nested} above are precisely the properties required of a graph prolem by our black-box reduction for worst-case update time. We show in Lemmas \ref{lma:cutSparsifierIsAGraphProblem},
\ref{lma:spannerIsAGraphProblem}, and
\ref{lem:spectralSparsifierIsAGraphProblem} that the sparsifiers discussed in this paper (spectral sparsifier, cut sparsifier, spanner) satisfy all these properties.

\paragraph{Sparsification Technique (\Cref{sec:sparsification})}

We now outline the sparsification technique, 
whose formal proof is presented in \Cref{sec:sparsification}. 
Let $G$ be an arbitrary graph. 
We partition the edges of $G$ into equally sized subgraphs $G_1, G_2,...,G_d$ for some $d > 1$.
Note that, if we have $\varepsilon$-approximate 
sparsifiers $H_1,...,H_d$ of $G_1,...,G_d$,
then $\bigcup_i H_i$ is a $\varepsilon$-approximate sparsifier of $G$ by property \eqref{con:union}.
In addition, if we have a $\varepsilon$-approximate sparsifier $H$
of $\bigcup_i H_i$, then $H$ is a $(2\varepsilon)$-approximate sparsifier of $G$
by property \eqref{con:nested}.
This allows us to obtain a faster algorithm as follows:
If $G$ has $m$ edges, then each $G_i$ has only $m/d$ edges,
so the dynamic algorithm runs faster on these sparse $G_i$.
Further, since the $H_i$ are sparse (let's say $O(n)$ edges),
the graph $\bigcup_i H_i$ has only $O(dn)$ edges
and maintaining $H$ is also faster than maintaining a sparsifier of $G$ directly, if $dn = o(m)$.
The next idea is to repeat this trick recursively:
We repeatedly split each $G_i$ into $d = n^{o(1)}$ graphs,
until the graphs have only $O(n^{1+o(1)})$ edges.
This means we obtain some tree-like structure rooted at $G$,
where each tree-node $G'$ represents a subgraph of $G$ and 
its tree-children are the $d$ subgraphs $G'_1,...,G'_d$ of $G'$.
For the graphs that form leaves of this tree, 
we run our dynamic sparsifier algorithm.
We also obtain a sparsifier $H'$ of any non-leaf tree-node $G'$, 
by running our dynamic algorithm on $\bigcup_{i=1}^d H'_i$,
where the $H'_i$ are sparsifiers of the child-tree-nodes $G'_i$ of the tree-node $G'$.
Thus all instances of our dynamic algorithm always run on sparse input graphs.
However, there is one downside:
When some sparsifier $H'_i$ changes, the sparsifier $H'$ must also change.
Let's say some edge is deleted from $G$, 
then the edge is deleted from one leaf-node of the tree-structure,
and this update will propagate from the leaf-node all the way to the root of the tree.
This can be problematic because when the dynamic algorithm changes some $c > 1$ many edges of the sparsifier for each edge update, 
then the number of updates grows exponentially with the depth of the tree-like structure.
In \cite{EppsteinGIN97} Eppstein et al.~circumvented this issue by assuming an extra property which they call \emph{stability property}, 
which essentially says that this exponential growth does not occur.
Our modified sparsification technique no longer requires this assumption, 
instead we balance the parameter $d$ carefully 
to make sure the blow-up of the propagation is only some sub-polynomial factor.

\section{Preliminaries}
\label{sec:preliminiaries}

\paragraph{Graphs and Edge sets.} Given a graph $G$, we will often identify $G$ with the edge set. 
For example we write $(u,v) \in G$ to express that the edge $(u,v)$ exists in $G$.
Sometimes we will also use set operations, for example let $G, H$ be two graphs and $W$ be a set of edges.
Then $G \cup H$ refers to the graph $(V(G) \cup V(H), E(G) \cup E(H))$,
and $G \cup W$ refers to the graph induced by $E(G) \cup W$.
For simplicity we also write $G \cup e$ for a single edge $e$ instead of $G \cup \{e\}$.

We have $G - W = (V(G), E(G) \setminus W)$,
and $G \setminus W$ is the graph induced by $E(G) \setminus W$, 
i.e. the graph $G - W$ but without the isolated vertices.
Likewise, $G\setminus H$ is the graph induced by $E(G) \setminus E(H)$. 
We write $H \subset G$ for a subgraph $H$ of $G$.
We denote by $\textsc{Inc}_G(v)$ the set of incident edges to vertex $v \in V$ in $G$.
For a scalar $s \in \R_{>0}$ we define $s \cdot G$ as the graph $G$ where all edge weights were scaled by $s$.

\paragraph{Spanners.}

Given a graph $G$, we call a subgraph $S \subset G$ an $s$-stretch spanner, 
if for every two nodes $u,v \in V$ we have 
$\text{dist}_S(u,v) \le s \cdot \text{dist}_G(u,v)$, 
so the distance between every pair in $S$ is increased 
by a factor of at most $s$ compared to $G$.
The inequality in the other direction 
$\text{dist}_G(u,v) \le \text{dist}_S(u,v)$ holds automatically, 
since $S$ is a subgraph of $G$ and we do not allow negative edge weights.

\paragraph{Spectral sparsifier.}

Given a graph $G = (V,E)$ with edge weights $w \in \R^{|E|}$,
the Laplacian matrix $L_G$ of $G$ is an $n\times n$ matrix, where $(L_G)_{u,u}$ is the degree of node $u$, $(L_G)_{u,v} = -w(u,v)$ if $u$ and $v$ are connected by an edge, and $(L_G)_{u,v}=0$ otherwise. Then,  an $e^\epsilon$-spectral
sparsifier $H \subset G$ 
is a re-weighted subgraph that satisfies for all vectors $\vec{x} \in \R^n$ that
$\vec{x}^\top L_G \vec{x} \le \vec{x}^\top L_H \vec{x} \le e^\epsilon \vec{x}^\top L_G \vec{x}$.\footnote{%
Typically one analyzes $(1+\epsilon)$-spectral sparsifiers. 
For small $\epsilon \ge 0$ we have $1+\epsilon \approx e^{\epsilon}$ 
so we will use the latter term instead. 
This has the benefit that $e^{\epsilon}\cdot e^{\epsilon} = e^{2\epsilon}$, 
i.e. the $e^\epsilon$-approximation of an $e^\epsilon$-approximation is a $e^{2\epsilon}$-approximation. 
In comparison, $(1+\epsilon)^2 = 1+2\epsilon+\epsilon^2$ 
so it is harder to analyze what happens 
if we have repeated approximations of approximations.}

\paragraph{Weights, volume, and degree.}
We write $w_G(u,v)$ for the weight of edge $(u,v)$ in graph $G$ 
and $\vol_G(U) = \sum_{u\in U,v\in V, (u,v) \in E} w_G(u,v)$ is the volume of $U$ in $G$. 
The degree of node $v$ in $G$ is denoted by $\deg_G(v)$.

\paragraph{Subgraphs.}
Let 
$G[S]$ be the subgraph of $G$, when restricting the node set to $S \subset V$.
We denote   $G[A,B]$ by the graph on nodes $A \cup B$ and only including  the edges of $G$ 
that have one endpoint in $A$ and one endpoint in $B$.
We use 
$G\{S\}$ to denote the subgraph of $G$, 
when restricting the node set to $S \subset V$, 
but additionally we add self loops to every node, such that $\deg_{G\{S\}}(v) = \deg_G(v)$.

\paragraph{Cuts and conductance.}

  We write $\delta(S)$ for the number of edges crossing $S$ to $V \setminus S$,
so $\delta(S)$ is the total weight of the cut $(S, V \setminus S)$ that we also sometime denote by $(S, \overline{S})$. %
We define the conductance of a set $S$ by
\[
\phi_G(S) = \frac{\delta(S)}{ \min \{ \vol_G(S), \vol_G(V\setminus S) \}}.
\]
We write $\Phi_G$ for the conductance of $G$, where $\Phi_G = \min_{\emptyset \neq S \subset V} \phi_G(S)$.
We call a graph $G$ with conductance $\phi$ a $\phi$-expander.

Further, for any two (not necessarily disjoint) sets $X, Y \subseteq V$, we denote by $E_G(X,Y)$ the set of edges in $E$ with one endpoint in $X$ and another endpoint in $Y$, and by $E(X)$ the set of edges incident to a vertex in $X$, i.e. $E(X) = \bigcup_{v \in X} \textsc{Inc}_G(v)$.

\paragraph{Recourse.}

Given a dynamic algorithm that maintains some sparsifier $H$ of an input graph $G$,
the recourse is the number of edges that the algorithm changes in $H$ per changed edge in $G$.

\subsection*{Preliminaries about Expanders}

Our framework is heavily based on expanders and decomposition of graphs,
so we first list  some helpful lemmas.
The first one is a decomposition of any graph $G$ into expanders, 
with few edges   between the expanders.

\begin{theorem}[{\cite{SaranurakW19}}]\label{lem:decomposition}
There is a randomized algorithm that
with high probability
given a graph $G = (V,E)$
with $m$ edges
and a parameter $\phi$,
partitions $V$ into $V_1,...,V_k$
in time $O(m \log^4 m/\phi)$ 
such that $\min_{1\leq i\leq k} \Phi_{G[V_i]} \ge \phi$
and $\sum_{i=1}^k \delta(V_i) = O( \phi m \log^3 m)$.
In fact, the algorithm has a stronger guarantee that $\min_i \Phi_{G\{V_i\}} \ge \phi$.
\end{theorem}

The following lemma shows that it is easy to construct an expander with high conductance and almost uniform degree, and the proof of the lemma can be found in \Cref{sec:omit}.

\begin{lemma}
[Fast explicit construction of expanders]
Given numbers $n\ge10$ and $d\ge9$, there
is an algorithm with running time $O(nd)$ that constructs a graph
$H_{n,d}$ with $n$ vertices such that each vertex has degree between
$d-8$ and $2d$, and conductance $\Phi_{H_{n,d}}=\Omega(1)$.
\label{prop:explicit expander}
\end{lemma}

As outlined in \Cref{sec:adaptive_overview},
we require our expanders to be of almost uniform degree in order to efficiently maintain our sparsifiers.
The following algorithm allows us to transform a graph to be of almost uniform degree.
The exact properties of this algorithm are stated in \Cref{prop:delta-reduction prop}.

\begin{definition}
[$\Delta$-reduction]\label{def:delta-reduction}Let $G=(V,E)$ be
a graph, and let $\Delta\ge9$ be a parameter. The $\Delta$-reduction
graph $G'$ of $G$ is obtained from $G$ by the following operations: 
\begin{itemize}
\item For each node $u\in V$ with $\deg(u)\ge10\Delta$, we replace $u$
by the expander $X_{u}=H_{\left\lceil \deg(u)/\Delta\right\rceil ,\Delta}$
defined in \Cref{prop:explicit expander} with $\left\lceil \deg(u)/\Delta\right\rceil $
nodes and $\Theta(\Delta)$ degree. We call $X_{u}$ a \emph{super-node}
in $G'$.
\item Let $E_{u}$ be set of edges in $G$ incident to $u$. Let $E_{u}^{(1)},\dots,E_{u}^{(\left\lceil \deg(u)/\Delta\right\rceil )}$
be a partition of $E_{u}$ into groups of size at most $\Delta$. 
\item For each $(u,v)\in E_{u}^{(i)}\cap E_{v}^{(j)}$, we  add an
edge between the $i$-th of node $X_{u}$ and the $j$-th node of
$X_{v}$.
\end{itemize}
\end{definition}

\begin{lemma} 
\label{prop:delta-reduction prop}For any graph $G=(V,E)$, the $\Delta$-reduction
graph $G'=(V',E')$ of $G$ has the following properties:
\begin{enumerate}
\item $G'$ can be obtained from $G$ in $O(|E|)$ time;
\item Each node $u$ in $G'$ is such that $\deg_{G'}(u)=\Theta(\deg_{X_{u}}(u))=\Theta(\Delta)$;
\item Each super-node $X_{u}$ in $G'$ has volume $\vol(X_{u})=\Theta(\deg_{G}(u))$;
\item $|V|\le|V'|\le 2|V|+|E|/\Delta$, and
\item $\Phi_{G'}=\Theta(\Phi_{G})$.
\end{enumerate}
\end{lemma}

The proof of the lemma above can be found in \Cref{sec:omit}.

\part{Reduction to Expanders \label{part:reduction}}
\label{define_properties} %

In this part, we show how to reduce a large class of  dynamic sparsifier problems to the same problem on expanders with several additional nice properties.

As a warm-up, in \Cref{sec:simple_amortized_reduction}, we show the simplest reduction to expanders that works for an algorithm with amortized update time. This reduction is not strong enough for our main result.
In \Cref{sec:uniform_degree_reduction}, we strengthen the previous reduction so that we can assume that the expander has almost uniform degree at all time. 
This amortized reduction is strong enough and will be needed to prove \Cref{thm:main spanner,thm:main cut,thm:main spectral adaptive} in \Cref{part:algorithm}.

Then, our next goal is to extend the reduction further to work with algorithms with worst-case update time.
We need several tools for doing this. \Cref{sec:pruning} discusses the expander pruning algorithm with \emph{worst-case} update time, which generalizes and improves the one in \cite{NanongkaiSW17}.
\Cref{sec:sparsification} discusses the extention of the sparsification technique by Eppstein et al. \cite{EppsteinGIN97} which allows the algorithm to assume that the input expander is already quite sparse and the algorithm's task is just to sparsify it even further. This is crucial for the worst-case time reduction.
Lastly, in \Cref{sec:worst_case_blackbox}, we complete the reduction for algorithms with worst-case update time. This blackbox allows us to obtain the deamortized versions of \Cref{thm:main spanner,thm:main cut,thm:main spectral adaptive} and also easily obtain \Cref{cor:main spectral} in \Cref{part:algorithm}.

The reduction holds for any graph problem that satisfies a small number of conditions.
For ease of discussion, let us recall the properties defined earlier: 
for any graph $G$ and accuracy parameter $\epsilon$, 
we view a graph problem as a function $\H$ that maps $(G,\epsilon)$ to a set of graphs. 
We say a dynamic algorithm $\A$ solves $\H(\epsilon)$,  
if for every input graph $G$  algorithm $\A$ maintains/computes a graph $H \in \H(G,\epsilon)$. 
For example, we could define $\H(G,\epsilon)$ to be the set of all $\epsilon$-approximate spectral sparsifiers of $G$. 
Then, we say $\A$ solves $\H(\epsilon)$ if, for any input graph $G$, $\A$ returns an $\epsilon$-approximate spectral sparsifier of $G$. 
For parameters $0\leq \epsilon,\delta$, we will study problems satisfying the  following properties:

\newcounter{backupequation}%
\setcounter{backupequation}{\value{equation}}
\setcounter{equation}{0}
\begin{align}
&\text{Let $G'$ be $G$ with edges scaled by up to $e^{\epsilon}$, then $G' \in \H(G,\epsilon)$ and $e^{\epsilon} \cdot G \in \H(G',\epsilon)$.}\\
&\text{If }H_i \in \H(G_i,\epsilon) \text{ and $0\le s_i$, then }
\bigcup_{i} s_i\cdot H_i \in \H\left( \bigcup_{i} s_i \cdot G_i, \epsilon \right).\\
&\begin{array}{l}
\text{When contracting $W \subset V$ in both $G$ and $H\in\H(G,\epsilon)$,}\\
\text{let $G'$ and $H'$ be the resulting graphs, then $H' \in \H(G',\epsilon)$.}
\end{array} \\ 
&\text{Let $H_1,H_2 \in \H(G, \epsilon)$ and $H \subset H_1$, then $(e^{\delta}-1) H \cup H_2 \in \H(G, \epsilon + \delta)$.}\\
&\text{If }H \in \H(G, \epsilon) \text{, then }
\H\left( H, \delta \right) 
\subset 
\H\left( G, \delta+\epsilon \right).
\end{align}
\setcounter{equation}{\value{backupequation}}

The warm-up amortized reduction from \Cref{sec:simple_amortized_reduction} only needs Property \eqref{con:identity} and \eqref{con:union}. To be able to assume almost uniform degree, the reduction in \Cref{sec:uniform_degree_reduction} additionally need Property \eqref{con:contraction}.
Then, our reduction in \Cref{sec:worst_case_blackbox} for worst-case algorithms needs all five properties.

\section{Warm-up: Amortized Reduction to Arbitrary Expanders}
\label{sec:simple_amortized_reduction}

Let $\H$ be some graph problem that satisfies the perturbation property \eqref{con:identity}
and union property \eqref{con:union}.
In this section we present a reduction that reduces the fully dynamic problem $\H$ on general graphs 
to the decremental variant of $\H$ (i.e.~only edge deletions) when the input graph is an expander.
Note that this reduction is not strong enough for our main result, but it gives a key idea.
Formally we define such a decremental algorithm as follows:

\begin{definition}
\label{def:simple:decremental_algorithm}
Let $\H$ be a graph problem.
We call an algorithm $\mathcal{A}$ a ``decremental algorithm on $\phi$-expanders for $\H(\epsilon)$'', 
if for any \emph{unweighted} $n$-node graph $G$
\begin{itemize}
\item $\mathcal{A}$ maintains $H \in \mathcal{H}(G,\epsilon)$ under edge deletions to $G$. 
In each update the algorithm is given a set $D \subset E(G)$ of edges that are to be removed from $G$.
\item The algorithm is allowed to assume that before and after every update the graph $G$ is a $\phi$-expander.
\end{itemize}
\end{definition}

We show that any decremental algorithm on $\phi$-expanders for $\H(\epsilon)$ 
results in a fully dynamic algorithm for general graphs.

\begin{theorem}[Amortized Blackbox Reduction (without Uniform Degree Promise)]\label{thm:simple:fully_dynamic_weighted}
Assume $\H$ satisfies 
\eqref{con:identity} and \eqref{con:union}, 
and there exists a decremental algorithm $\A$ for $\H(\epsilon)$
on $\phi$-expanders for any $\phi = O(1/\log^4 n)$.
Then, there exists a fully dynamic algorithm $\B$ 
for $\H(\epsilon)$ on general \emph{weighted} graphs
whose ratio between the largest and the smallest edge weight is bounded by $W$.

For the time complexity, let $P(m)\ge m$ be the pre-processing time of $\A$, 
$S(n)\ge n$ an upper bound on the size of the output graph throughout all updates, 
and $T(n)$ the amortized update time of $\A$ per deleted edge.
Then, the pre-processing time of $\B$ is $O(P(m))$,
with amortized update time $$
O\left(
	\phi^{-3} \log^7 n \cdot T(n) 
	+ \frac{P(m)}{m \phi^{3}} \log^{7} n
\right).
$$
The size of the output graph is bounded by $O(S(n \log^3 n) (1+\epsilon^{-1}) \log W)$.
\end{theorem}

In rest of this section is for proving \Cref{thm:simple:fully_dynamic_weighted}.
The key tool for proving this is the dynamic expander decomposition from \Cref{thm:dynamicExpanderDecomposition}.
Then, we use \Cref{thm:dynamicExpanderDecomposition} to  prove \Cref{thm:simple:fully_dynamic_weighted} in \Cref{sec:amortized:reduction}.

\subsection{Dynamic Expander Decomposition}
\label{sec:amortized:dynamic_decomposition}

It is known that, given any graph, we can partition the edge set of the graph into parts such that each part induces an expander. 
The theorem below shows that we can quickly maintain this partitioning of edges in a dynamic graph:

\begin{theorem}[Dynamic Expander Decomposition]
\label{thm:dynamicExpanderDecomposition}
For any $\phi = O(1/\log^4 m)$ there exists a dynamic algorithm against an adaptive adversary, 
that preprocesses an unweighted graph $G$ 
in $O(\phi^{-1} m \log^5 n)$ time. 
The algorithm can then maintain a decomposition of $G$ 
into $\phi$-expanders $G_1,...G_t$, 
supporting both edge deletions and insertions in $O(\phi^{-2} \log^6 n)$ amortized time.
The graphs $(G_i)_{1 \le i \le t}$ are edge disjoint
and we have $\sum_{i=1}^t |V(G_i)| = O(n \log^2 n)$.

After each update, the output consists of a list of changes to the decomposition.
The changes consist of (i) edge deletions to some graphs $G_i$, 
(ii) removing some graphs $G_i$ from the decomposition, 
and (iii) new graphs $G_i$ are added to the decomposition.
\end{theorem}

We first formally show how to obtain the static version of the algorithm in \Cref{thm:dynamicExpanderDecomposition}. 
This is obtained by simply repeatedly applying the algorithm from \Cref{lem:decomposition} by \cite{SaranurakW19}:

\begin{lemma}\label{lem:recursiveDecomposition}
There is a randomised algorithm that, given a graph $G=(V,E)$ with $n$ vertices, $m$ edges and parameter $\phi = o\left(1/(\log m)^3\right)$, the algorithm finishes in $O\left( (1/\phi)\cdot m\log^5 n \right)$ time and constructs a sequence of  graphs $G_1,\ldots, G_t$ such that the following holds: 
\begin{itemize}
\item $\sum_{i=1}^t |V_i|$ = $O(n \log m)$;
\item $\bigcup_{i=1}^t E[G_i] = E[G]$;
\item $E[G_i] \cap E[G_j] = \emptyset$ for all $i,j$;
\item each $G_i$ is a $\phi$-expander. 
\end{itemize}
\end{lemma}

\begin{proof}
Let $G^{(1)}= \left( V^{(1)}, E^{(1)} \right)$ be the initial input graph $G$ with $m$ edges, and we construct the graph sequence $G^{(2)}, \ldots, G^{(s)}$ inductively as follows: Given any $G^{(\ell)}$ for $\ell\geq 1$, we apply \Cref{lem:decomposition} to decompose $G^{(\ell)}$ into $V_{\ell, 1},\ldots, V_{\ell, k_{\ell}}$, which satisfies
\[
\sum_{j=1}^{k_{\ell}} \delta\left({V_{\ell, j}}\right)=  O\left(\phi\cdot |E(G^{(\ell)})|\cdot \log^3 |E(G^{(\ell)})| \right),
\]
i.e., the total weight of the edges that do not belong to $G[V_{\ell, 1}],\ldots, G[V_{\ell, k_{\ell}}]$ is $O\left(\phi m \log^3 m \right)$. Let $E^{(\ell+1)}$ be the set of such edges, and $G^{(\ell+1)}$ be the graph induced by the edge set $E^{(\ell+1)}$. To bound the total number of iterations $s$ needed for this process, notice that, as long as $\phi = o\left(1/(\log |E(G^{(\ell)})|)^3\right)$, the total weight of the edges inside each $G^{(\ell)}$ satisfies
\[
\sum_{j=1}^{k}\delta_{G^{(\ell)}}(V_{\ell,j}) =O\left(\phi\cdot|E(G^{(\ell)})| \log^3 |E(G^{(\ell)})| \right) = o\left(|E(G^{(\ell)})|\right),
\]
i.e., comparing with $|E(G^{(\ell)})|$ the total weight of the edges in $G^{(\ell+1)}$ decreases by at least a constant factor, which implies that after $s=O(\log m)$ iterations the resulting graph will be an empty one.

Now let the set of subgraphs $G[V_{\ell,j_{\ell}}]$ for  $1\leq \ell \leq s, 1\leq j_{\ell}\leq k_{\ell}$ be the final output graphs of the algorithms, and let $t$ be the number of such subgraphs. Since the total number of vertices of the subgraphs in each iteration is $O(n)$ and there are at most $O(\log m)$ iterations, we have that 
\[
\sum_{i=1}^t |V_i| = O(n\log m),
\]
and the runtime of the algorithm is at most $O\left( (1/\phi)\cdot m\log^5 n \right)$. The remaining statements of the lemma  follow by \Cref{lem:decomposition}. 
\end{proof}

We want to turn this static algorithm into a dynamic one.
For that we require the expander pruning technique from \cite{SaranurakW19},
which can be interpreted as a decremental expander algorithm,
i.e. it maintains an expander when the graph receives edge deletions.

\begin{theorem}
[Amortized Expander Pruning, \cite{SaranurakW19}]\label{thm:amortized prune}There
is a deterministic algorithm, that, given an access to the adjacency
list of a connected multi-graph $G=(V,E)$ with $m$ edges, a sequence of
$\sigma=(e_{1},\dots,e_{k})$ of $k\le\phi m/10$ online edge deletions
and a parameter $\phi>0$, maintains a vertex set $P\subseteq V$
with the following properties. Let $G_{i}$ be the graph $G$ after
the edges $e_{1},\dots,e_{i}$ have been deleted from it; let $P_{0}=\emptyset$
be the set $P$ at the beginning of the algorithm, and for all $0<i\le k$,
let $P_{i}$ be the set $P$ after the deletion of $e_{1},\dots,e_{i}$.
Then, the following statements hold for all $1\le i\le k$:
\begin{itemize}
\item $P_{i-1}\subseteq P_{i}$,
\item $\vol(P_{i})\le8i/\phi$,
\item $|E(P_{i},V \setminus P_{i})|\le4i$, and
\item If $G$ is a $\phi$-expander, then $G_{i}[V-P_{i}]$ a $\phi/6$-expander.
\end{itemize}
The total running time for the algorithm is $O((k\log m)/\phi^{2})$.
\end{theorem}

By combining \Cref{lem:recursiveDecomposition} and \Cref{thm:amortized prune} we obtain the following
fully dynamic expander decomposition algorithm.
That is, a dynamic algorithm that maintains an edge disjoint decomposition $\bigcup_i G_i = G$,
where each $G_i$ as an expander.

\subsubsection*{Proof of \Cref{thm:dynamicExpanderDecomposition}.}

We start by describing the algorithm, then we analyze complexity and correctness.

\paragraph{Algorithm}

Assume graph $G$ is decomposed into $G_1, G_2, ...$,
where each $G_i$ contains at most $2^i$ edges.
During the preprocessing this is obtained by just setting $G_i := G$ for $i = \log m$
and setting all other $G_j$ to be empty graphs.

The idea is to maintain for each $G_i$ an expander decomposition $\bigcup_j G_{i,j} = G_i$,
where each $G_j$ is a $\phi$-expander. 
Then $\bigcup_{i,j} G_{i,j} = G$ is an expander decomposition of $G$.
So for now fix some $G_i$ and we describe how to maintain the decomposition of $G_i$.

If some set of edges $I$ is to be inserted into $G_i$, then we consider two cases:
(i) If $|E(G_i) \cup I| > 2^i$, 
then we set $G_i$ to be an empty graph 
and insert $E(G_i) \cup I$ into $G_{i+1}$.
(ii) If on the other hand $|E(G_i) \cup I| \le 2^i$,
then we perform the expander decomposition of \Cref{lem:recursiveDecomposition}
on $G_i$ to obtain a decomposition of $G_i$ into $\phi$-expanders, 
so $\bigcup_j G_{i,j} = G_i$ where $G_{i,j}$ is a $\phi$-expander for all $j$. 
We also initialize the expander pruning algorithm of \Cref{thm:amortized prune}
on $G_{i,j}$ for all $j$.

If some edge is to be deleted from $G_i$,
then the pruning algorithm of \Cref{thm:amortized prune} is notified.
We then prune (i.e. delete) some nodes from $G_i$ according to the dynamic pruning algorithm.
All edges that are removed from $G_i$ because of these node deletions
are reinserted into $G_1$ and handled like edge insertions (described in the previous paragraph).

Finally, consider edge deletions and insertions to $G$.
When an edge is inserted into $G$, then we insert the edge into $G_1$.
If an edge is deleted from $G$, then it is removed from the $G_i$ that contained the edge.

\paragraph{Complexity}

We first analyze edge insertions.
When inserting some set of edges $I$ into some $G_i$, and $|E(G_i) \cup I| \le 2^i$,
then the cost is
$$
O(
\underbrace{
	\phi^{-1} 2^i \log^5 n
}_{\text{\Cref{lem:recursiveDecomposition}}}
+
\underbrace{
	2^i
}_{\text{\Cref{thm:amortized prune}}}
).
$$
Note that when inserting an edge into $G$ (or $G_1$),
then it takes $2^{i-1}$ edge insertions until an edge is inserted into $G_i$.
Thus the amortized cost per edge insertion is $O(\phi^{-2} \log^6 n)$.

When deleting an edge from some $G_i$, the amortized cost is 
$O(\phi^{-1} \log n)$ by the pruning algorithm \Cref{thm:amortized prune}.
On average we remove $O(\phi^{-1})$ edges from $G_i$ because of pruned nodes.
These nodes are inserted into $G_1$ again,
and by the previous complexity bound for edge insertions,
an edge deletion thus has amortized cost $O(\phi^{-2} \log^6 n)$.

\paragraph{Correctness}

Each $G_{i,j}$ is a $\phi$-expander by the guarantee of the expander decomposition~(\Cref{lem:recursiveDecomposition})
and the expander pruning~(\Cref{thm:amortized prune}).
As all pruned edges are reinserted into $G_1$,
no edge gets lost and we always have $\bigcup_{i,j} G_{i,j} = G$.
Further note that for every $i$ we have $\sum_j |V(G_{i,j})| = O(|V(G_i)| \log^2 n)$ by \Cref{lem:recursiveDecomposition}.
Thus we have $\sum_{i,j} |V(G_{i,j})| = O(|V(G)| \log^3 n)$.

\subsection{Reduction via Dynamic Expander Decomposition}
\label{sec:amortized:reduction}

To prove \Cref{thm:simple:fully_dynamic_weighted}, we first prove the version of the theorem in which the input graph is unweighted. 
Such a result can be proven by a direct application of the dynamic expander decomposition \Cref{thm:dynamicExpanderDecomposition}. 
Remember that in that setting it suffices to study the case in which the algorithm only receives edge deletions.  

\begin{lemma}\label{lem:simple:fully_dynamic}
Assume $\H$ satisfies 
\eqref{con:union},
and there exists a decremental algorithm 
$\A$ for $\H(\epsilon)$
on $\phi$-expanders for any $\phi = O(1/\log^4 n)$. 
Then, there exists a fully dynamic algorithm $\B$ 
for $\H(\epsilon)$ on general unweighted graphs.

For the time complexity, let $P(m)\ge m$ be the pre-processing time of $\A$,
$S(n)\ge n$ an upper bound on the size of the output graph throughout all updates,
and $T(n)$ the amortized update time of $\A$. 
Then, the pre-processing time of $\B$ is $O(P(m))$, 
with amortized update time 
$$
O\left(
	\phi^{-3} \log^7 n \cdot T(n) 
	+ \frac{P(m)}{m \phi^{3}} \log^{7} n
\right).
$$
The size of the output graph is bounded by $O(S(n \log^3 n))$.

If $\A$ is a query-algorithm, then so is $\B$.
\end{lemma}

\begin{proof}
We start by describing the algorithm, 
then analyze the complexity and prove the correctness.

\paragraph{Algorithm}

We run the dynamic expander decomposition of \Cref{thm:dynamicExpanderDecomposition}.
The output is a decomposition $G = \bigcup_i G_i$
where each $G_i$ is a $\phi$-expander.
Next, we run algorithm $\A$ on $G_i$ to obtain $H_i \in \H(G_i, \epsilon)$ and
we maintain $H := \bigcup_i H_i$.
This graph $H$ is the maintained output of our dynamic algorithm.

All these graphs are maintained dynamically,
so when an edge is deleted from $G$,
we pass that edge deletion to the dynamic expander decomposition algorithm \Cref{thm:dynamicExpanderDecomposition}.
The decomposition might change in the following way:
(i) Some $G_i$ is removed from the decomposition. 
In that case also remove $H_i$ from $H$.
(ii) Some new $G_i$ is added to the decomposition. 
Initialize $\A$ on $G_i$ to construct new $H_i$, and add $H_i$ to the output graph $H$.
(iii) Some edge is removed from $G_i$. 
We pass the edge deletion to algorithm $\A$ to update $H_i$. 
The graph $H$ is updated accordingly to the changes performed to $H_i$.

\paragraph{Complexity}

The preprocessing costs $
\sum_i P(|E(G'_i)|) \le P(\sum_i |E(G'_i)|) \le P(O(m))$,
where we used $P(K) \ge K$ to move the sum inside and $\sum_i |E(G'_i)| = O(m)$ 
as guaranteed by \Cref{thm:dynamicExpanderDecomposition}.

During each update,
at most $O(\phi^{-2} \log^6 n)$ edges change in the decomposition per update (amortized).
So updating the instances of $\A$ costs $O((\phi^{-1} \log^6 n) \cdot T(n))$ amortized time per update,
where we used the fact that each $G_i$  has  at most $n$ nodes.
The amortized cost of initializing $\A$ on some new $G_i$ is bounded by
$$
O\left(\frac{P(|E(G'_i)|)}{|E(G'_i)| \phi^{2} \log^{-6} n}\right)
\le
O\left(\frac{P(m)}{m \phi^{2}} \log^{6} n\right),
$$
where we used  the fact that a $G_i$ with $k$ edges is created after at least $\phi^{2} \log^{-6} k$ updates,
and that $P(k) \ge k$.
The total amortized update time is thus
$$
O\left(\phi^{-2} \log^6 n \cdot T(n) + \frac{P(m)}{m \phi^{2}} \log^{6} n\right).
$$
The size of the output graph is bounded by
$$
\sum_i S(|V(G'_i)|) 
\le
S(\sum_i |V(G'_i)|) 
\le
O(S(n \log^2 n)),
$$
because $\sum_i |V(G'_i)| = \sum_i O(|V(G_i)|) = O(n \log^2 n)$ 
by \Cref{thm:dynamicExpanderDecomposition}.
\paragraph{Correctness}

Each $G_i$ is always a $\phi$-expander by \Cref{thm:dynamicExpanderDecomposition}, 
and therefore we have $H_i \in \H(G_i, \epsilon)$ by guarantee of $\A$.
The algorithm always maintains $H = \bigcup_i H_i$ where each $H_i \in \H(G_i, \epsilon)$,
so $H \in \H(G, \epsilon)$ by $\bigcup_i G_i = G$ and property \eqref{con:union}.

\end{proof}

The extension to weighted graphs now follows directly 
from splitting the edges based on their edge weights 
into powers of $e^\epsilon$.

\begin{proof}[Proof of \Cref{thm:simple:fully_dynamic_weighted}] 
Split the edges into groups $[e^{k\cdot \epsilon/2}, e^{(k+1)\epsilon/2})$ 
to obtain subgraphs $G_0,G_1,...G_t$ with $t = O((1+\epsilon^{-1}) \log W)$.
Then run \Cref{lem:simple:fully_dynamic} on each subgraph to maintain $H_i \in \H(G_i, \epsilon/2)$
and the union of these graphs then satisfies $H := \bigcup_{i=0}^t e^{(i+1) \epsilon/2} \cdot H_i \in \H(G, \epsilon)$
by property \eqref{con:identity} and \eqref{con:union}.
\end{proof}

\section{Amortized Reduction to Almost Uniform Degree Expanders}
\label{sec:uniform_degree_reduction}

In this section, we strengthen the reduction from \Cref{sec:simple_amortized_reduction} so that we can assume that the expander has almost unifrom degree. 
This is already strong enough to give our main results with amortized update time. Specifically, it will be used to prove \Cref{thm:main spanner,thm:main cut,thm:main spectral adaptive} in \Cref{part:algorithm}.

Let $\H$ be some graph problem that satisfies the perturbation property \eqref{con:identity},
union property \eqref{con:union}, and contraction property \eqref{con:contraction}.
Formally we define such a decremental algorithm as follows:

\begin{definition}
\label{def:amortized:decremental_algorithm}
Let $\H$ be a graph problem.
We call an algorithm $\mathcal{A}$ a ``decremental algorithm on almost-uniform-degree $\phi$-expanders for $\H(\epsilon)$'', 
if for any \emph{unweighted} $n$-node graph $G$
\begin{itemize}
\item $\mathcal{A}$ maintains $H \in \mathcal{H}(G,\epsilon)$ under edge deletions to $G$. 
In each update the algorithm is given a set $D \subset E(G)$ of edges that are to be removed from $G$.
\item The algorithm is allowed to assume that before and after every update the graph $G$ 
every node has degree between $\Omega(\phi \Delta)$ and $O(\phi^{-1}\Delta)$, 
where $\Delta$ is the minimum degree of $G$ during the initialization.
\item The algorithm is allowed to assume that before and after every update the graph $G$ is a $\phi$-expander.
\end{itemize}
When algorithm $\A$ does not directly maintain $H$, 
but it instead has a \textsc{Query}-operation which returns the graph $H$,
then we call $\A$ a \emph{query}-algorithm.
\end{definition}

The main result of this section is as follows:
\begin{theorem}[Amortized Blackbox Reduction with Uniform Degree Promise]
\label{thm:amortized:fully_dynamic_weighted}
Assume $\H$ satisfies 
\eqref{con:identity}, \eqref{con:union} and \eqref{con:contraction}, 
and there exists a decremental algorithm $\A$ for $\H(\epsilon)$
on almost-uniform-degree $\phi$-expanders for any $\phi = O(1/\log^4 n)$. 
Then, there exists a fully dynamic algorithm $\B$ 
for $\H(\epsilon)$ on general \emph{weighted} graphs
whose ratio between the largest and the smallest edge weight is bounded by $W$.

For the time complexity, let $P(m)\ge m$ be the pre-processing time of $\A$, 
$S(n)\ge n$ an upper bound on the size of the output graph throughout all updates, 
and $T(n)$ the amortized update time of $\A$ per deleted edge.
Then, the pre-processing time of $\B$ is $O(P(m))$,
with amortized update time $$
O\left(
	\phi^{-3} \log^7 n \cdot T(n) 
	+ \frac{P(m)}{m \phi^{3}} \log^{7} n
\right).
$$
The size of the output graph is bounded by $O(S(n \log^3 n) (1+\epsilon^{-1}) \log W)$.

If $\A$ is a query-algorithm with query complexity $Q(n)$, 
then so is $\B$ with query time $O( Q(n \log^3 n) \log W)$.
\end{theorem}

The outline for how to prove \Cref{thm:amortized:fully_dynamic_weighted} is similar to how we prove \Cref{thm:simple:fully_dynamic_weighted} in \Cref{sec:simple_amortized_reduction}. 
That is, we need a dynamic expander decomposition similar to \Cref{thm:dynamicExpanderDecomposition}, but now with an additional guarantee that each expander should have almost uniform degree. 
With this additional constraint, we cannot immediately apply  the known tools as in the proof of  \Cref{thm:dynamicExpanderDecomposition} anymore. 
We proceed as follows. First, in \Cref{sec:uniform:decomposition} we show a static algorithm for partitioning edges of a graph into expanders with the uniform degree constraint. Second, in \Cref{sec:pruning_uniform} we extend the expander pruning algorithm \cite{SaranurakW19} to incorporate the uniform degree constraint as well. 
Then, we combine both tools and obtain the dynamic expander decomposition with uniform degree constraint in \Cref{sec:uniform:dynamic_decomposition}.
Given this, it is straightforward to prove \Cref{thm:amortized:fully_dynamic_weighted} in the same way that we did before for \Cref{thm:simple:fully_dynamic_weighted}.

\subsection{Static Expander Decomposition with Uniform Degree Constraint}
\label{sec:uniform:decomposition}

In this section, we prove the following:

\begin{theorem}[Static Uniform Degree Expander Decomposition]
\label{thm:uniformDegreeExpanderDecomposition}
For any $\phi = O(1/\log^4 m)$ we can construct for any graph $G$, 
in $O(\phi^{-1} m \log^6 m)$ time a decomposition $\bigcup_i G_i = G$,
where each $G_i$ is a $\phi$-expander and $\sum_i |V(G_i)| \le O(n \log^2 m)$.
We further obtain a second set of $\phi$-expanders $G'_1,...,G'_t$, 
and sets $X_{v,i} \subset V(G_i)$ for $v \in V(G_i)$
with the following properties: 
\begin{itemize}
\item For all $i$, we obtain $G_i$ from $G'_i$ when contracting each $X_{v,i}$ for $v \in V(G_i)$.
\item For all $i$ and $u \in V(G'_i)$ we have $\deg_{G'_i}(u) = \Theta(\deg_{X_v}(u))$.
\item $|V(G'_i)| = O(|V(G_i)|)$, $\vol(X_{v,i}) = \Theta(\deg_{G_i}(v))$ and $\sum_i |V(G_i)| = O(n \log^2 m)$.
\item $\max\deg(G_i) = O(\phi^{-1} \min\deg(G_i))$
\end{itemize}

\end{theorem}

To prove \Cref{thm:uniformDegreeExpanderDecomposition}, we proceed as follows.
First, let us   prove 
a more involved variant of \Cref{lem:recursiveDecomposition}. Specifically,   \Cref{alg:minDegreeExpanderDecomposition} 
and \Cref{lem:degreeConstrainedExpanderDecomposition} show that an input graph $G$ can be decomposed into expanders $G_1,\ldots, G_t$, each of which has its average degree close to the minimum degree.
Based on this, we can split these high degree nodes into many smaller degree ones,
such   that the minimum and maximum degrees are close to each other. 
This will finally lead to   \Cref{thm:uniformDegreeExpanderDecomposition}.

\begin{algorithm}[th]
\caption{\label{alg:minDegreeExpanderDecomposition}
Let $G=(V,E)$ be an $n$-node graph and $\phi = o(1/\log^3 n)$ be some parameter.
Further define $\Delta$ to be the maximum degree of $G$, $k = 0$ and $G^0 := G$.}
We perform the following procedure:

\begin{enumerate}
\item Terminate if $G^k$ is an empty graph.
Otherwise set $s \leftarrow 0$ and let $G^{k+1}$ be an empty graph on the node set $V$.

\item Repeatedly pick a node $v$ with $\deg_{G^k}(v) < \Delta/2^k$
and remove it with its adjacent edges from the graph $G^k$.
The removed edges are added to $G^{k+1}$.
\label{step:degree}

\item The graph $G^k$ now has  only nodes of degree at least $\Delta/2^k$ left.
(Or $G^k$ is an empty graph, in which case we set $k \leftarrow k+1$ and go back to step 1.)

\item We apply the expander decomposition of \Cref{lem:decomposition} to $G^k_s := G^k$ and
obtain a partition of $V$ into $V_{1},...,V_{t}$ for some $t \ge 1$.
Let $G^k_{s,i} := G^k_s[V_{i}]$ for $i=1,...,t$, set $s \leftarrow s + 1$,
and remove the edges of $G^k_s[V_{1}],...,G^k_s[V_{t}]$ from $G^k$.
\label{step:decomposition}

\item Go back to step \ref{step:degree} \emph{without} increasing $k$.
\end{enumerate}

\end{algorithm}

\begin{lemma}\label{lem:degreeConstrainedExpanderDecomposition}
Let $(G^k_{s,i})_{0 \le k,s,i}$ be the subgraphs constructed in step
\ref{step:decomposition} of \Cref{alg:minDegreeExpanderDecomposition}.
Then $\dot\bigcup_{k,s,i} E[G^k_{s,i}] = E[G]$,
so $(G^k_{s,i})_{0 \le k,s,i}$ is a decomposition of $G$. Moreover,   \Cref{alg:minDegreeExpanderDecomposition} terminates after
$O(m \phi^{-1} \log^6 m)$ time and the largest value for $s$ is $O(\log m)$.
In addition, each  $G^k_{s,i}$ has the following properties:
\begin{itemize}
\item $G^k_{s,i}$ is a $\phi$-expander.
\item $|E[G^k_{s,i}]| \le 2 \Delta / 2^{k-1} |V[G^k_{s,i}]|$
\item $\min \deg G^k_{s,i} \ge \phi \Delta /2^k $
\item $\sum_{k,s,i} |V(G^k_{s,i})| \le O(n \log^2 m)$.
\end{itemize}

\end{lemma}

\begin{proof}

Every edge we remove from $G^k$ in step \ref{step:degree} is added to $G^{k+1}$, and every edge we remove from $G^k$ in step \ref{step:decomposition} is part of some $G^k_{s,i}$.
For $k > \log \Delta$ no more edges are removed from $G^k$ in step \ref{step:degree},
hence every edge must be included in some $G^k_{s,i}$, so $\dot\bigcup_{k,s,i} E[G^k_{s,i}] = E[G]$.

We first analyse the runtime of the algorithm. 
We will prove that the time until incrementing $k$ 
is bounded by $O(m \phi^{-1} \log^5 m)$.
The variable $k$ will be incremented up to 
$O(\log \Delta) = O(\log m)$ times,
until no more nodes can be removed in step \ref{step:degree}.
Thus the total time is $O(m \phi^{-1} \log^6 m)$.

Removing the nodes in step \ref{step:degree} 
and moving the edges to $G^{k+1}$
takes at most $O(m)$ time.
The decomposition of \Cref{lem:decomposition} 
requires $O(m \phi^{-1} \log^4 m)$ time.
Note that only $O(  \phi m \log^3 m )$ edges 
are not deleted from $G^k$,
so for $\phi = o(1/\log^3 m)$,
we jump at most $O(\log m)$ times back to step \ref{step:degree},
until all edges in $G^k$ are removed.
This has two implications:
(i) we spend a total of $O(m \phi^{-1} \log^5 m)$ time 
between two increments of $k$.
(ii) Step \ref{step:decomposition} is executed 
at most $O(\log m)$ times, 
so the largest value for $s$ is bounded by $O(\log m)$.

It remains to prove the three statements of the lemma. 
Our first statement follows directly from  \Cref{lem:decomposition}. 
To prove the second statement,  
we consider the graph $G^k$ before we delete any of its edges in step \ref{step:degree}.
All edges of this graph were added in step \ref{step:degree},
when the value of the variable $k$ was $k-1$.
Because of this, every edge in $G^k$ can be assigned to a node,
such that every node has at most $\Delta/2^{k-1}$ assigned edges:
When some node $v$ is removed from $G^{k-1}$ and its adjacent edges
are added to $G^k$ in step \ref{step:degree},
then we assign these edges to node $v$.
Thus every node has at most $\Delta / 2^{k-1}$ assigned edges
and thus each $G^k_{s,i}$ can have at most $\Delta / 2^{k-1} |V(G^k_{s,i})|$ many edges.

For the third statement, 
let $V^k_{s,i} := V[G^k_{s,i}]$ and $u \in V^k_{s,i}$ for any $s,i$,
then $\delta_{G^k_s\{V^k_{s,i}\}}(\{v\})
= \deg_{G^k_{s,i}}(v)$,
because $G^k_{s,i}$ is obtained by decomposing $G^k_s$ via \Cref{lem:decomposition}.
Further,
$\vol_{G^k_s\{V^k_{s,i}\}}(\{v\})
= \deg_{G^k_s}(v)$.
Since $G^k_{s,i}$ is a $\phi$ expander we have
$\phi \le \delta_{G^k_s\{V^k_{s,i}\}}(\{v\}) / \vol_{G^k_s\{V^k_{s,i}\}}(\{v\})$,
and thus $\deg_{G^k_{s,i}}(v) \ge \phi \deg_{G^k_s}(v) \ge \phi \Delta / 2^k $.

Finally, for the last statement observe that 
for fixed $k,s$ the graphs $G^k_{s,1},G^k_{s,2},...$ are vertex disjoint,
so $\sum_i |V(G^k_{s,i})| \le n$.
As both $k$ and $s$ are bounded by $O(\log m)$,
we have $\sum_{k,s,i} |V(G^k_{s,i})| = O(n \log^2 m)$.
\end{proof}

\begin{corollary}\label{cor:averageDegreeExpanderDecomposition}
For any $\phi = O(1/\log^4 m)$ we can construct for any graph $G$
in $O(\phi^{-1} m \log^6 m)$ time a decomposition $\bigcup_i G_i = G$,
where each $G_i$ is a $\phi$-expander,
$\min\deg(G_i) \ge \phi |E(G_i)| / (4|V(G_i)|)$,
and $\sum_i |V(G_i)| = O(n \log^2 m)$.
\end{corollary}

\begin{proof}
Direct implication of \Cref{lem:degreeConstrainedExpanderDecomposition}
when re-indexing the sequence $(G^k_{s,i})_{k,s,i,\ge 0}$ to $(G_j)_{j \ge 0}$.
\end{proof}

\subsubsection*{Proof of \Cref{thm:uniformDegreeExpanderDecomposition}.}

We start by describing the algorithm and then analyze complexity and correctness.
\paragraph{Algorithm}
Let $G$ be the input graph.
We run the expander decomposition of 
\Cref{cor:averageDegreeExpanderDecomposition}
to obtain the decomposition $\bigcup_i G_i = G$.

Next we want to create a graph $G'_i$ for each $G_i$,
so fix some $i$.
Then we run \Cref{prop:delta-reduction prop} for $\Delta = \phi^{-1} \min\deg(G_i)$
to obtain $G'_i$ and the node sets $X_{v,i}$.

\paragraph{Complexity}

Constructing the $G_i$ takes $O(\phi^{-1}m \log^6 m)$ time
by \Cref{cor:averageDegreeExpanderDecomposition}.
Constructing one $G'_i$ takes $O(|E(G_i)|)$ time by \Cref{prop:delta-reduction prop},
so constructing all $G'_i$ takes $O(m)$ total time.

\paragraph{Correctness}

By \Cref{cor:averageDegreeExpanderDecomposition}
each $G_i$ is a $\phi$-expander, $\bigcup_i G_i = G$,
and $\sum_i V(G_i) \le O(|V(G)| \log^2 m)$.
By \Cref{prop:delta-reduction prop}
the $G'_i$ are $\Theta(\phi)$-expanders.
We can make sure that all $G'_i$ are $\phi$-expanders,
by simply choosing decreasing $\phi$ by a constant factor 
when running the algorithm of \Cref{cor:averageDegreeExpanderDecomposition}
and \Cref{prop:delta-reduction prop}.

We have $V(G'_i) \le 2|V(G_i)| + |E(G_i)|/(\phi^{-1}\min\deg(G_i)) \le O(|V(G_i)|)$
by \Cref{prop:delta-reduction prop} and \Cref{cor:averageDegreeExpanderDecomposition},
and further 
$E(G'_i) = \frac{1}{2}\sum_{v \in V(G_i)} \vol_{G'_i}(X_{v,i}) = O(\sum_{v \in V(G_i)} \deg_{G_i}(v)) = O(|E(G_i)|)$
and each node in $G'_i$ has degree $O(\Delta) = O(\phi^{-1} \min\deg(G_i))$,
so $\max\deg(G'_i) \le O(\phi^{-1} \min\deg(G'_i))$.

The same claims are true for $G'_i = G_i$ where $G_i$ has less than $9$ nodes,
as the number of nodes and edges is bounded by $O(1)$.

\subsection{Expander Pruning with Uniform Degree Constraint}
\label{sec:pruning_uniform}

Here, we show how to extend \Cref{thm:amortized prune},
so it also guarantees that the minimum degree of the expander 
does not decrease too much.
The idea is to first prune out low degree nodes
and then perform the expander pruning to maintain an expander.
By expansion property we are able to show that the second pruning does not decrease the degree too much.

\begin{lemma}[Uniform Degree Expander Pruning]
\label{lem:uniformDegreeNodePruning}
Let $G = (V,E)$ be a $\phi$-expander with $m$ edges and minimum degree $\Delta$.
There is a deterministic algorithm with access to adjacency lists of $G$,
such that given an online sequence of edge deletions in $G$,
can maintain a \emph{pruned set} $P \subset V$
such that the following property holds.
Let $G_i,P_i$ be the graph and set $P$ after $i$ edge deletions
and $\Delta$ be the minimum degree of $G_0$, i.e. the initial graph before any deletions.
We have for all $i$
\begin{enumerate}
\item $P_0 = \emptyset$, $P_{i} \subset P_{i+1}$
\item $\vol(P_i) \le 30 i/\phi$
\item $G_i[V \setminus P_i]$ is a $\phi/6$ expander or an empty graph.
\item $\min\deg G_i[V \setminus P_i] \ge \phi\Delta/18 $.
\end{enumerate}
The total time for updating $P_0,...,P_k$ is $O(k \phi^{-2} \log m)$.
\end{lemma}

\begin{proof}
The algorithm works as follows.
Let $G_i$ be the graph $G$ after the $i$th edge deletion.
After each such edge deletion,
we also repeatedly remove all nodes with degree less than $\Delta/3$
until all nodes have degree at least $\Delta/3$.
Let $G'_i$ be the resulting graph.
We then inform the pruning algorithm \Cref{thm:amortized prune} 
of all the deleted edges
and obtain a set of pruned nodes $P_i$.
Instead of $P_i$, the actual output of our algorithm will be $Q_i$, 
which is the union of $P_i$ and all nodes that we removed because of too small degrees.

For any node $v$ consider the cut when cutting $v$,
then by expansion guarantee
$$
\frac{\deg_{G'_i[V \setminus P_i]}(v)}
{\min\{\vol_{G'_i\{V \setminus P_i\}}(v), \vol_{G'_i\{V \setminus P_i\}}(V \setminus v)\}} 
\ge 
\phi/6.
$$
Here $G'_i\{V \setminus P_i\}$ is the graph $G'_i[V \setminus P_i]$, 
except that we add self-loops to each node 
such that $\deg_{G'_i}(v) = \deg_{G'_i\{V \setminus P_i\}}(v)$.
Thus we have
$\deg_{G'_i[V \setminus P_i]}(v) \ge \Delta \phi / 18$.

Note that a node is only removed after at least $\frac{2}{3}\Delta$ edge deletions occured.
And the removal of a node is equivalent to $\Delta/3$ edge deletions.
Thus we can bound 
$$
\vol(Q_i) 
\le 
10i/\phi + 10i/\phi \sum_{k \ge 0} \left(\frac{\Delta/3}{\frac{2}{3}\Delta}\right)^k
=
10i/\phi + 10i/\phi \sum_{k \ge 0} \left(\frac{1}{2}\right)^k
=
10i/\phi + 20i/\phi
=
30i/\phi
$$

Finally, consider the complexity of the algorithm.
Per real edge deletion to $G$, 
we have on average $2$ further edge deletions because of pruning low degree nodes.
Hence the complexity of the pruning algorithm increases only by a constant factor, 
which is hidden in the $O(\cdot)$ notation.
\end{proof}

\subsection{Dynamic Expander Decomposition with Uniform Degree Constraint}
\label{sec:uniform:dynamic_decomposition}

By combining the expander decomposition in \Cref{thm:uniformDegreeExpanderDecomposition} and the expander pruning algorithm in \Cref{lem:uniformDegreeNodePruning} with the uniform degree constraint, we prove the variant of the dynamic expander decomposition~(\Cref{thm:dynamicExpanderDecomposition}),
where we obtain expanders of near uniform degree.

\begin{theorem}[Dynamic Uniform Degree Expander Decomposition]
\label{thm:dynamicExpanderDecompositionUniformDegree}
For any $\phi = o(1/\log^3 m)$ there exists a dynamic algorithm against an adaptive adversary, 
that preprocesses an unweighted graph $G$ 
in $O(\phi^{-1} m \log^5 n)$ time. 
The algorithm maintains graphs 
$G_1,...,G_t$, 
$G'_1,...,G'_t$, 
and sets $X_{v,i} \subset V(G)$ for all $i\in[t],v \in V(G)$ with the following properties:
\begin{itemize}
\item $\bigcup_i G_i = G$ is a decomposition of $G$.
\item $G_i$ is obtained from graph $G'_i$, when contracting each $X_{v,i}$ for $v \in V(G_i)$. 
\item Each $G_i$ and $G'_i$ is a $\phi$-expander.
\item We have $\min\deg(G'_i) \le \max\deg(G'_i) \le O(\phi^{-2}) \min\deg(G'_i)$, 
and throughout all updates the minimum degree of each $G_i$ can decrease by at most a factor of $O(\phi)$.
\item $|V(G'_i)| \le O(V(G'_i))$ and $\sum_i |V(G_i)| = O(n \log^3 m)$.
\item $\sum_i |E(G'_i)| = O(m)$.
\end{itemize}
The algorithm supports both edge deletions and insertions in $O(\phi^{-3} \log^7 n)$ amortized time.

After each update, the output consists of a list of changes to the decomposition.
The changes consist of (i) edge deletions to some $G_i$ and $G'_i$, 
(ii) removing some graphs $G_i$ and $G'_i$ from the decomposition, 
and (iii) new $G_i$, $G'_i$ are added to the decomposition.
\end{theorem}

\begin{proof}
The algorithm is similar to \Cref{thm:dynamicExpanderDecomposition}.
The main difference is that we run the uniform degree expander decomposition of \Cref{thm:uniformDegreeExpanderDecomposition}
instead of \Cref{lem:recursiveDecomposition}.

\paragraph{Algorithm}

Assume graph $G$ is decomposed into $G_1, G_2, ...$,
where each $G_i$ contains at most $2^i$ edges.
During the preprocessing this is obtained by just setting $G_i := G$ for $i = \log m$
and setting all other $G_j$ to be empty graphs.

The idea is to maintain for each $G_i$ 
the uniform degree expander decomposition $\bigcup_j G_{i,j} = G_i$
and the graphs $G'_{i,j}$.
So for now fix some $G_i$ and we describe how to maintain the decomposition of $G_i$.

If some set of edges $I$ is to be inserted into $G_i$, then we consider two cases:
(i) If $|E(G_i) \cup I| > 2^i$, 
then we set $G_i$ and $G'_i$ to be empty graphs
and we set $X_{v,i} = \emptyset$ for all $v \in V$. 
Then we insert $E(G_i) \cup I$ into $G_{i+1}$.
(ii) If on the other hand $|E(G_i) \cup I| \le 2^i$,
then we perform the expander decomposition of \Cref{thm:uniformDegreeExpanderDecomposition}
on $G_i \cup I$ to obtain a decomposition $G_i = \bigcup_j G_{i,j}$, $\phi$-expanders $G'_{i,j}$
and the node sets $X_{v,i,j}$ for $v \in V(G_{i,j})$.
We also initialize the expander pruning algorithm of \Cref{lem:uniformDegreeNodePruning}
on $G'_{i,j}$ for all $j$.

If some edge $\{u,v\}$ is to be deleted from $G_i$,
then that edge also exists in some $G_{i,j}$
and there is an edge $\{u', v'\}$ in $G'_{i,j}$ with $u' \in X_{u,i}$ and $v' \in X_{v,i}$.
We remove these edges from $G_i$, $G_{i,j}$ and $G'_{i,j}$
and the pruning algorithm of \Cref{lem:uniformDegreeNodePruning} 
is notified of the edge deletion performed to $G'_{i,j}$.
We then prune (i.e. delete) some nodes from $G'_{i,j}$ according to the dynamic pruning algorithm.
All edges that are removed from $G'_{i,j}$ because of these node deletions 
might correspond to some edges from the contracted graph $G_{i,j}$.
Remove these edges from $G_{i,j}$ as well and 
reinsert them into $G_1$ where they are handled like edge insertions (described in the previous paragraph).

If the number of edges if some $G_{i,j}$ decreased by $\phi K$, 
where $K$ is the number of edges in $G_{i,j}$ when it was created,
then we remove $G_{i,j}$ from the decomposition of $G_i$.
And reinsert the edges of $G_{i,j}$ into $G_1$.

Finally, consider edge deletions and insertions to $G$.
When an edge is inserted into $G$, then we insert the edge into $G_1$.
If an edge is deleted from $G$, then it is removed from the $G_i$ that contained the edge.

\paragraph{Complexity}

The complexity can be analyzed as in \Cref{thm:dynamicExpanderDecomposition}.
The first difference is an extra $\log m$ factor,
because \Cref{thm:uniformDegreeExpanderDecomposition} is slower than \Cref{lem:recursiveDecomposition} by a $\log m$ factor.
The second difference is that we remove and reinsert any $G_{i,j}$ once $\phi K$ edges were removed,
where $K$ is the number of edges in $G_{i,j}$ when it was created.
As per edge deletion to $G$, at most $O(\phi^{-1})$ edges are removed from $G_{i,j}$ via pruning,
we know it takes at least $\Omega(K \phi^2)$ updates until we remove and reinsert $G_{i,j}$.
As the complexity of inserting an edge is $O(\phi^{-1} \log^7 n)$,
and we must insert $O(K)$ edges, we obtain a total update complexity of at most $O(\phi^{-3} \log^7 n)$.

\paragraph{Correctness}

Each $G'_{i,j}$ is a $\phi$-expander by guarantee of the expander decomposition \Cref{thm:uniformDegreeExpanderDecomposition}
and the expander pruning \Cref{lem:uniformDegreeNodePruning}.
(Technically they are $\phi/6$-expanders, but we can just choose a smaller $\phi$.)
As $G_{i,j}$ is obtained from contracting $G'_{i,j}$, the graph $G_{i,j}$ is also a $\phi$-expander.

As all pruned edges are reinserted into $G_1$,
no edge gets lost and we always have $\bigcup_{i,j} G_{i,j} = G$.

By \Cref{thm:uniformDegreeExpanderDecomposition} and \Cref{lem:uniformDegreeNodePruning} 
we always have $\max\deg G'_{i,j} \le O(\phi^{-2} \min\deg G'_{i,j})$
and the minimum degree can drop by at most a factor of $O(\phi)$ throughout all updates.

After initializing \Cref{thm:uniformDegreeExpanderDecomposition} on some $G_i$,
we have $|V(G'_{i,j})| = O(|V(G_i)|)$ and $E(G'_{i,j}) = O(E(G_{i,j}))$
for all $j$ by \Cref{thm:uniformDegreeExpanderDecomposition}.
Also since at most $\phi |E(G_{i,j})|$ edges are removed from $G_{i,j}$,
we have $E(G'_{i,j}) = O(E(G_{i,j}))$ for all $i,j$ throughout all updates.

Further note that the minimum degree of $G_{i,j}$ 
is initially at least $\phi |E(G_{i,j})| / (4|V(G_{i,j})|)$ 
by \Cref{cor:averageDegreeExpanderDecomposition}.
Then during the updates, at most $\phi |E(G_{i,j})|$ edges are removed from $G_{i,j}$,
so we know the number of nodes in $G_{i,j}$ can only decrease by a constant factor.
Thus we have $|V(G'_{i,j})| = O(|V(G_i)|)$ for all $i,j$ throughout all updates.

\end{proof}

\subsection{Reduction via Dynamic Expander Decomposition with Uniform Degree Constraint}

To prove \Cref{thm:amortized:fully_dynamic_weighted}, we show the case in which the input graph is unweighted. 
Such a result can be proven by a direct application of the dynamic expander decomposition from \Cref{sec:uniform:dynamic_decomposition}. 
Remember that in  that setting  it suffices to study the case in which  the algorithm only receives  edge deletions.

\begin{lemma}\label{lem:amortized:fully_dynamic}
Assume $\H$ satisfies 
\eqref{con:union}
and \eqref{con:contraction},
and there exists a decremental algorithm 
$\A$ for $\H(\epsilon)$
on almost-uniform-degree $\phi$-expanders for any $\phi = O(1/\log^4 n)$. 
Then, there exists a fully dynamic algorithm $\B$ 
for $\H(\epsilon)$ on general unweighted graphs.

For the time complexity, let $P(m)\ge m$ be the pre-processing time of $\A$,
$S(n)\ge n$ an upper bound on the size of the output graph throughout all updates,
and $T(n)$ the amortized update time of $\A$. 
Then, the pre-processing time of $\B$ is $O(P(m))$, 
with amortized update time 
$$
O\left(
	\phi^{-3} \log^7 n \cdot T(n) 
	+ \frac{P(m)}{m \phi^{3}} \log^{7} n
\right).
$$
The size of the output graph is bounded by $O(S(n \log^3 n))$.

If $\A$ is a query-algorithm, then so is $\B$.
\end{lemma}

\begin{proof}
We start by describing the algorithm, 
then analyze the complexity and prove the correctness.

\paragraph{Algorithm}

We run the dynamic expander decomposition of \Cref{thm:dynamicExpanderDecompositionUniformDegree}.
The output is a decomposition $G = \bigcup_i G_i$
and graphs $G'_i$ where each $G'_i$ is a $\phi$-expander.
We further obtain decompositions $\bigcup_{v \in V(G_i)} X_{v,i} = V(G'_i)$ for all $i$.
For any $i$, the sets $X_{v,i}$ have the property that contracting each $X_{v,i}$ in $G'_i$ results in graph $G_i$
(i.e. $X_{v,i}$ represents node $v$ in $G_i$). 
Next, we run algorithm $\A$ on $G'_i$ to obtain $H'_i \in \H(G'_i, \epsilon)$.
Contracting the $X_v$ in $H'_i$ for $v \in V(G_i)$ results in some graph $H_i$.
We maintain $H := \bigcup_i H_i$.

All these graphs are maintained dynamically,
so when an edge is deleted from $G$,
we pass that edge deletion to the dynamic expander decomposition algorithm \Cref{thm:dynamicExpanderDecompositionUniformDegree}.
The decomposition might change in the following way:
(i) Some $G_i$ is removed from the decomposition. 
In that case also remove $H_i$ from $H$.
(ii) Some new $G_i$ is added to the decomposition. 
Initialize $\A$ on $G'_i$ to construct new $H'_i$ and $H_i$ and add $H_i$ to the output graph $H$.
(iii) Some edge is removed from $G'_i$. 
We pass the edge deletion to algorithm $\A$ to update $H'_i$. 
The graph $H_i$ and $H$ is updated accordingly to the changes performed to $H_i$.

\paragraph{Complexity}

The preprocessing costs $
\sum_i P(|E(G'_i)|) \le P(\sum_i |E(G'_i)|) \le P(O(m))$,
where we used $P(K) \ge K$ to move the sum inside and $\sum_i |E(G'_i)| = O(m)$ as guaranteed by \Cref{thm:dynamicExpanderDecompositionUniformDegree}.

During each update,
at most $O(\phi^{-3} \log^7 n)$ edges change in the decomposition per update (amortized).
So updating the instances of $\A$ costs $O((\phi^{-3} \log^7 n) \cdot T(n))$ amortized time per update,
where we used the fact that each $G'_i$  has  at most $O(n)$ nodes.
The amortized cost of initializing $\A$ on some new $G_i$ is bounded by
$$
O\left(\frac{P(|E(G'_i)|)}{|E(G'_i)| \phi^{3} \log^{-7} n}\right)
\le
O\left(\frac{P(m)}{m \phi^{3}} \log^{7} n\right),
$$
where we used  the fact that a $G'_i$ with $k$ edges is created after at least $\phi^{3} \log^{-7} k$ updates,
and that $P(k) \ge k$.
The total amortized update time is thus
$$
O\left(\phi^{-3} \log^7 n \cdot T(n) + \frac{P(m)}{m \phi^{3}} \log^{7} n\right).
$$
The size of the output graph is bounded by
$$
\sum_i S(|V(G'_i)|) 
\le
S(\sum_i |V(G'_i)|) 
\le
O(S(n \log^3 n)),
$$
because $\sum_i |V(G'_i)| = \sum_i O(|V(G_i)|) = O(n \log^3 n)$ 
by \Cref{thm:dynamicExpanderDecompositionUniformDegree}.
\paragraph{Correctness}

Each $G'_i$ is always a $\phi$-expander by \Cref{thm:dynamicExpanderDecompositionUniformDegree}, 
and therefore we have $H'_i \in \H(G'_i, \epsilon)$ by guarantee of $\A$, 
and also $H_i \in \H(G_i, \epsilon)$ by property \eqref{con:contraction}.
The algorithm always maintains $H = \bigcup_i H_i$ where each $H_i \in \H(G_i, \epsilon)$,
so $H \in \H(G, \epsilon)$ by property \eqref{con:union}.

\paragraph{Query-Case}

When $\A$ is a query-algorithm, then the graphs $H'_i$ and $H_i$ are not maintained explicitly.
Instead, whenever a query is performed for algorithm $\B$,
we perform the query operation of $\A$ for each $G'_i$ to obtain a $H'_i \in \H(G'_i, \epsilon)$.
Then the sets $X_{v,i}$ are contracted in $H'_i$ to obtain $H_i \in \H(G_i, \epsilon)$.
At last, we return $H := \bigcup_i H_i \subset \H(G, \epsilon)$.

The pre-processing and update time complexity is the same as before.
The query complexity is bounded by
$$
O\left(\sum_i Q(|V(G'_i)|)\right)
\le
O\left(\sum_i Q(|V(G_i)|)\right)
\le
O\left( Q(\sum_i |V(G_i)|)\right)
\le
O( Q(n \log^3 n)).
$$

\end{proof}

The extension to weighted graphs now follows directly 
from splitting the edges based on their edge weights 
into powers of $e^\epsilon$.

\begin{proof}[Proof of \Cref{thm:amortized:fully_dynamic_weighted}] 
Split the edges into groups $[e^{k\cdot \epsilon/2}, e^{(k+1)\epsilon/2})$ 
to obtain subgraphs $G_0,G_1,...G_t$ with $t = O((1+\epsilon^{-1}) \log W)$.
Then run \Cref{lem:amortized:fully_dynamic} on each subgraph to maintain $H_i \in \H(G_i, \epsilon/2)$
and the union of these graphs then satisfies $H := \bigcup_{i=0}^t e^{(i+1) \epsilon/2} \cdot H_i \in \H(G, \epsilon)$
by property \eqref{con:identity} and \eqref{con:union}.
\end{proof}

\section{Expander Pruning with Worst-case Update Time}
\label{sec:pruning}

In this section, we present an \emph{expander pruning} algorithm with
\emph{worst-case }update time for, roughly speaking, maintaining an
expander under edge deletions.  
This is the crucial primitive for the blackbox reduction  in \Cref{sec:worst_case_blackbox} for algorithms with worst-case update time.

Our  main result is summarised in   \Cref{thm:wc prune edge} below, which  
is an improvement of \cite[Theorem 5.1]{NanongkaiSW17} that
gives an algorithm for the same purpose. The improvement is twofold: first of all,  our algorithm works in arbitrary graphs and not just
on constant degree graphs; secondly,
  quantitatively our algorithm dominates
\cite[Theorem 5.1]{NanongkaiSW17} in both running time and all
quality guarantees. 
\begin{theorem}
[Worst-Case Expander Pruning]\label{thm:wc prune edge}
\label{lem:expanderPruning}
There exists a parameter $\gamma=2^{O(\sqrt{\log m})}$
and a deterministic algorithm 
that can preprocess a connected multi-graph $G=(V,E)$ with $m$
edges in $O(m)$  time so that, given a sequence of $\sigma=(e_{1},\dots,e_{k})$
of $k$ online edge deletions and a parameter $\phi\ge 1/\gamma$,
the algorithm maintains an edge set $P\subseteq E$ with the following
properties. Let $G_{i}$ be the graph $G$ after the edges $e_{1},\dots,e_{i}$
have been deleted from it; let $P_{0}=\emptyset$ be the set $P$
at the beginning of the algorithm, and for all $0<i\le k$, let $P_{i}$
be the set $P$ after the deletion of $e_{1},\dots,e_{i}$. Then, the following statements hold
for all $1\le i\le k$: 
\begin{itemize}
\item $P_{i-1}\subseteq P_{i}$,
\item $|P_{i}|\le i\gamma$, $|P_i \setminus P_{i-1}| \le \gamma$,
\item $P_{i}=E$ when $i\ge m/\gamma$.
\item If $G_0$ was a $\phi$-expander, then there is a set $W_i\subseteq P_i$
such that $G_{i}-W_i$ has only one connected component $C$ which is
not an isolated vertex, and $C$ is a $(1/\gamma)$-expander.
\end{itemize}
The worst-case time for updating
$P_{i-1}$ to $P_{i}$ is $O(\gamma)$ for each $i$.
\end{theorem}

We will prove   a variant of the theorem when
we only maintain a set of nodes instead of edges, which is summarised in \Cref{thm:wc prune node}. But, let's first show how 
\Cref{thm:wc prune node} could be applied to prove the main result of the section~(\Cref{thm:wc prune edge}).
\begin{lemma}
\label{thm:wc prune node}
There exists a parameter $\gamma=2^{O(\sqrt{\log m})}$ and  a deterministic algorithm such that,
given the adjacency list of a connected multi-graph $G=(V,E)$
with $m$ nodes, a sequence of $\sigma=(e_{1},\dots,e_{k})$ of $k$
online edge deletions, and a parameter $\phi\ge1/\gamma$
as input, the   algorithm
maintains a  vertex set $P\subseteq V$ with the following properties.
Let $G_{i}$ be the graph $G$ after the edges $e_{1},\dots,e_{i}$
have been deleted from it; let $P_{0}=\emptyset$ be the set $P$
at the beginning of the algorithm, and for all $0<i\le k$, let $P_{i}$
be the set $P$ after the deletion of $e_{1},\dots,e_{i}$. Then, the following holds 
for all $1\le i\le k$:
\begin{itemize}
\item $P_{i-1}\subseteq P_{i}$;
\item $\vol(P_{i})\le i\gamma$;
\item $P_{i}=V$ when $i\ge m/\gamma$;
\item If $G$ is a $\phi$-expander, then there is a set $Q\subseteq P$
such that $G_{i}[V\setminus Q]$ is a $(1/\gamma)$-expander. 
\end{itemize}
The worst-case time for updating $P_{i-1}$ to $P_{i}$ is $O(\gamma)$
for each $i$. 
\end{lemma}

\begin{proof}
[Proof of \Cref{thm:wc prune edge}]Given \Cref{thm:wc prune node},
we can prove \Cref{thm:wc prune edge} as follows. Given a graph $G=(V,E)$
with $m$ edges, we construct in time $O(m)$ a $\Delta$-reduction
graph $G'=(V',E')$ of $G$ where $\Delta=9$. So $G'$ is a constant
degree graph and $|V'|=\Theta(m)$. Note that each edge $e=(u,v)\in E$
has a corresponding edge $e'\in E'$ connecting super-nodes $X_{u}$
and $X_{v}$ of $G'$. We make a pointer between $e$ and $e'$ for
each $e\in E$. Then, we apply \Cref{thm:wc prune node} on $G'$ and
maintain a set $P'\subseteq V'$ of nodes in $G'$. Let $\gamma'=2^{O(\sqrt{\log m})}$ be
the parameter $\gamma$ in \Cref{thm:wc prune node} when the input
is $G'$. To maintain a set $P\subseteq E$ of edges in $G$, we just let 
$P$   contain all edges $e\in E$ such that the corresponding edge
$e'\in E'$ is incident to $P'$. Now, we prove the correctness.

Let $P_{i}$ and $P'_{i}$ be the sets $P$ and $P'$ after the $i$-th
edge deletion. First, $P_{0}=P'_{0}=\emptyset$. Next, we have $P_{i-1}\subseteq P_{i}$
because $P'_{i-1}\subseteq P'_{i}$. Next, we have $|P_{i}|\le\vol(P'_{i})\le i\gamma'$.
Also, when $i\ge m/\Theta(\gamma')$, $P_{i}=E$ as $P'_{i}=V$. On
each step, the algorithm from \Cref{thm:wc prune node} takes $O(\gamma')$
time. Additionally, the time for updating $P_{i}$ from $P_{i-1}$
is at most $O(1)\cdot|P'_{i} \setminus P'_{i-1}|=O(\gamma')$. This is where
we exploit the fact that $G'$ has constant degree.

It remains to prove the statement about the expansion guarantee. Suppose
$G$ is a $\phi$-expander where $\phi\ge1/\gamma$,
then we know that $G'$ is also an $\Omega(\phi)$-expander by \Cref{prop:delta-reduction prop}.   
By \Cref{thm:wc prune node}, we have that there is a set $Q'\subseteq P'$
of nodes in $G'$ such that $G'_{i}[V'-Q']$ is a $(1/\gamma')$-expander. 

We let $Q$ contain all edges $e\in E$ such that the corresponding
edge $e'\in E'$ is incident to $Q'$. Note that $Q\subseteq P$.
Moreover, $G_{i}[V\setminus Q]$ has only one connected component $C$ which is
not an isolated vertex. Observe that this connected component $C$
can be obtained from $G'_{i}[V'\setminus Q']$ by contracting all nodes
in the same super-node. As contraction never decreases conductance
and $G'_{i}[V' \setminus Q']$ is a $(1/\gamma')$-expander, we have that $C$
must be a $(1/\gamma')$-expander. By setting the parameter $\gamma=\Theta(\gamma')$
appropriately, we obtain \Cref{thm:wc prune edge}.
\end{proof}

Hence, the remaining part of the section is   to prove  \Cref{thm:wc prune node}. The main idea
is to apply \Cref{thm:amortized prune} repeatedly as already
done in Section 5.2 of \cite{NanongkaiSW17}. However, our analysis
is arguably cleaner.

\paragraph{Setting up.}
Before describing the algorithm, we start with some notations. When
\Cref{thm:amortized prune} is given $G=(V,E)$, $\sigma$, and $\phi$
as the input, we let $(X,P)=\prune_{\phi}(G,\sigma)$ denote the
output where $X=G[V \setminus P]$ is an $\phi/6$-expander if $G$ initially
is a $\phi$-expander. 

Let $G$ be an input graph of \Cref{thm:wc prune node} with $m$ nodes.
Let $\phi\ge1/2^{O(\sqrt{\log m})}$ be the conductance parameter
from \Cref{thm:wc prune node}. Throughout the sequence $\sigma=(e_{1},\dots,e_{k})$
of edge deletions, let $G_{\tau}$ denote the graph after the $\tau$-th
deletion and $\sigma_{[\tau,\tau']}=(e_{\tau},\dots,e_{\tau'})$.
Let $G_{0}$ be the graph $G$ before any deletion. Let $T=\phi m/10$,
$\ell=\left\lceil \sqrt{\log m}\right\rceil $ and $\Delta=T^{1/\ell}$.
For convenience, assume that $T$ is such that $\Delta$ is a integer.
We claim that we can assume the number of edge deletions is at most
$T$, i.e., $k\le T$. We will prove this claim at the end.

In our algorithms, there will be $\ell+1$ \emph{levels.} For each
$1\le i\le\ell+1$, let $\phi_{i}=\phi/6^{i}$. We maintain a level-$i$
graph $X^{i}$ and a level-$i$ pruning set $P^{i}$. Let $X_{\tau}^{i}$
and $P_{\tau}^{i}$ denote $X^{i}$ and $P^{i}$ after the $\tau$-th
deletion. We define $X_{0}^{i}=G_{0}$ and $X_{\tau}^{0}=G_{\tau}$.
For any time step $\tau<T$, we let $\round_{i}(\tau)=\left\lfloor \frac{\tau}{T/\Delta^{i}}\right\rfloor \cdot T/\Delta^{i}$
be the biggest multiple of $T/\Delta^{i}$ which is at most $\tau$.
Notice that $\round_{\ell}(\tau)=\tau$, and $\round_{0}(\tau)=0$
as $\tau<T$. 

\paragraph{The algorithm.}

For each level $i\in\{1,\dots,\ell\}$, if $\tau\le T/\Delta^{i}$
or $\tau\neq\round_{i}(\tau)$, we just define $(X_{\tau+1}^{i},P_{\tau+1}^{i})=(X_{\tau}^{i},P_{\tau}^{i})$.
Otherwise, when $\tau=\round_{i}(\tau)$ and $\tau>T/\Delta^{i}$,
we ensure that we have finished executing 
\[
(X_{\tau}^{i},P_{\tau}^{i})=\prune_{\phi_{i-1}}\left(X_{\tau}^{i-1},\sigma_{[\max\{0,\round_{i-1}(\round_{i-1}(\tau)-1)\}+1,\round_{i}(\tau-1)]}\right).
\]
For $i=\ell+1$, for every $\tau$, we execute
\[
(X_{\tau}^{\ell+1},P_{\tau}^{\ell+1})=\prune_{\phi_{\ell}}\left(X_{\tau}^{\ell},\sigma_{[\tau,\tau]}\right).
\]
 Let $P$ denote our output pruning set and $P_{\tau}$ be the set
$P$ after time step $\tau$. $P$ is simply the union of all $P_{\tau}^{i}$
for all $i$ and $\tau$ that have been computed. We ensure that when
we finished computing $P_{\tau}^{i}$, we already include $P_{\tau}^{i}$
into $P_{\tau}$. This completes the description of the algorithm.
Now, we analyze the algorithm.

\paragraph{Analysis.}

For any $\phi$ and $\tau$, we say that $X$ is \emph{an induced
$\phi$-expander of time $\tau$ }if $\Phi_{X}\ge\phi$ and $X=G_{\tau}[U]$
for some $U\subseteq V$. From this definition and by \Cref{thm:amortized prune},
we have the following fact:

\begin{fact}
\label{fact:induced expander}For any $\phi,\tau,\tau'$ where $\tau\le\tau'$,
suppose that $X$ is an induced $\phi$-expander of time $\tau$ and
$(X',P')=\prune_{\phi}(X,\sigma_{[\tau+1,\tau']})$. Then $X'$ is
an induced $\phi/6$-expander of time $\tau'$. 
\end{fact}

\begin{lemma}
Suppose $G_{0}$ is a $\phi$-expander. Then, for any $1\le\tau<T$,
we have
\begin{itemize}
\item for $i\in\{1,\dots,\ell\}$, $X_{\tau}^{i}$ is an induced $\phi_{i}$-expander
of time $\max\{0,\round_{i}(\round_{i}(\tau)-1)\}$.
\item for $i=\ell+1$, $X_{\tau}^{i}$ is an induced $\phi_{i}$-expander
of time $\tau$. That is, $\Phi_{G_{\tau}[V(X_{\tau}^{\ell+1})]}\ge\phi_{\ell+1}$.
\end{itemize}
\end{lemma}

\begin{proof}
We prove by induction on $\tau$ and $i$. For each $i\in\{1,\dots,\ell\}$,
if $\tau\le T/\Delta^{i}$, then $X_{\tau}^{i}=G_{0}$ which is an
induced $\phi$-expander of time $0$. Now, assume $\tau>T/\Delta^{i}$,
and so $\round_{i}(\round_{i}(\tau)-1)\ge0$. When $\tau\neq\round_{i}(\tau)$,
we have $X_{\tau}^{i}=X_{\round_{i}(\tau)}^{i}$ which is an induced
$\phi_{i}$-expander of time $\round_{i}(\round_{i}(\tau)-1)$ by
induction. When $\tau=\round_{i}(\tau)$, we have
\[
(X_{\tau}^{i},P_{\tau}^{i})=\prune_{\phi_{i-1}}\left(X_{\tau}^{i-1},\sigma_{[\max\{0,\round_{i-1}(\round_{i-1}(\tau)-1)\}+1,\round_{i}(\tau-1)]}\right).
\]
By induction hypothesis, $X_{\tau}^{i-1}$ is an induced $\phi_{i-1}$-expander
of time $\max\{0,\round_{i-1}(\round_{i-1}(\tau)-1)\}$. By \Cref{fact:induced expander},
$X_{\tau}^{i}$ is then an induced $\phi_{i}$-expander of time $\round_{i}(\tau-1)=\round_{i}(\round_{i}(\tau)-1)$
as $\tau=\round_{i}(\tau)$. When $i=\ell+1$, as $X_{\tau}^{\ell}$
is an induced $\phi_{\ell}$-expander of time $\max\{0,\round_{i}(\round_{i}(\tau)-1)\}=\tau-1$
and $(X_{\tau}^{\ell+1},P_{\tau}^{\ell+1})=\prune_{\phi_{\ell}}(X_{\tau}^{\ell},\sigma_{[\tau,\tau]})$,
so we have $X_{\tau}^{\ell+1}$ is an induced $\phi_{\ell+1}$-expander
of time $\tau$.
\end{proof}

\begin{lemma}
Suppose $G_{0}$ is a $\phi$-expander. For any $\tau$, there is
$Q_{\tau}\subseteq P_{\tau}$ where $G_{\tau}[V \setminus Q_{\tau}]$ is a $2^{-O(\sqrt{\log m})}$-expander. 
\end{lemma}

\begin{proof}
Observe that $V-V(X_{\tau}^{\ell+1})\subseteq P_{\tau}$ and $\Phi_{G_{\tau}[V(X_{\tau}^{\ell+1})]}\ge\phi_{\ell+1}=\phi/6^{\ell+1}$.
By setting $Q_{\tau}=V \setminus V(X_{\tau}^{\ell+1})$ for each $\tau$, we
are done.
\end{proof}
\begin{lemma}
The worst-case update time is $O((\Delta\log m)/\phi_{\ell}^{2})=O((T^{1/\ell}\log m)6^{2\ell}/\phi^{2})=2^{O(\sqrt{\log m})}$.
\end{lemma}

\begin{proof}
For $i\le\ell$, by \Cref{thm:amortized prune}, the time for executing
$$(X_{\tau}^{i},P_{\tau}^{i})=\prune_{\phi_{i-1}}\left(X_{\tau}^{i-1},\sigma_{[\max\{0,\round_{i-1}(\round_{i-1}(\tau)-1)\}+1,\round_{i}(\tau-1)]}\right)$$
is at most $O\left(\frac{T\log m}{\Delta^{i-1}\phi_{i-1}^{2}}\right)$ and $\vol(P_{\tau}^{i})=O\left(\frac{T}{\Delta^{i-1}\phi_{i-1}^{2}}\right)$.
Additionally, the time for including the output $P_{\tau}^{i}$ into
$P$ is $O(|P_{\tau}^{i}|)$, and the time for constructing $X_{\tau}^{i}=X_{\tau}^{i-1}-P_{\tau}^{i}$
is $O(\vol(P_{\tau}^{i}))$. In total, this takes $O\left(\frac{T\log m}{\Delta^{i-1}\phi_{i-1}^{2}}\right)$.

We will distribute the work above evenly into each time step between
time $\round_{i}(\tau-1)+1$ and $\round_{i}(\tau)$. As $\round_{i}(\tau)=\tau$
and so $\round_{i}(\tau-1)\le\tau-T/\Delta^{i}$, the work on each
step is at most $O\left(\frac{T\log m}{\Delta^{i-1}\phi_{i-1}^{2}}\right)/(T/\Delta^{i})=O(\Delta\log m/\phi_{i-1}^{2})$
in worst-case. For $i=\ell+1$, executing $\prune_{\phi_{\ell}}(X_{\tau}^{\ell},\sigma_{[\tau,\tau]})$,
including $P_{\tau}^{\ell+1}$ into $P$, and constructing $X_{\tau}^{\ell+1}$
take $O\left(\log m/\phi_{\ell}^{2}\right)$ per step. Hence, the worst-case
update time for each time step is $\sum_{i=1}^{\ell}O(\Delta\log m/\phi_{i-1}^{2})+O(\log m/\phi_{\ell}^{2})=O(\Delta\log m/\phi_{\ell}^{2})$
because $\sum_{i\le\ell}1/\phi_{i}^{2}=O(1/\phi_{\ell}^{2})$. 
\end{proof}
\begin{lemma}
For any time $\tau$, $P_{\tau-1}\subseteq P_{\tau}$ and $\vol(P_{\tau})\le\tau2^{O(\sqrt{\log m})}$.
\end{lemma}

\begin{proof}
It is obvious that $P_{\tau-1}\subseteq P_{\tau}$ because we never
remove any vertex out of $P$. Next, we bound $\vol(P_{\tau})$. Fix
$i$. Let $\tau_{a}=\round_{i-1}(\round_{i-1}(\tau)-1)$. We bound
$\vol(\bigcup_{\tau'=1}^{\tau_{a}}P_{\tau'}^{i})$ and $\vol(\bigcup_{\tau'=\tau_{a}}^{\tau}P_{\tau'}^{i})$. 

To bound $\vol(\bigcup_{\tau'=1}^{\tau_{a}}P_{\tau'}^{i})$, observe
that if $\tau_{a}>0$, then $\tau_{a}\ge T/\Delta^{i-1}$. So we assume
$\tau_{a}\ge T/\Delta^{i-1}$. Between $\tau'\in[1,\tau_{a}]$, the
number of times we compute a new $P_{\tau'}^{i}$ is $O(\frac{\tau_{a}}{T/\Delta^{i}})$.
Each such $P_{\tau}^{i}$ has volume $\vol(P_{\tau}^{i})=O((\frac{T}{\Delta^{i-1}}\log m)/\phi_{i-1}^{2})$.
So $\vol(\bigcup_{\tau'=1}^{\tau_{a}}P_{\tau'}^{i})=O((\tau_{a}\Delta\log m)/\phi_{i-1}^{2})=\tau2^{O(\sqrt{\log m})}$
as $\tau_{a}\le\tau$. To bound $\vol(\bigcup_{\tau'=\tau_{a}}^{\tau}P_{\tau'}^{i})$,
the number of time steps $\tau'\in[\tau_{a},\tau]$ we compute a new
$P_{\tau'}^{i}$ is $|\{\tau'\mid\tau'=\round_{i}(\tau')\}|=O(\Delta)$.
Each such $P_{\tau}^{i}$ has volume $\vol(P_{\tau}^{i})=O((\tau\log m)/\phi_{i-1}^{2})$.
So $\vol(\bigcup_{\tau'=\tau_{a}}^{\tau}P_{\tau'}^{i})=O((\tau_{a}\Delta\log m)/\phi_{i-1}^{2})=\tau2^{O(\sqrt{\log m})}$.
Summing up over all $i$, we are done.
\end{proof}
It remains to prove that the assumption that $k\le T$ is without
loss of generality.
\begin{proposition}
We can assume the number of edge deletions is at most $T$, i.e.,
$k\le T$. Also $P_{\tau}=V$ when $\tau\ge m/2^{O(\sqrt{\log m})}$.
\end{proposition}

\begin{proof}
We partition $V$ into sets $V_{1},\dots,V_{T}$ where $|V_{i}|=\gamma=2^{O(\sqrt{\log m})}$.
After each step $\tau$, we artificially add $V_{\tau}$ into $P$.
This does not affect the correctness of any claim above, but additionally,
we have that once $\tau\ge m/\gamma$, then $P_{\tau}=V$. In particular,
when $\tau\ge T\ge m/\gamma$, we can just stop our algorithm as there
is nothing to be done further.
\end{proof}

\section{Extension of Eppstein et al.~Sparsification}
\label{sec:sparsification}

In this section, we give another crucial primitive for obtaining the blackbox reduction in \Cref{sec:worst_case_blackbox} for algorithms with worst-case update time.

Sometimes it is  easier to design a dynamic algorithm
that is fast on sparse input graphs.
In this section we will show that, as long as the maintained graph satisfies the properties \eqref{con:union} and 
\eqref{con:nested}, the existence of a dynamic algorithm on sparse graphs is sufficient for us to design the one for dense graphs.  Our high-level idea behind this reduction is to sparsify the dense input graph via a modification of the sparsification technique of \cite{EppsteinGIN97}. 
Remember that any graph problem $\H$ that satisfies  \eqref{con:union} and \eqref{con:nested} has the following property:
Let $\bigcup_{i=1}^d G_i = G$ be a decomposition of some graph $G$ into $d$ subgraphs.
Assume we maintain a sparsifier $H_i \in \mathcal{H}(G_i,\epsilon)$ for each $i=1,...,d$,
and additionally a sparsifier $H \in \H(\bigcup_{i=1}^d H_i, \epsilon)$.
Then, it holds that   $H \in \H(G, 2\epsilon)$,
because $\bigcup_{i=1}^d H_i \in \H(G, \epsilon)$ by \eqref{con:union} 
and $\H(\bigcup_{i=1}^d H_i, \epsilon) \subset \H(G, 2\epsilon)$ by \eqref{con:nested}. Hence, it is sufficient for us to assume that the input graph is sparser than the initial input graph $G$, since every subgraph  $G_i$ of  $G$ is clearly sparser than $G$ and the union of every $G_i$'s sparsifier $H_i$, i.e., $\bigcup_{i=1}^d H_i$, is sparser than $G$. Applying this reduction technique recursively gives us the following result:

\begin{theorem}\label{thm:eppstein}
Let $N \ge 1$, $d \ge 2$ be some fixed parameters and $L :=  \lceil \log(N)/ \log d \rceil$. Assume that   $\H$ satisfies \eqref{con:union} and \eqref{con:nested}, and the following holds:
\begin{itemize}
    \item There exists an algorithm $\A$ for $\H(\epsilon)$
with the property that the ratio
between the largest and the  smallest edge weight in the output graph, 
compared to the one of the input graph, increases by at most $w \ge 1$.
\item The input graph $G$ for algorithm $\A$ has $n$ nodes, $m$ edges, and the ratio between the largest and smallest edge weights is bounded by $wW^L$ for some  $W\geq 1$. 
\end{itemize}
We define $T(n,m), R(n), S(n), P(m)$ as follows:
\begin{itemize}
	\item $P(m)\ge m$ is the preprocessing time.
	\item $S(n) \ge n$ is an upper bound on the number of edges in the output graph.
	\item $T(n,m)$ is the update time of $\A$.
    \item $R(n, m)$ is the recourse. %
\end{itemize}
Then there exists an algorithm $\B$ for problem $\H(\epsilon\cdot L)$
on up to $N$ node graphs\footnote{%
This means the number of nodes must be bounded by $N$ throughout all updates. }
with ratio  between the largest  and the  smallest edge weight bounded by  $W$.
The output graph size is $S$
with weight ratio $Ww^L$.
Algorithm $\B$ has update time $O(L \cdot R(n,dS(n))^{L} \cdot T(n,d S(n)))$
and recourse $O(R(n,dS(n))^{L})$.
The preprocessing time is bounded by $O(L N d P(dS(n)))$,
though if $\H(\epsilon, G)$ is a set of subgraphs of $G$,
then the preprocessing decreases to
$O(L P(m))$.

\end{theorem}

\begin{proof}
We will first describe the high level idea:
We split the graph $G$ into $d$ equally sized parts $G_1,...,G_d$ and
use the given dynamic algorithm to maintain $H_1,...,H_d$ where
$H_i \in \mathcal{H}(G_i, p)$.
Note that one update to $G$ leads to a single update in only one of the
$G_i$ at cost $T(n,m/d)$ and recourse $R(n, m/d)$.
We also maintain $H \in \mathcal{H}( \bigcup_{i=1}^d H_i, p )$ in
$R(n, m/d) \cdot T(n,dS(n))$ time, 
because $\bigcup_{i=1}^d H_i$ has $dS(n)$ edges 
and  one update to some $G_i$ leads to $R(n,m/d)$ changes in $H_i$.
This means we now obtained a new algorithm, which for some $G$ maintains
$H \in \mathcal{H}(G, 2p)$ in $R(n, m/d) \cdot T(n,dS(n))$ update time.
To prove the lemma, the trick  is to repeat this
recursively by splitting each $G_i$ into another $d$ smaller graphs.

\paragraph{Constructing smaller graphs.}

Let $G^{(0,1)} := G$  be the original
graph and $G^{(\ell+1,d(i-1)+1)},\ldots,G^{(\ell+1,di)}$ the $d$ graphs obtained
by splitting $G^{(\ell,i)}$ into $d$ equal sized graphs. So $G^{(\ell,i)}$
refers to the $i$th graph in level $\ell$, i.e., the number of recursions. 

\paragraph{Maintaining the graph property.}

Let $L = \lceil \log (N) / \log d \rceil$ be the bottom level,
then each $G^{(L,i)}$ has at most $m/N \le n \le S(n)$ edges.
We use the given algorithm $\mathcal{A}$ to maintain
$H^{(L,i)} \in \mathcal{H}(G^{(L,i)}, p)$ for every $1\leq i\leq d^L$, and use
 the same  algorithm to maintain
$H^{(\ell ,i)} \in \mathcal{H}\left(\bigcup_{k=1}^d H^{(\ell+1, di+k)}, p\right)$
for every $0 \leq \ell
\leq L-1$ and $1\leq i \leq d^{\ell}-1$.
By induction, we have that 
$$
H^{(\ell,i)} \in 
\mathcal{H}\left(
	\bigcup_{j = d^{L-\ell}\cdot (i-1)}^{d^{L-\ell}\cdot i} G^{(L,j)}, 
	p(L-\ell)
\right)
$$
and thus $H^{(0,1)} \in \mathcal{H}(G, pL)$, 
so we are maintaining the property $\H$ of $G$.
Note that, as long as  the ratio of the edge weights of $G$ is by $W$,
  the ratio of the edge weights of $H^{(\ell,i)}$ must be bounded by $Ww^L$,
as the ratio increases by a factor of $w$ in each level.
Moreover, the size of each input graph given to $\mathcal{A}$ is bounded by $dS(n)$, 
thus every instance of $\mathcal{A}$ has update and recourse complexity $T(n,dS(n)),R(n,dS(n))$.

\paragraph{Update Propagation and Complexity.}

When there is an edge deletion to $G$, we only have to traverse the
decomposition of $G$ to find the graph $G^{(L,i)}$ that contains the edge
that needs to be deleted, and update that graph accordingly. 
Likewise, if an edge is added to $G$, simply add the edge to some
$G^{(L,i)}$ with less than $S(n)$ edges 
(we can maintain a list of such $G^{(L,i)}$).
These updates can result in $R(n,dS(n))$ changes to $H^{(L,i)}$ 
which means we must now perform $R(n,dS(n))$ updates to the algorithm that maintains
$H^{(L-1 ,j)} \in \mathcal{H}(\bigcup_{k=1}^d H^{(L, dj+k)}, p)$
where $j$ is such that $i \in [dj+1,dj+d]$. 
This is repeated recursively,
i.e. for every changed edge in some $H^{(\ell ,i)}$ for some $\ell,i$, 
we perform up to $R(n,dS(n))$ updates to the algorithm 
that maintains $H^{(\ell-1,j)} \in \mathcal{H}(\bigcup_{k=1}^d H^{(\ell, dj+k)}, p)$, 
where $j$ is such that $i \in [dj+1,dj+d]$.
This implies that, when updating one edge in $G$, 
we  need to  perform up to $R(n,dS(n))^L$ updates
to the algorithm maintaining $H^{(0,1)}$,
because the number of updates increases by a factor of $R(n,dS(n))$ for every level.
The update time is thus
$O(L \cdot R(n,dS(n))^L \cdot T(n,dS(n)))$.

\paragraph{Pre-processing Complexity.}

We are left with analysing the pre-processing time.
For any $\ell,i$ the graph $H^{(\ell,i)}$ has at most $dS(n)$ edges,
so the preprocessing of that graph requires $O(P(dS(n)))$ time.
The total complexity of initialization is thus bounded by
$$
O\left(\sum_{\ell}\sum_{i} P(dS(n))\right)
=
O(L d^L P(dS(n)))
=
O(L N d P(dS(n))).
$$
However, if  $\H(G, \epsilon)$ is a set of subgraphs,
then for any fixed $\ell$  the total number of edges 
of all $H^{(\ell,i)}$, $i=1,...,d^\ell$, together is bounded by $m$,
because they are edge disjoint and subgraphs of $G$.
Let $m_i^\ell := \sum_{j=1}^d |E[H^{(\ell-1,id+j)}]|$ be  the size of the graph that is sparsified to $H^{(\ell,i)}$, 
then the pre-processing time of layer $\ell$ can be bounded by
\begin{align*}
\sum_{i \ge 0} P(m^\ell_i)
&\le
P\left(\sum_{i \ge 0} m^\ell_i\right) \le O(P(m))
\end{align*}
since $P(m) \ge m$. Hence, the total pre-processing time is  $O(L P(m))$.
\end{proof}

\section{Worst-Case Reduction to Almost Uniform Degree Expanders}
\label{sec:worst_case_blackbox}

In this section, we finally show how to obtain the black box reduction for dynamic algorithms with worst-case time.
This reduction will allow us to deamortize our main results in \Cref{thm:main spanner,thm:main cut} formally proved in \Cref{part:algorithm}. Moreover, it is crucial for  \Cref{cor:main spectral} which shows the first non-trivial dynamic algorithm for maintaining spectral sparsifiers with worst-case update time.

To show this reduction, we combine are three important tools that we have developed from previous sections (\Cref{sec:uniform_degree_reduction,sec:pruning,sec:sparsification}): 
(1) the dynamic uniform degree expander decomposition of \Cref{thm:dynamicExpanderDecompositionUniformDegree},
(2) the improved expander pruning algorithm with worst-case update time from \Cref{thm:wc prune edge}, and 
(3) the extension of the sparsification technique from \Cref{thm:eppstein}.

The reduction for worst-case update time is a bit more complicated
than the reductions proven in \Cref{sec:simple_amortized_reduction} and \Cref{sec:uniform_degree_reduction},
because we can neither guarantee uniform degrees nor that the graph stays an expander.
We can only guarantee that the degree is initially near uniform, 
and that the graphs contains an expander that is not much smaller. This is formalized as follows:
\begin{definition}
\label{def:worst-case:decremental_algorithm}
Let $\H$ be a graph problem.
We call an algorithm $\mathcal{A}$ a 
``decremental algorithm on pruned $\phi$-sub-expanders for $\H(\epsilon)$'', 
if the following holds  for any \emph{unweighted} $n$-node graph $G$:
\begin{itemize}
\item $\mathcal{A}$ maintains $H \in \mathcal{H}(G,\epsilon)$ 
under edge deletions to $G$.
\item During the preprocessing/initialization,
the algorithm is allowed to assume 
that the initial graph $G$ is a $\phi$-expander and
for minimum degree $\Delta$ we have $\max\deg_G(v) \le O(1/\phi)\Delta$.
\item The algorithm is allowed to assume 
that the $i$th update (i.e $i$th edge deletion) 
also receives as input a set $P_i \subset E$ 
with $|P_i| \le 2^{O(\sqrt{\log n})}$. 
Let $P = \bigcup P_i$, 
then there exists a $2^{O(\sqrt{\log n})}$-expander $W \subset G$ 
with $G \setminus P \subset W$.
Further $\deg_W(v) \ge \Delta/2$ for all $v \in V(G \setminus P)$, 
where $\Delta$ is the minimum degree of $G$ during initialization.
(Note that $\deg_W(v)$ may be smaller for nodes that were completely removed by the pruning.)
\end{itemize}
\end{definition}

The main result is summarized as follows:

\begin{theorem}[Worst-Case Blackbox Reduction]
\label{thm:reduction_worst_case}
Assume $\H$ satisfies \property. We
fix some $N \ge 1$ and $d\geq 2 $ such that 
$L :=  \lceil \log(N)/ \log d \rceil$, and assume that there exists a decremental algorithm $\A$ for $\H(\epsilon)$
on pruned $\phi$-sub-expanders for any 
$\phi = O(1/\log^4 n)$. Then, there exists 
  a fully dynamic algorithm $\B$ 
for $\H(\epsilon\cdot L)$ on general \emph{weighted} on (up to) $N$ node graphs\footnote{%
This means that throughout all updates, the number of nodes is not allowed to be larger than $N$.}
whose ratio of largest to smallest weight is $W$.

For the time complexity, let $P(m)$ be the pre-processing time of $\A$ on an input \emph{unweighted} graph $G$ with $n$ nodes and $m$ edges, 
$S(n)$   the size of the output graph after the preprocessing, and 
  $T(n)$ and $R(n)$ be the worst-case update time and recourse of $\A$
and let $w \ge W$ be the ratio between the  largest and the  smallest edge weight of the maintained output graph.
Then,  algorithm $\B$ maintains an output graph of size $\tilde O(S(n) L \log(wW))$
with weight ratio $W(w/\epsilon)^L$.
The update time is
$$\
\left(R(n)^{O(L)} + (R(n)/\epsilon)^{O(L)}\right)
\cdot \left(T(n)+\frac{P(dS(n)\log (wW))}{S(n)}\right),$$
and the preprocessing time is $\tilde{O}(N d P\left(dS(n)(1+\epsilon^{-2}) \log(wW)\right))$.
However, if $\H(G, \epsilon)$ is a set of subgraphs of $G$,
then the preprocessing decreases to
$\tilde{O}(P(m)(1+\epsilon^{-2}) \log(wW))$.

\end{theorem}

In order to prove \Cref{thm:reduction_worst_case}
we start by extending the decremental algorithm on sub-expanders
to a fully dynamic algorithm on general unweighted graphs.
The resulting algorithm has the property 
that the output graph grows more and more dense with each update,
so the algorithm is only useful for short sequences of updates.

All lemmas throughout this subsection always assume 
that $\H$ satisfies \property{}, 
and there exists a decremental algorithm $\A$ for $\H(\epsilon)$
on pruned $\phi$-sub-expanders for any $\phi = O(1/\log^4 n)$. 
For this algorithm $\A$ the preprcoessing time can be bounded by $P(m)\ge m$,
the size of the output graph after the pre-processing is bounded by $S(n)\ge n \log n$, 
and $T(n),R(n) \ge 2^{O(\sqrt{\log n})}$ are bounds for the worst-case update time and recourse of $\A$. 

\begin{lemma}\label{lem:fully_dynamic_growing_output}
Then, there exists a fully dynamic algorithm $\B$ 
for $\H(\epsilon)$ on general unweighted graphs.

The pre-processing time of $\B$ is $O(P(m))$, 
with worst-case update time $O(T(n))$ and recourse $O(R(n))$.
After $t$ updates the size of the output graph is bounded by $\tilde{O}(S(n) + t\cdot R(n))$.
\end{lemma}

\begin{proof}

We first describe the algorithm,
then we prove the complexity and correctness.

\paragraph{Algorithm}

During the preprocessing we run \Cref{thm:uniformDegreeExpanderDecomposition}
to obtain the decomposition $\bigcup_i G_i = G$.
We further obtain for each $G_i$ a graph $G'_i$ of near uniform degree
and node sets $X_{v,i} \subset V(G'_i)$.
Next, we initialize the pruning algorithm of \Cref{thm:wc prune edge} 
the assumed algorithm $\A$ on each $G'_i$.
This way we obtain graphs $H'_i$ for each $G'_i$.
When contracting the sets $X_{v,i}$ in $H_i$ we obtain a graph $H_i$.
At last, we define $H := \bigcup_i H_i$.

When performing an edge insertion to $G$,
we add the edge directly to $H$.
If that edge is later deleted again, 
we simply remove it from $H$ again.

When deleting an edge $\{u,v\}$ that was already part of $G$ during the initializaton,
then that edge is contained in some $G_i$.
Then there also exists $u',v' \in V(G'_i)$ with $\{u', v'\} \in E(G'_i)$.
Delete this edge from $G'_i$ and inform the pruning algorithm.
The pruning algorithm will prune out another set of edges from $G'_i$.
In addition to that, for every edge $\{u,v\}$ that is pruned by the pruning algorithm,
we also prune the two additional edges incident to $u$ and and two additional edges incident to $v$.

Next, we inform the algorithm $\A$ of the edge deletion and the pruned edges.
Now algorithm $\A$ changes the graph $H'_i$ in some way
and we perform the corresponding changes to the graph $H_i$
and $H$,
so $H_i$ it is still a valid contraction of $H'_i$ and $H = \bigcup_i H_i$.

\paragraph{Correctness}

By \Cref{thm:uniformDegreeExpanderDecomposition} we have that each $G'_i$ 
is a $\phi$-expander of almost uniform degree during the preprocessing.
So $\A$ returns $H'_i \in \H(G'_i,\epsilon)$.
As $G_i$ is obtained from $G'_i$ via contracting the $X_{v,i}$ for each $v \in V(G_i)$,
we have that $H_i \in \H(G_i, \epsilon)$ by property \eqref{con:contraction}.
Next, we have $H = \bigcup_i H_i \in \H(G, \epsilon)$ by property \eqref{con:union} and $\bigcup_i G_i = G$.
This property is still true when performing edge insertions by \Cref{lem:insertion}.
When performing edge deletions, 
the expander pruning makes sure that each $G_i$ is still contained in a $2^{O(\sqrt{\log n})}$-expander.
More accurately, let $P_i$ be the so far pruned edges by \Cref{thm:wc prune edge}, 
and let $P'_i$ be the edges we pruned so far in addition to that from nodes incident to $P_i$.
As we prune an extra edge from $u$ and $v$ for every $\{u,v\} \in P_i$,
we have that $G_i \setminus (P_i \cup P'_i) \subset G_i \setminus P_i$.
As $G_i \setminus P_i$ is contained in some expander $W_i$ 
with $G_i \setminus P_i \subset W \subset G_i$ by \Cref{thm:wc prune edge},
$G_i \setminus (P_i \cup P'_i)$ must be contained in the same expander $W$ as well.
Further, we have $\deg_{W_i}(v) \ge \deg_{G_i\setminus P_i}(v) \ge \Delta_i / 2$ 
for all $v \in V(G_i \setminus (P_i \cup P'_i))$,
where $\Delta_i$ is the minimum degree of $G_i$ during the preprocessing,
because we always prune one additional edge, 
so the degree in $G_i \setminus (P_i \cup P'_i)$ decreases twice as fast 
as in $G_i \setminus P_i$.

\paragraph{Complexity}

Computing the expander decomposition takes
$O(\phi^{-1} m \log^6 m)$ time.
Initializing algorithm $\A$ on each $G'_i$ takes
$$
\sum_i P(|E(G'_i)|)
\le
P\left(\sum_i |E(G'_i)|\right)
\le
O(P(m))
$$
time.
Edge insertions take only $O(1)$ time
and edge deletions require
$$
O(T(n)) + 2^{O(\sqrt{\log n})} = O(T(n))
$$
time, because each $G'_i$ has at most $O(|V(G_i)|)) = O(n)$ nodes
by \Cref{thm:uniformDegreeExpanderDecomposition}.
The size of the graph after the initialization is bounded by
$$
\sum_i S(|V(G'_i)|)
\le
S\left(\sum_i |V(G'_i)|\right)
\le
O(S(n \log^2 n))
=
\tilde{O}(S(n)),
$$
where we used $S(k) \ge k$ and $\sum_i |V(G'_i)| = O(n \log^2 n)$ 
by \Cref{thm:uniformDegreeExpanderDecomposition}.
Further, after $t$ updates at most $O(R(n))$ edges are added to the graph,
so after $t$ updates the size is bounded by
$\tilde{O}(S( n) + t\cdot R(n))$.
\end{proof}

The previous result from \Cref{lem:fully_dynamic_growing_output}
showed how to obtain a fully dynamic algorithm whose output graph becomes denser with each update.
We now show that by performing periodic resets,
the size of the output graph can be bounded.
However, the resulting update time becomes much  slower on dense input graphs,
because one has to pay $\Omega(m)$ whenever the graph resets.

\begin{lemma}\label{lem:fully_dynamic_bounded_output}
There exists a fully dynamic algorithm $\B$ 
for $\H(p)$ on general graphs.

The preprocessing time of $\B$ is $O(P(m)/\epsilon)$, 
with worst-case update time 
$$
\tilde O\left(\left(T(n)+\frac{R(n)P(m)}{S(n)}\right)(1+1/\epsilon)\right)
$$ 
and recourse $\tilde{O}(R(n)(1+1/\epsilon))$.
The output graph is bounded by $\tilde{O}(S(n)(1+1/\epsilon))$.

\end{lemma}

\begin{proof} 
The high level idea of our algorithm is to run $C := 4+(e^{\epsilon/2}-1)^{-1} \approx 4+2/\epsilon$ copies 
of the algorithm from \Cref{lem:fully_dynamic_growing_output} in parallel.
Assuming that $H_i \in \H(G_i,\epsilon/2)$ is maintained for $G_i = G$
by the $i$th copy of \Cref{lem:fully_dynamic_growing_output},
our algorithm scales every edge of $H_i$ by a factor of $\Delta =(3+(e^{\epsilon/2}-1)^{-1})^{-1}$ 
and sets $H := \bigcup_{i=1}^{C}\Delta\cdot H_i$. %

When the output $H_i$ of one of the copies grows too large,
we will slowly (i.e., a few edges per update) remove its edges from $H$.
After all edges of $H_i$ are removed from $H$, 
then the $i$th copy of \Cref{lem:fully_dynamic_growing_output}
will be re-initialized by performing the preprocessing 
on the current input graph $G$, and we slowly add the edges of the new $H_i$ to $H$ again.
By carefully synchronizing the copies, 
the algorithm ensures that only one $H_i$ is removed from $H$ at a time,
so $H$ stays in $\H(G,\epsilon)$ by \Cref{lem:interpolation}.

The repeated resets ensure that the output never grows too large, 
and in the mean time the slow addition/removal of edges from $H_i$ to $H$ ensures 
that the worst-case update time stays low.
We remark that this strategy above is to avoid 
adding all the edges of some $H_i$ to $H$ at once,
as this would lead to $\Omega(S(n))$ update time 
due to the large output size when listing all the edges.

To describe our approach in detail, 
we call $k$ consecutive steps~(edge updates) a \emph{cycle}, 
and we discuss the implementation of the algorithm within one such cycle.
As mentioned before, let $G_i=G$ be the $i$th copy of $G$, 
and $H_i \in \H(G_i,\epsilon/2)$. 
These $C$ many $H_i$ are updated as in the following manner, 
depending on the specific time step within each cycle: 

\begin{enumerate}
\item For the first $(1-\Delta)\cdot k$ steps, 
$G$ and $G_i$ are updated in the usual way. 
That is, every edge insertion and deletion is updated in both $H_i$ and $H$. 
 
\item For the next $k\Delta/4$ edge updates,
$G$ and $G_i$ receive the same updates. 
The algorithm performs the corresponding edge removals from $H_i$ and $H$, 
but anything added to $H_i$ is \emph{not} added to $H$.
In addition, the algorithm removes $\tilde O\left(S(n)/(k\epsilon)+R(n)/\epsilon\right)$
edges of $H_i$ from $H$. Since the total number of edges in $H_i$ is 
$\tilde O\left(S(n)+k\cdot R(n)\right)$
by \Cref{lem:fully_dynamic_growing_output}, 
at the end of $(1-\Delta)\cdot k+k\Delta/4$ steps 
$H_i$ will be completely removed from $H$. 
\label{phase:removal}

\item We know that at the end of $(1-\Delta)\cdot k+k\Delta/4$ steps 
all edges of $H_i$ are removed from $H$. 
At this point, the algorithm applies \Cref{lem:fully_dynamic_growing_output} 
to recompute a new $H_i$ of $G_i$ in $O(P(m))$ time, 
but the overall re-computation spreads out over the next $k\Delta/4$ steps. 
That is, in each of the $k\Delta/4$ steps the algorithm spends $O(P(m)(1+1/\epsilon)/k)$ time for recomputing $H_i$. 
We remark that $G_i$ will not be updated during this process, 
and instead the algorithm queues the performed changes to $G$.

\item At the end of $k(1-\Delta)+ k\Delta/4 + k\Delta/4 = k-k\Delta/2$ steps, 
the reconstruction of $H_i$ described above is finished. 
However, this $H_i$ is in $\H(G_i,\epsilon/2)$ where $G_i$ is what graph $G$ looked like $k\Delta/4$ steps ago.
To catch up to the current input graph $G$, 
the algorithm starts to always perform two queued updates whenever $G$ is updated. 
At the same time all new updates to $G$ are also added to the queue.
This means after $k\Delta/4$ rounds, 
the graph $G_i$ will be identical to $G$ again 
and $H_i$ will again satisfy $H_i \in \H(G,\varepsilon/2)$.

\item At the end of $k-k\Delta/4$ steps, 
$H_i$ is now a valid element of $\H(G,\epsilon/2)$.
For the next $k\Delta/4$ steps all updates to $G$ are also performed on $G_i$, 
and the algorithm adds an edge to $H$ whenever it adds the edge to $H_i$. 
Similarly, for every edge in $H$, the algorithm removes it from $H_i$ and $H$ when the deletion of this edge is requested.
In addition, for each step the algorithm adds 
$\tilde O\left(S(n)(1+1/\epsilon) k+R(n)(1+1/\epsilon)\right)$
edges of $H_i$ to $H$. 
Therefore, at the end of $k\Delta/4$ additional steps, 
all edges of $H_i$ are added to $H$.
\label{phase:addition}

\item Combining the items above, 
we know that at the end of a cycle all of $H_i$ is added to $H$. 
The algorithm goes back to item~1 and starts a new cycle. 
\end{enumerate}
We remark that the algorithm runs 
$C$ copies of \Cref{lem:fully_dynamic_growing_output}.
To make sure that $H \in \H(G,\epsilon)$,
the algorithm requires each copy to be slightly phase shifted, 
so that only one of the copies is somewhere in phase \ref{phase:removal} to \ref{phase:addition}.
This can be obtained during the initial preprocessing 
by having copy $i$ jump to phase \ref{phase:removal} preemptively 
after $ik\Delta$ steps.

It remains to analyze the time complexity of the algorithm.
Consider the cost of one cycle for the $i$th copy.
We know that the update time and recourse for every $H_i$ is $O(T(n))$ and $R(n)$ respectively
when the $i$th copy is in phase~1. 
By observing that we run $O(1+1/\epsilon)$ many copies in parallel for phase 1, 
these term increase by an $O(1+1/\epsilon)$ factor.
At any point in time, 
there is also a single copy somewhere in phase~2 to phase~5, 
for which the worst case update time is 
$\tilde O\left(
(S(n)/k + R(n) +P(m)/k)(1+1/\epsilon) +T(n)
\right)$
and the recourse is $\tilde O((R(n)+S(n)/k)(1+1/\epsilon))$.
The size of the output graph is bonded by $\tilde O((S(n)+ k R(n))(1+1/\epsilon))$
as the size of each $H_i$ is bounded by $O(S(n) + k R(n))$.
The statement of \Cref{lem:fully_dynamic_bounded_output} is then obtained by using
$k = S(n)/R(n)$ and $R(n) \le T(n)$.
\end{proof}

Via the classic decomposition of the input graph based on its edge weights,
we can extend the previous result to weighted input graphs.

\begin{lemma}\label{lem:fully_dynamic_weighted}
There exists a fully dynamic algorithm $\B$ 
for $\H(p)$ on general \emph{weighted} graphs
whose ratio between the largest and the smallest weight is $W$.

The pre-processing time of $\B$ is $O(P(m) (1+\epsilon^{-2}) \log W)$, 
with worst-case update time $\tilde O((T(n)+R(n)P(m)/S(n))(1+1/\epsilon))$ 
and recourse $\tilde O(R(n)(1+1/\epsilon))$.
The output graph is bounded by $\tilde O(S(n)(1+\epsilon^{-2}) \log W)$.

\end{lemma}

\begin{proof}
Split the edges into groups $[e^{k\cdot \epsilon/2}, e^{(k+1)\epsilon/2})$ 
to obtain subgraphs $G_0,G_1,...G_t$ with $t = O((1+\epsilon^{-1}) \log W)$.
Then we run \Cref{lem:fully_dynamic_bounded_output} on each graph to obtain results $H_i \in \H(G_i, \epsilon/2)$.
The union of the $H_i$ results in $H := \bigcup_i H_i \in \H(G, \epsilon)$.
The pre-processing time and output size increase by an $O((1+\epsilon^{-1}) \log W)$ factor, 
as we have that many different $G_i$.
The update time and recourse does not change because with each update we modify only a single $G_i$.
\end{proof}

As \Cref{lem:fully_dynamic_weighted} is only fast on sparse input graphs,
we now apply the sparsification technique of \Cref{sec:sparsification}
to speed-up the result to dense input graphs.

\begin{proof}[Proof of \Cref{thm:reduction_worst_case}]
By \Cref{lem:fully_dynamic_weighted} there exists an Algorithm $\C$ for $\H(\epsilon)$ on general \emph{weighted} graphs
whose ratio between the  largest and the  smallest weight is $(wW^L)$.
The preprocessing time of $\C$ is $O(P(m)(1+\epsilon^{-2}) L \log(wW) )$, 
with worst-case update time $\tilde O((T(n)+R(n)P(m)/S(n))(1+\epsilon^{-1}))$ 
and recourse $\tilde O(R(n)(1+\epsilon^{-1}))$.
The output graph is bounded by $\tilde O(S(n)(1+\epsilon^{-2}) L \log(wW))$.

We now apply \Cref{thm:eppstein} to this algorithm $\C$ to obtain $\B$.
The preprocessing time of $\B$ is either bounded by
\begin{align*}
O(L N d P\left(dS(n)(1+\epsilon^{-2}) L \log(wW)\right))
\le
\tilde{O}(N d P\left(dS(n)(1+\epsilon^{-2}) \log(wW)\right))
,
\end{align*}
as $L = \tilde{O}(1)$ and $P(m)$ can be assumed to be polynomial time.
Alternative, if the sparsifier is a subgraph, then we can bound the preprocessing time by
\begin{align*}
&~
O((m + 
\underbrace{\tilde{O}\left(P(m)(1+\epsilon^{-2}) L \log(wW) \right)}_{\text{Preprocessing of $\C$}}
)L) \\
\le&~
\tilde{O}(P(m)(1+\epsilon^{-2}) \log(wW)).
\end{align*}
The update time is
\begin{align*}
&~\tilde O(L \cdot \left.\underbrace{\tilde O(R(n)(1+1/\epsilon))}_{\text{Recourse of $\C$}}\right.^{L} 
\cdot \underbrace{\tilde O\left(\left(T(n)+\frac{R(n)P(dS(n))(1+\epsilon^{-2})L\log (wW)}{S(n)}\right)(1+1/\epsilon)\right)}_{\text{Update time of $\C$}}) \\
\le&~
(R(n)^{O(L)} + (R(n)/\epsilon)^{O(L)})
\cdot \left(T(n)+\frac{P(dS(n))\log (wW))}{S(n)}\right)
\end{align*}
where we use that $R(n) \ge \tilde{O}(1)$ and $L \ge 1$.

\end{proof}

\part{Dynamic Sparsifiers\label{part:algorithm}}

\section{Cut Sparsifiers Against an Adaptive Adversary}
\label{sec:adaptive}
In this section, we show an adaptive algorithm for maintaining cut sparsifiers with amortized and worst-case update time:

\begin{theorem}\label{thm:amortizedCutSparsifierMainResult}
For any $1 \leq k \leq \log n$, there exists an algorithm that maintains an $O(k)$-approximate cut-sparsifier $\tilde{G}$ on any dynamic graph $G$. The algorithm  works against an adaptive adversary with high probability, has $O(m)$ initialization time, $\tilde{O}(n^{1/k})$ amortized update time and contains $\tilde{O}(n \log W)$ edges.
\end{theorem}
Note that the amortized update time is $\polylog(n)$ when $k = \log n$. 

\begin{theorem}\label{thm:worstCaseCutSparsifierMainResult}
For any $1 < k \leq \sqrt{\log n} /\log \log n$, there exists an algorithm that maintains an $k^{O(k)}$-approximate cut-sparsifier $\tilde{G}$ on any dynamic graph $G$, that works against an adaptive adversary with high probability, has $\tilde{O}(m \log W)$ initialization time, $n^{O(1/k)}$ worst-case update time and contains $\tilde{O}(n \log W)$ edges.
\end{theorem}

Setting $k = \log \log n/ (\log\log\log n)^c$ for some constant $c$, we obtain the following corollary. Setting $k = 1/0.001$, we obtain second corollary.

\begin{corollary}
There exists an algorithm with worst-case update time $n^{o(1)} \log W$ that maintains an $O(\log n)$-approximate cut-sparsifier of size $n^{1+o(1)} \log W$ against an adaptive adversary.
\end{corollary}
\begin{corollary}
There exists an algorithm with worst-case update time $O(n^{0.001} \log W)$ that maintains an $O(1)$-approximate cut-sparsifier of size $\tilde{O}(n^{1+0.001} \log W)$ against an adaptive adversary.
\end{corollary}

To prove the above theorems, we apply the reductions  from \Cref{sec:uniform_degree_reduction,sec:worst_case_blackbox} and work on decremental unweighted graphs $H$ that have near-uniform degrees and remain $\phi$-expanders. The main technical result of the section is \Cref{thm:mainDecrCutSparsifier} below which is achieved using a new \emph{proactive sampling} approach. 
Note that,  instead of   requiring that the input graph $H$ always has the desired properties, we allow the theorem to work on any $H$, but only guarantee that the output $\tilde{H}$ is a cut sparsifier when $H$ is a near-uniform degree expander. After this, we can combine \Cref{thm:mainDecrCutSparsifier} with our black box reduction to get similar results for the more general fully dynamic setting, with no requirements on $H$.

\begin{theorem}\label{thm:mainDecrCutSparsifier}
Let $H$ be   a decremental unweighted graph 
 that satisfies $\Delta_{\max} \geq \Delta_{\min} \geq \frac{80 \log n}{\phi}$ for some fixed $1/\phi \in (1, n)$, where  $\Delta_{\max}$ is the maximum degree of the initial   graph $H$. Then, there exists an algorithm that maintains $\tilde{H} \subseteq H$ along with a weight function $\tilde{w}$ such that
\begin{itemize}
    \item At any stage, where $H$ has $\Delta_{\min} \leq \min\deg(H)$ and $\Phi(H) \geq \phi$, we have that $\tilde{H}$ weighted by $\tilde{w}$ is an  $O(\log n)$-approximate cut sparsifier of $H$, and
    \item the graph $\tilde{H}$ has at most $\tilde{O}\left(|V(H)| \left(\frac{\Delta_{\max}}{ \phi \Delta_{\min}}\right)^2\right)$ edges.
\end{itemize}
The algorithm can be initialized in $O(m)$ time and has worst-case update time $\tilde{O}\left(\left(\frac{\Delta_{\max}}{\Delta_{\min}\phi}\right)^3 \right)$. The algorithm runs correctly with high probability.
\end{theorem}

The section is organized as follows: 
\Cref{subsec:algoCutSparsifier} presents the algorithm that achieves bounds of \Cref{thm:mainDecrCutSparsifier} but with only the amortized update time. Our presented algorithm will be further analyzed in \Cref{subsec:anaylsisCutSparsifier} with the focus on analyzing the upper bound of the weight of every cut; within the same subsection we also show how to modify our algorithm to obtain the worst-case update time. Finally, we combine \Cref{thm:mainDecrCutSparsifier} with our black-box reduction to obtain several the main results of this section in \Cref{subsec:resultsViaBlackBoxFromCutSparsifier}.

\subsection{Algorithm}
\label{subsec:algoCutSparsifier}

In this section, we give the description of the initialization and update procedure of our algorithm with amortized update time. For the ease of presentation, we express an edge sampling probability as 
  $$\rho = \rhoValue,$$
  where $\alpha$ is the constant that controls the exponent in our high probability bound. This probability $\rho$ is also  called compression probability in \cite{BenczurK15}. We also use a vertex-degree threshold defined by  $\zeta = \phi \Delta_{min}$.

\paragraph{Overview.} 
The algorithm maintains   for every vertex $v \in V(H)$ a set $S_v$, which is   a subset of selected incident edges $\textsc{Inc}_H(v)$ of $v$, and  maintains the cut sparsifier $\tilde{H}$ to be the union of these selected edges, i.e., $\tilde{H} = \bigcup_v S_v$.  The algorithm further gives  every sampled   edge in $\tilde{H}$ the weight $1/\rho$, and these new weights defines the weight function $\tilde{w}$.
  For convenience,  we assume that the assigned weight for edges not in $\tilde{H}$ is $0$.

\begin{algorithm2e}
\caption{Vertex Sampling.}
\label{alg:preprocessing}

\SetKwProg{procedure}{Procedure}{}{}
\SetKwFunction{init}{Init}
\SetKwFunction{sample}{SampleVertex}
\SetKwFunction{update}{EdgeUpdate}
\procedure{\sample{$v$, $H$, $\rho$}}{
    $S_v \gets \emptyset$\;
    \tcc{(Implement the loop below using \Cref{thm:fast sampling}).\label{lne:sampleEdge}}
    \ForEach{$(u,v) \in \textsc{Inc}_H(v)$}{
        Sample $S_v$    by adding each edge independently with probability $\rho$ 
    }
}
\end{algorithm2e} 

\paragraph{Initialization.} We initialize each set $S_v$, using the procedure $\textsc{SampleVertex}(v, H, \rho)$ described in \Cref{alg:preprocessing}. This procedure is given a vertex $v \in V(H)$ and first empties the set $S_v$ and then samples every edge in $\textsc{Inc}_H(v)$ independently and uniformly at random into a collection $S_v$. Whilst we use in the analysis that every edge is evaluated one after another, we use for efficiency, the following efficient
subset sampling result,  which allows us to sample a subset of edges in time proportional to the number of edges sampled.

\begin{theorem}[Efficient Subset Sampling, see \cite{knuth1997seminumerical, devroye2006nonuniform, bringmann2012efficient}]\label{thm:fast sampling}
Given a universe $U$ of size $n$ and a sampling probability $p$. Then, we can compute a set $S \subseteq U$ in worst-case time $O(pn \log n)$ where each element of $U$ is in $S$ with probability $p$ independently. The algorithm succeeds with high probability.
\end{theorem}

\paragraph{Edge Deletions and the Sampling Schedule.} The algorithm maintains a \emph{sampling schedule} $T_{\schedule}(u) \subseteq \mathbb{N}$ for each vertex $u \in V$  that records at which stages  vertex $u$ should be sampled, i.e., at which stage we invoke $\textsc{SampleVertex}(u, H, \rho)$ to replace the current set $S_u$ in the cut sparsifier $\tilde{H}$. We initialize each schedule $T_{\schedule}(u)$ to $\{0\}$ since we run the procedure $\textsc{SampleVertex}(u, H, \rho)$ for every $u \in V$ to initialize the algorithm. 
 
To process an adversarial deletion of edge $(u,v)$ from $H$ at stage $t$, we invoke procedure $\textsc{EdgeDeletion}((u,v),t)$ given in \Cref{alg:update} which adapts the schedule and executes the scheduled sampling procedures. More precisely, the algorithm first adds $t, t + 2^0, t + 2^1, t + 2^2, \dots$ to both $T_{\schedule}(u)$ and $T_{\schedule}(v)$. Moreover, if the degree of one of them is divisible by $\zeta$ (i.e. every $\zeta$ times that an incident edge is deleted to a vertex $z$), we schedule a vertex resampling for all of its neighbors at stages $t, t + 2^0, t + 2^1, t + 2^2, \dots$. While for the current cut-sparsifier only the updates scheduled for stage $t$ affect $\tilde{H}$ directly, the subsequent updates manifest a \emph{proactive} sampling strategy that makes it harder for the adversary to change $\tilde{H}$ to deviate from the promised approximation bound on all cuts. We further point out that once a degree drop was realized at a vertex $c$, we also add the current stage to a set $T_{DegreeUpdate}(c)$. This set serves no direct purpose in our algorithm but is a bookkeeping device that is of crucial importance to our analysis. Once the schedule has been updated, all vertex samplings scheduled for the current stage are executed.

\begin{algorithm2e}
\caption{Handling an Edge Deletion.}
\label{alg:update}
\SetKwProg{procedure}{Procedure}{}{}
\SetKwFunction{sample}{SampleVertex}
\SetKwFunction{update}{EdgeDeletion}
\procedure{\update{$(u,v), t$}}{
    \ForEach{$z \in \{u,v\}$}{
        Add $t, t + 2^0, t + 2^1, t + 2^2, t + 2^3, \dots$ to $T_{\schedule}(z)$\;
        \If(\label{lne:ifDegreeHigh}){$\deg_H(z)$ is divisible by $\zeta$}{
            Add $t$ to $T_{DegreeUpdate}(z)$ \label{lne:degreeUpdateT} \tcc*[f]{used only for analysis}\;
            \ForEach{$y \in \mathcal{N}(z)$}
            {
                Add $t, t + 2^0, t + 2^1, t + 2^2, t + 2^3, \dots$ to $T_{\schedule}(y)$ \label{lne:neighborInducedStages}
            }
        }
    }
    \ForEach(\label{lne:foreachSchedule}){$v \in V$ where $t \in T_{\schedule}(v)$}{
        $S_v \gets$ \sample{$v$, $H$, $\rho$}\;
        Adjust $\tilde{H} = \bigcup_{u \in V} S_u$ accordingly.
    }
}
\end{algorithm2e}

\subsection{Analysis}
\label{subsec:anaylsisCutSparsifier}

We analyze the above algorithm in three parts:
\begin{enumerate}
    \item \label{item:lowerBound} In \Cref{subsec:lowerBoundAdaptiveSparsifier}, we prove a lower bound on the weight of cuts in $\tilde{H}$, i.e.,  we show that for every $X \subseteq V(H)$, with high probability we have 
    \[
        |E_H(X, \overline{X})|/2 \leq  w_{\tilde{H}}(E_H(X, \overline{X})).
    \]
    \item \label{item:upperBound} We provide an upper bound on cuts in $\tilde{H}$, i.e.,  we show that for every $X \subseteq V(H)$, with high probability we have
    \[
        w_{\tilde{H}}(E_H(X, \overline{X})) \leq O(\log n) |E_H(X, \overline{X})|.
    \]
    Since it is rather involved to derive the upper bound directly, we first give a simple upper bound in \Cref{subsec:simpleUpperBoundForCutSparsifier} that provides some intuition about the algorithm and our \emph{proactive sampling} approach and provides a first approximation ratio of $\tilde{O}(\Delta_{\max}/(\Delta_{\min}\phi))$. We then improve the upper bound to a logarithmic approximation ratio by refining the analysis in \Cref{subsec:improvingTheUpperBoundForCutSparsifier}.
    \item Finally we show how to tweak the algorithm slightly, and show that the algorithm can obtain the claimed worst-case update time, the space bound and the guarantees from the bounds \ref{item:lowerBound} and \ref{item:upperBound} which culminates in a proof of \Cref{thm:mainDecrCutSparsifier}. 
\end{enumerate}
We point out that the two guarantees on the weight of the cuts in $\tilde{H}$ (weighted by $\tilde{w}$) are only guaranteed at stages where $\min\deg(H) \geq \Delta_{\min}$ and $\Phi(H) \geq \phi$. We implicitly assume that the above guarantees are given in our proofs. 

Throughout the proofs, we use a superscript to indicate the version of a variable or function, for example $H^t$ refers to the graph $H$ at the end of stage $t$. Finally, the following folklore result will be used in various parts of the proof.

\begin{theorem}[Folklore Result]\label{thm:folklore}
It holds for any positive integer $n$ and $k$  that ${n \choose k} \leq \left(\frac{e \cdot n}{k}\right)^k$.
\end{theorem}

\subsubsection{Lower Bound}
\label{subsec:lowerBoundAdaptiveSparsifier} 
In this section  we establish a lower bound on the weight of any cut in the cut-sparsifier $\tilde{H}$. The main result is summarized in the following lemma.

\begin{lemma} \label{lma:lowerBoundAdaptiveCutSparsifier}
Throughout the entire algorithm, 
we have that for any cut $(X, \overline{X})$ at any stage $t$, we have
\begin{equation}\label{eq:cutSatisfiesLowerBound}
|E_{H^t}(X, \overline{X})|/2 \leq  \tilde{w}^t(X, \overline{X})
\end{equation}
with probability $1- n^{-\alpha}$ for $\alpha > 1$.
\end{lemma}
\begin{proof}
The first crucial observation for the proof is that each edge $(u,v)$ has two chances to appear in the sparsifier $\tilde{H}$: it can be in $S_u$ or it can be in $S_v$. In particular, if $(u,v) \in S_u$, then calling \textsc{VertexSample}$(v,H,\rho)$ cannot remove $(u,v)$ from $\tilde{H}$; the only way it can be removed is via a call to \textsc{VertexSample}$(u,H,\rho)$

To prove the lemma, let us first consider a specific cut $(X, \overline{X})$, where $k = |X|\leq |\overline{X}|$. The main focus of our proof is on establishing that constraint 
(\ref{eq:cutSatisfiesLowerBound}) is violated for a specific cut $(X, \overline{X})$ at a specific stage $t$ with probability at most $n^{-10\alpha k}$ (this is the result of the final claim in this lemma, \Cref{clm:finalClaimLowerBound}).

We can then obtain 
\Cref{lma:lowerBoundAdaptiveCutSparsifier} which stipulates that there exists \emph{no} cut that violates constraint (\ref{eq:cutSatisfiesLowerBound}) by taking the union over the violating events for all cuts and obtain that there indeed does not exist such a violating cut with probability at most
\begin{align*}
    \mathbb{P}&\left[\bigcap_{X \subseteq V} \{|E_{H^t}(X, \overline{X})|/2 \leq  \tilde{w}^t(E_{\tilde{H}^t}(X, \overline{X}))\}\right] \\
    &\geq 1 - \mathbb{P}\left[\bigcup_{X \subseteq V} \{|E_{H^t}(X, \overline{X})|/2 >  \tilde{w}^t(E_{\tilde{H}^t}(X, \overline{X}))\}\right] \\
    &\geq 1 - 2 \cdot \sum_{k=1}^{n/2} {n \choose k} n^{-10\alpha k} \geq 1 - n^{-\alpha}
\end{align*}
by using DeMorgan and a simple union bound.

To show that the probability that cut $(X, \overline{X})$ violates Constraint (\ref{eq:cutSatisfiesLowerBound}) at any stage $t$ is at most $n^{-10\alpha k}$, we proceed to prove  
through the following three steps: 
\begin{enumerate}
    \item There is a subset $C_{\geq \epsilon}$ of the vertices in $X$ that carries at least an $(1-\epsilon)$-fraction of the weight of the edges in the cut $(X, \overline{X})$ in $H$, for every $0 \leq \epsilon \leq 1$, and \label{item:subsetCwithMany}
    \item if Constraint (\ref{eq:cutSatisfiesLowerBound}) is violated, then a large number of the vertices in $C_{\geq \epsilon/4}$ have sampled an extremely small fraction of the edges in the cut that they are incident to, and
    \item this latter event occurs with extremely small probability.
\end{enumerate}

To make Item \ref{item:subsetCwithMany} more formal,  let us define the set ${C}_{\geq \epsilon}$ to be the collection of vertices $c \in X$, such that $|E_{H^t}(\{ c \}, \overline{X})| \geq \phi\Delta_{\min} \epsilon$. We can then prove that most edges in $E_{H^t}(X,\overline{X})$ are incident to a vertex in $C_{\geq \epsilon}$.

\begin{restatable}{claim}{almostEntireCutAtFraction}
\label{clm:concentrationOfEdgesAtCertainVertices}
Let $\epsilon$ be defined as before, Then, it holds for  every stage $t$ that 
\begin{equation}
(1-\epsilon)|E_{H^t}(X, \overline{X})| \leq  |E_{H^t}({C}_{\geq \epsilon}, \overline{X})|.
\end{equation}
\end{restatable}
\begin{proof}
Since $k = |X|$ and  every vertex in $X$ has degree at least $\Delta_{\min}(H^t) \geq \Delta_{\min}$, it holds that  the total volume of $X$ is $\vol_{H^t}(X) \geq k \cdot \Delta_{\min}$. By the definition of conductance, we therefore have that 
\begin{equation}\label{eq:manyEdgesInCut}
|E_{H^t}(X, \overline{X})| \geq \vol_{H^t}(X) \cdot \phi \geq k \phi\Delta_{\min}.
\end{equation}

Next, observe that every vertex that is not in ${C}_{\geq \epsilon}$ contributes less than $\phi\Delta_{\min} \epsilon$ edges to $E_{H^t}(X \setminus C_{\geq \epsilon}, \overline{X})$. However, since there are at most $k$ vertices in $X \setminus C_{\geq \epsilon}$, we have that 
\begin{equation} \label{eq:fewEdgesInRestrCut}
    |E_{H^t}(X \setminus C_{\geq \epsilon}, \overline{X})| < k \phi\Delta_{\min}\epsilon.
\end{equation}
Combining  (\ref{eq:manyEdgesInCut}) and (\ref{eq:fewEdgesInRestrCut}), we have that at least a $(1-\epsilon)$-fraction of all edges in the cut has to be incident to a vertex in ${C}_{\geq \epsilon}$.
\end{proof}

Next, let us define the notion of underestimating.

\begin{definition}\label{def:underestimating}
For each vertex $c \in C_{\geq \epsilon}$, stage $t$ and a real number $\beta$, we say that $c$ is $\beta$-underestimating at stage $t$ if 
\[
\tilde{w}^t(\{ c \}, \overline{X}) \leq \beta \cdot |E_{H^t}(\{ c \}, \overline{X})|.
\]
We say that $c$ is $\beta$-underestimating if there exists a stage $t$ where it is $\beta$-underestimating at stage $t$.
\end{definition}

We point out that $\tilde{w}^t(\{ c \}, \overline{X})$ in $\tilde{H}^t$ is lower bounded by the number of edges in $S_c$ at stage $t$ time $1/\rho$, since an edge $(x,y)$ is in $\tilde{H}$ with weight $1/\rho$ if it is sampled into either the set $S_x$ or $S_y$. For the remaining part  of the proof,  for each edge $(x,y)$ in the cut with $x \in X$ and $y \in \overline{X}$ we only  focus  on whether $(x,y)$ is present in $S_x$, as this is sufficient to establish the lower bound. The following claim will be used in our analysis.
 
\begin{claim}\label{clm:ifLemmaVioThenManyAlphaUnderestimating}
If $|E_{H^t}(X, \overline{X})|/2 > \tilde{w}^t(X, \overline{X})$, then there are  at least $\eta = \phi k \Delta_{\min}/(8\Delta_{\max})$ vertices $C$ in $C_{\geq 1/20}$ that are $3/4$-underestimating.
\end{claim}
\begin{proof}
Assume for contradiction that $C$, the set of $3/4$-underestimating vertices is of size less than $\eta$. Observe that every vertex $c$ in $C_{\geq 1/20} \setminus C$ satisfies at any stage $t$ that
\[
\tilde{w}^t(\{ c \}, \overline{X}) > 3/4 \cdot |E_{H^t}(\{ c \}, \overline{X})|
\]
Now, if $C$ was empty, then we could lower bound the number of edges in $|E_{H^t}(C_{\geq 1/20}, \overline{X})|$ by $(19/20)k\phi\Delta_{\min}$ using \Cref{clm:concentrationOfEdgesAtCertainVertices}. To account for the fact that $C$ might not be empty, and at most $ \eta$ vertices in $C$ could potentially have no incident edges in $\tilde{H}^t$. But since the maximum degree is bounded by $\Delta_{\max}$, we have
\begin{align*}
    \tilde{w}^t(X, \overline{X}) \geq (19/20)\tilde{w}^t(C_{\geq 1/20}, \overline{X}) &> (19/20)(3/4) \cdot |E_{H^t}(C_{\geq 1/20}, \overline{X})| - \eta \cdot \Delta_{\max} \\
    &> (7/10) \cdot |E_{H^t}(C_{\geq 1/20}, \overline{X})| - \phi k \Delta_{\min}/8 \\
    &> (1/2) \cdot |E_{H^t}(C_{\geq 1/20}, \overline{X})|\\
    &\geq |E_{H^t}(X, \overline{X})|/2,
\end{align*}
which leads to a   contradiction. Hence, the statement follows.
\end{proof}

Let us now prove the final crucial claim where we use the following version of a classic Chernoff bound.

\begin{theorem}[Scaled Chernoff Bound]
\label{thm:scaledChernoff}
Let $Y = \sum_i Y_i$, where each random variable $Y_i \in [0,W]$ for some $W > 0$, and $\{Y_i\}_i$ are independently distributed. Then, it holds for any $0 < \delta < 1$ that
\[
    \mathbb{P}[Y < (1-\delta)\mathbb{E}[Y]] \leq e^{-\frac{\delta \mathbb{E}[Y]}{ 2W}}.
\]  
\end{theorem}

\begin{claim}
\label{clm:alphaUnderestimateSeldom}
For any subset $C \subseteq C_{\geq 1/20}$ of size $\ell$, the probability that \emph{all} vertices in $C$ are $3/4$-underestimated at some stage $t$ is at most $n^{-\frac{2^6(\alpha+1)\Delta_{\max}}{\Delta_{\min}\phi}\ell}$.
\end{claim}
\begin{proof}
The algorithm proceeds by repeatedly calling $\textsc{SampleVertex}(v, H, \rho)$ for various vertices $v \in V$. Let $c_1$ be the first vertex in $C$ for which $\textsc{SampleVertex}$ is called, let $c_2$ be the second, and so on. Note that we could have $c_i = c_j$, if a $\textsc{SampleVertex}$ is called twice for some vertex in $C$. 
Let  $Y_i$ be the boolean random variable that is true if (1) at the time that $\textsc{SampleVertex}(v, H, \rho)$ is called we have $c_i \in C_{\geq 1/20}$ and (2) after the new $S_{c_i}$ is sampled we have that $c_i$ is $\alpha$-underestimating. 

We now observe that if $c_i \in C_{\geq 1/20}$ we have 
\begin{align*}
    \mathbb{P}&[Y_i \;|\; Y_0, Y_1, Y_2, \dots, Y_{i-1}, c_i \in C_{\geq 1/20} ] \\ 
    &= \mathbb{P}\left[\tilde{w}^t(\{ c \}, \overline{X}) < 3/4 \cdot |E_{H^t}(\{ c \}, \overline{X})|\right] \\
    &\leq  e^{-\frac{|E_{H^t}(\{ c \}, \overline{X})|\rho}{16}} \\
    &\leq  e^{-\frac{\phi\Delta_{\min}\rho}{2^9}} \\
    &<  n^{-\frac{2^7(\alpha+1)\Delta_{\max}}{\Delta_{\min}\phi}}, 
\end{align*}
where we use that the sampling of edges incident to $c$ is independent of previous events in the first step, and therefore the probability is exactly the probability that $c$ is $3/4$-underestimating at the particular stage $t$. But this in turn is based on $|E_{H^t}(\{ c \}, \overline{X})|$ independent edge sampling experiments that have value $0$ or $1/\rho$. The first inequality comes from the scaled Chernoff bound in \Cref{thm:scaledChernoff} and finally we use the definition of $C_{\geq 1/20}$ which implies that there are at least $\phi\Delta_{\min}/20$ edges in the set and plug in $\rho = \frac{2^{16}(\alpha+1)\log n \Delta_{\max} }{\Delta_{\min}^2\phi^2}$ explicitly.

We further note that if $c_i \notin C_{\geq 1/20}$, then the probability of $Y_i$ is $0$ by definition. 

Now, note that are clearly at most $mn$ vertices $c_i$, and hence at most $mn$ variables $Y_i$: the adversary can perform at most $m$ deletions, and, as a very rough upper bound, each deletion can at most lead to a resampling of each of the $\ell \leq n$ vertices in $C$. We further note that we have 
\begin{align*}
\lefteqn{\mathbb{P}[\textrm{at least } \ell \textrm{ of the } Y_i \textrm{ are true}]}\\
& \leq \mathbb{P}[\textrm{at least } \ell \textrm{ vertices in } C_{\geq 1/20} \textrm{ are } 3/4\textrm{-underestimating at time } t ].
\end{align*}
This is straight-forward to see from the fact that if at least $\ell$ vertices in $C_{\geq 1/20}$ are underestimating at some stage, then for each of them there is a distinct variable $Y_i$ that evaluates to true.  We thus have that
\[
    \mathbb{P}\left[\bigcup_{0\leq t\leq m} \bigcap_{c \in C_{\geq 1/20}} \left\{\tilde{w}^t(\{ c \}, \overline{X}) \leq 3/4 \cdot |E_{H^t}(\{ c \}, \overline{X})|\right\}\right] \leq  \mathbb{P}\left[\bigcup_{I \subseteq [mn], |I| \geq \ell} \bigcap_{i \in I} Y_i \right]
\]

Finally, we can use a union bound and, since we can find a strict time-wise ordering of the indices in any set $I$, we can apply the chain rule to obtain
\begin{align*}
    \mathbb{P}\left[\bigcup_{I \subseteq [mn], |I| \geq \ell} \bigcap_{i \in I} Y_i \right]
   &\leq \sum_{I \subseteq [mn], |I| \geq \ell} \mathbb{P}\left[ \bigcap_{i \in I} Y_i \right] \\
   &= \sum_{i_1 < i_2 <\dots < i_{k}\subseteq [mn], k \geq \ell} \;\;\prod_{j = 1}^k \mathbb{P}\left[ Y_{i_j} | Y_{i_1}, Y_{i_2}, \dots, Y_{i_{j-1}} \right]
\end{align*}
and since each term $\mathbb{P}\left[ Y_{i_j} | Y_{i_1}, Y_{i_2}, \dots, Y_{i_{j-1}} \right]$ is upper bounded by $n^{-\frac{2^7(\alpha+1)\Delta_{\max}}{\Delta_{\min}\phi}} $ as shown before (here we can use the law of total probability to condition on a subset of the variables $Y_1,Y_2, \dots, Y_{i-1}$), we can finally obtain
\[
    \sum_{i_1 < i_2 <\dots < i_{k} \subseteq [mn], k \geq \ell} \;\;\prod_{j = 1}^k \mathbb{P}\left[ Y_{i_j} | Y_{i_1}, Y_{i_2}, \dots, Y_{i_{j-1}} \right]
   \leq {mn \choose \ell} \prod_{j = 1}^\ell n^{-\frac{2^7(\alpha+1)\Delta_{\max}}{\Delta_{\min}\phi}}.
\]
The last term can be upper bounded by $\left(\frac{emn}{\ell}\right)^\ell n^{-\frac{2^7(\alpha+1)\Delta_{\max}}{\Delta_{\min}\phi}\ell} $. To see that this is at most $n^{-\frac{2^6(\alpha+1)\Delta_{\max}}{\Delta_{\min}\phi}\ell} $, note that $\left(\frac{emn}{\ell}\right)^\ell \leq n^{5\ell}$ by \Cref{thm:folklore} while $\frac{2^7(\alpha+1)\Delta_{\max}}{\Delta_{\min}\phi}\ell > 10\ell$.
\end{proof}

It remains to combine the last two  claims to finish the proof of the lemma.

\begin{claim}
\label{clm:finalClaimLowerBound}
The probability that at any stage $t$, 
\[
|E_{H^t}(X, \overline{X})|/2 >  \tilde{w}^t(E_{H^t}(X, \overline{X}))
\]
is at most $n^{-10 \alpha k}$.
\end{claim}
\begin{proof}
We have by \Cref{clm:ifLemmaVioThenManyAlphaUnderestimating} that
\begin{align*}
    \mathbb{P}&\left[\bigcup_{t=0}^{m} \{|E_{H^t}(X, \overline{X})|/2 >  \tilde{w}^t(X, \overline{X})\}\right] \\
    &\leq \mathbb{P}\left[ \bigcup_{C \subseteq C_{\geq 1/20}, |C| \geq \eta} \{ \exists\mbox{ a stage }t \mbox{ where every vertex in } C \mbox{ is } 3/4\mbox{-underestimated}\}\right]\\
    &\leq \sum_{C \subseteq C_{\geq 1/20}, |C| \geq \eta} \mathbb{P}\left[ \exists\mbox{ a stage }t \mbox{ where every vertex in } C \mbox{ is } 3/4
    \mbox{-underestimated}\right]\\
    &\leq \sum_{C \subseteq C_{\geq 1/20}, |C| \geq \eta} n^{-\frac{2^6(\alpha+1)\Delta_{\max}}{\Delta_{\min}\phi}\eta} \leq {n \choose \eta} n^{-\frac{2^6(\alpha+1)\Delta_{\max}}{\Delta_{\min}\phi}\eta} \leq \left(\frac{en}{\eta}\right)^{\eta} \cdot n^{-\frac{2^6(\alpha+1)\Delta_{\max}}{\Delta_{\min}\phi}\eta} < n^{-10 \alpha k}
\end{align*}
where the second inequality stems from a union bound, we then use \Cref{clm:alphaUnderestimateSeldom} to derive the next inequality, and finally we use \Cref{thm:folklore} and plug in the value of $\eta = \phi k \Delta_{\min}/(8 \Delta_{\max})$.
\end{proof}
This completes the proof of the lemma.
\end{proof}

\subsubsection{A Simple Upper Bound on the Approximation}
\label{subsec:simpleUpperBoundForCutSparsifier}

In this section, we give an  upper bound on how much $\tilde{H}$ overestimates $H$. The central to our proof is the following lemma.

\begin{restatable}{lemma}{arbitrarySetUpperBound}
\label{clm:simpleUpperBound}
It holds for any stage $t$,  any set of edges $E' \subseteq E(H^0)$, and any $\ell \in \mathbb{R}^+$ that 
\[
    \mathbb{P}\left[\tilde{w}^t(E') \geq \ell \right] \leq \left(\frac{2e^2 \lg t \cdot |E'|}{\ell}\right)^{\ell \rho}
\]
where $\tilde{w}^t$ is the weight function of $\tilde{H}^t$ that assigns edges in $\tilde{H}^t$ weight $1 / \rho$ and edges not in $\tilde{H}^t$ weight $0$.
\end{restatable}

To explain why the lemma above is the key to our $O(\log n)$-approximation result shown in the next section, we apply the lemma in a straight forward manner to obtain the following upper bound result.

\begin{theorem}
\label{clm:simpleApprox}
At all stages $t$, we have that $\tilde{H}$ forms an $\tilde{O}\left(\frac{\Delta_{\max}}{\phi \Delta_{\min}}\right)$-approximate cut-sparsifier of $H$ with probability at least $1 - n^{-\alpha}$ for any constant $\alpha > 0$.
\end{theorem}
\begin{proof}
For any cut $(X, \overline{X})$, where $k = |X| \leq |\overline{X}|$, we use \Cref{clm:simpleUpperBound} on the edges $E_{H^0}(X, \overline{X})$, where we set $\ell = 4e^2 \cdot \lg n \cdot |E_{H^0}(X, \overline{X})|$ and obtain that
\begin{align} \label{eq:upperBoundOn0Version}
\mathbb{P}&\left[\tilde{w}^t(E_{H^0}(X, \overline{X})) \geq 4e^2 \cdot \lg n \cdot |E_{H^0}(X, \overline{X})| \right]
\leq 2^{-4e^2 \cdot \lg n \cdot |E_{H^0}(X, \overline{X})| \rho}
\end{align}
Further, we have that $X$ has volume at least $k \cdot \Delta_{\min}$ and since it forms the smaller side of the cut, we have by expansion of $H$ that $|E_{H}(X, \overline{X})| \geq \phi \cdot k \cdot \Delta_{\min}$ at any stage under consideration. Thus, 
\begin{align} \label{eq:upperBoundOnVersionSimpler}
    2^{-4e^2 \cdot \lg n \cdot |E_{H^0}(X, \overline{X})| \rho} \leq 2^{-4e^2 \cdot \lg(n) \phi \cdot k \cdot \Delta_{\min} \rho} \leq n^{-(\alpha+2)k}
\end{align}
because $\rho = \rhoValue$. Further, we have
\[
    \frac{|E_{H^0}(X, \overline{X})|}{|E_{H}^t(X, \overline{X})|} \leq \frac{\Delta_{\max} k}{\phi \Delta_{\min} k} =\frac{\Delta_{\max}}{\phi \Delta_{\min}}.
\]
because the $k$ vertices in $X$ have maximum degree $\Delta_{\max}$ which upper bounds the size of the cut and the minimum degree of each vertex is $\Delta_{\min}$ and a $\phi$-fraction of these incident edges is in the cut by the definition of conductance.

Thus, combining with   (\ref{eq:upperBoundOn0Version}) and (\ref{eq:upperBoundOnVersionSimpler}) it holds   with probability at least $n^{-(\alpha+2)k}$  that
\[
    \tilde{w}^t(E_{H^t}(X, \overline{X})) \leq  \frac{\Delta_{\max}}{\phi \Delta_{\min}} \cdot \tilde{w}^t(E_{H^0}(X, \overline{X})) \leq \frac{4e^2 \cdot \lg n \Delta_{\max}}{\phi \Delta_{\min}}  \cdot |E_{H^0}(X, \overline{X})|.
\]
It follows that if we use the analysis above for any cut and then use a union bound over all cuts and stages, we obtain
\begin{align*}
\mathbb{P}&\left[ \text{At some stage, } \tilde{H} \text{ is not a } \left(\frac{4e^2\Delta_{\max}}{\phi \Delta_{\min}} \lg n\right) \text{ cut-sparsifier of } H\right]\\
&= \mathbb{P}\left[\exists t, \exists (X, \overline{X}), \tilde{w}^t(E_{H^0}(X, \overline{X})| \geq 4e^2 \cdot \lg n \cdot |E_{H^0}(X, \overline{X})|\right] \\
&\leq \bigcup_{\substack{0 \leq t \leq m,\\ 1 \leq k \leq n/2,\\ X \subseteq V, k = |X|}} \mathbb{P}\left[\tilde{w}^t(E_{H^0}(X, \overline{X})) \geq 4e^2 \cdot \lg n \cdot |E_{H^0}(X, \overline{X})| \right] \\
&\leq n^{-\alpha}. \qedhere
\end{align*}
\end{proof}

Now we are ready to prove  \Cref{clm:simpleUpperBound}.

\begin{proof}[Proof of \Cref{clm:simpleUpperBound}]
Let us begin the proof by formalizing the random process that leads to the sparsifier $\tilde{H}^t$. %
The algorithm can be viewed  as repeatedly performing edge-experiments, each of which takes some edge $(u,v)$ and adds it to $S_u$ with probability $\rho$. Let $Z_0$ correspond to the first experiment performed by the algorithm,  $Z_1$ correspond to the second one, and etc.  We set $Z_i = 1$ if the corresponding edge $(u,v)$ was indeed sampled into $S_u$. Notice that,  regardless of the values of $Z_0, \ldots Z_{i-1}$, we have that the probability of $Z_i$ to evaluate to true is $\rho$.  Recall that time $t'$ refers to the state of the algorithm after $t'$ adversarial updates, and before $t'+1$ adversarial updates. Each experiment $Z_i$ is concerned with sampling some particular edge $(u,v)$ into a set $S_z$ for $z \in \{u,v\}$ and occurs at some particular time $t'$.

For convenience, we introduce the boolean random variables $Z_{z, (u,v)}^{t'}$ where $z \in \{u,v\}$ such that $Z_{z, (u,v)}^{t'}$ is true if and only if
\begin{itemize}
    \item there exists an experiment $Z_i$ that is conducted at stage $t'$ and concerned with sampling the edge $(u,v)$ into the set $S_z$, {\bf and}
    \item that experiment $Z_i$ is successful. 
\end{itemize}
With this set-up in place let us describe how to express the event that an edge $(u,v)$ is in $\tilde{H}^t$ using the notation defined so far: if the edge $(u,v)$ is in $\tilde{H}^t$ then we have that the edge $(u,v)$ is in the set $S_u$ \emph{or} in the set $S_v$ at stage $t$, i.e. $(u,v) \in S^t_z$ for some $z\in\{u,v\}$. Further, in order to be in such a set $S^t_z$ there must have been some experiment $Z_{z, (u,v)}^{t'}$ that was successful \emph{and} the set $S_z$ was not resampled thereafter, i.e. at no stage $t''$ after stage $t'$ and before (or including) stage $t$ was scheduled in $\mathcal{T}_{\schedule}(z)$.

Let us summarize our discussion formally as follows:
\begin{align*}
    \{ (u,v) \in \tilde{H}^t \} &= \{ (u,v) \in S_u^t\} \;\bigcup\; \{ (u,v) \in S_v^t\} \\&= \left\{ \bigcup_{t'= 0, z \in \{u,v\}}^{t} \left( Z_{z, (u,v)}^{t'} \; \land \not\exists t'' \in \mathcal{T}_{\schedule}(z), t' <
    t'' \leq t \right) \right\}.
\end{align*}
Here, we point out that the set $\mathcal{T}_{\schedule}(z)$ is a random set depending on the entire update sequence. However, as the key to our analysis is to order events time-wise, we henceforth argue about the event where $\mathcal{T}^{t'}_{\schedule}(z)$ is in place of $\mathcal{T}_{\schedule}(z)$, that is instead of considering the final update schedule $\mathcal{T}_{\schedule}(z)$ for vertex $z$, we consider only the schedule $\mathcal{T}^{t'}_{\schedule}(z)$ that the algorithm learnt during stage $t'$. It is not hard to observe that
\begin{align}\label{eq:subsetEventT}
   \left\{ \bigcup_{t'= 0, z \in \{u,v\}}^{t} \left( Z_{z, (u,v)}^{t'} \; \land \not\exists t'' \in \mathcal{T}_{\schedule}(z), t' <
    t'' \leq t \right) \right\}\\ 
    \subseteq \left\{ \bigcup_{t'= 0, z \in \{u,v\}}^{t} \left( Z_{z, (u,v)}^{t'} \; \land \not\exists t'' \in \mathcal{T}^{t'}_{\schedule}(z), t' <
    t'' \leq t \right) \right\}
\end{align}
since the set $\mathcal{T}^{t'}_{\schedule}(z) \subseteq \mathcal{T}_{\schedule}(z)$ and therefore a stage $t'$ might now be considered even if at a later stage $t''$ (still before or at $t$) new stages are added to $\mathcal{T}_{\schedule}(z)$.

Let us now focus on the second part of the event. To ease the discussion, let us introduce the concept of $u$ and $(u,v)$-relevance. 

\begin{definition}[$u$-relevant] For any vertex $u \in V$, and stage $t'$, we say that stage $t'$ is $u$-relevant if there exists no index $t'' \in \mathcal{T}^{t'}_{\schedule}(u)$ with $t' < t'' \leq t$.
\end{definition}
\begin{definition}[$(u,v)$-relevant]
For any edge $(u,v)$ and stage $t'$, we say that an experiment $Z_i$ that is evaluated at stage $t''$ and concerned with edge $(x,y)$ is $(u,v)$-\emph{relevant} if $(x,y) = (u,v)$, $t' = t''$ and $t''$ is $u$-\emph{relevant} or $v$-\emph{relevant}.
\end{definition}

The concept of relevance is crucial to our analysis, because an edge $(u,v)$ can only end up in $\tilde{H}^t$ as a result of some $(u,v)$-relevant experiment $Z_i$ at time $t' \leq t$; any experiment that is not $(u,v)$-relevant, either does not concern edge $(u,v)$, or it occurs at a non-relevant time $t'$, in which case both $S_u$ and $S_v$ will be resampled in the time interval $(t',t]$, so experiment $Z_i$ cannot be the one responsible for adding edge $(u,v)$ to up in $\tilde{H}^t$.

In the next claim, we argue that there is only a small number of $u$-relevant stages, and therefore only a small number of $(u,v)$-relevant experiments for each edge $(u,v)$.

\begin{restatable}{claim}{clmSchedule}
\label{clm:fewEvents}
For every vertex $u$, there are at most $\lg t$ stages that are $u$-\emph{relevant}. Moreover, for every edge $(u,v)$, there are at most $2\lg t$ experiments that are $(u,v)$-\emph{relevant}.
\end{restatable}
\begin{proof}
We have that for each \emph{relevant} stage $t'$ that there is no index $t'' \in \mathcal{T}^{t'}_{\schedule}(u)$ with $t' < t'' \leq t$. Thus, in order to create a \emph{relevant} stage $s' \leq t$ after $t'$, there has to be an adversarial update that \emph{touches} $u$ (either directly by deleting an edge incident to $u$ or indirectly by decreasing the degree of a neighbor of $u$), at some stage $s$, where $t' < s \leq s' \leq t$ such that $s'$ is added to $\mathcal{T}_{schedule}(u)$. 

But observe that then,   \Cref{alg:update} also adds stages $s, s+2^0, s+2^1, s+ 2^2, \dots$ to $\mathcal{T}_{\schedule}(u)$ during the same update procedure. This implies that
\[
s' \geq s + 2^{\lceil\lg (t - s) \rceil - 1} 
\]
since we have that the stage $s + 2^{\lceil\lg (t - s) \rceil - 1} \leq s + 2^{\lg (t - s)} \leq s + t - s = t$ occurs before or at $t$, so each stage before $s + 2^{\lceil\lg (t - s) \rceil - 1}$ cannot satisfy the condition. 

But this implies that $s' - t' \geq s + 2^{\lceil\lg t - s \rceil - 1} - t' \geq s + (t - s)/2 - t' \geq (t-t')/2$ which in turn implies that between any two \emph{relevant} stages, the distance between the stage and $t$ halves. But this upper bounds the number of \emph{relevant} stages by $\lg t$. The second statement of the claim follows immediately from the first.
\end{proof}

Equipped with the developed notation and \Cref{clm:fewEvents}, we can now lay out the strategy to finish our proof: we next prove that for every set $F \subseteq E'$ of edges, the probability that \emph{all} edges $F$ are in $\tilde{H}^t$ is small. Once this proof is established, we can prove our lemma by taking a union bound over all large sets of edges, and show that   no large subset of $E'$ is contained in $\tilde{H}^t$ with high probability. 

We start this proof by defining the concept of $F$-relevance which is a straight-forward extension of $(u,v)$-relevance.

\begin{definition}[$F$-relevant] For any set $F \subseteq E'$, we say that an experiment $Z_i$ is $F$-\emph{relevant} if $Z_i$ is \emph{relevant} for some edge $(u,v) \in F$.
\end{definition}
We make two crucial observations. The first observation is that the $F$-relevance of an experiment $Z_i$ depends only on events that occurred before the experiment $Z_i$ was carried out. This follows since we update the schedule $T_{\schedule}(u)$ for every $u$ before we carry out new experiments at each stage. Thus, when the experiment is conducted, we can already check whether it is $F$-relevant.

\begin{observation}
\label{obs:dependentOnPast}
Whether $Z_i$ is $F$-relevant depends on the variables $Z_0, Z_1, \dots, Z_{i-1}$ and the updates issued by the adversary up to stage $t'$.
\end{observation}

But since $F$-relevance only depends on information that was accessible before the experiment $Z_i$ was evaluated, it is straight-forward to conclude that conditioning on $F$-relevance does not change the probability of $Z_i$ which is set to true with probability $\rho$ independent of past events. 

\begin{observation}
\label{obs:independentOfPast}
For every experiment $Z_i$,
\[
    \mathbb{P}[Z_i \;|\; Z_i \mbox{ is } F \mbox{-relevant}] = \rho
\]
\end{observation}

Then for any fixed set $F \subseteq E'$, we define a new stochastic process $X_1, X_2, \dots$ where $X_i$ is the random indicator variable for the event that
\begin{itemize}
    \item the adversary performs at least $i$ $F$-relevant experiments, \emph{and}
    \item the $i$-th $F$-relevant experiment succeeds.
\end{itemize}

We now introduce another key observation, which follows from the two above observations. Here, we use the fact if we reach the $i$-th $F$-relevant experiment $Z_{j}$ for some $j$, then the variables $X_1, X_2, \dots, X_{i-1}$ have already been determined, and since each experiment uses fresh randomness, $Z_j$ will be true with probability $\rho$ regardless of the outcomes of $X_1, X_2, \dots, X_{i-1}$.

\begin{observation}
\label{clm:probForXi}
We have for any $i$ that $\mathbb{P}[X_i = 1 | X_1, X_2, \dots, X_{i-1}] \leq \rho$.
\end{observation}

We note a subtlety in the observation above. The outcomes $X_1, \ldots X_{i-1}$ can influence $X_i$, because it is possible that the adversary's strategy is such that if $X_1, \ldots X_{i-1}$ are all true, then there simply will not be an $i$-th $F$-relevant experiment, so $X_i$ will necessarily be $0$. That is why we could not have strict equality in the observation above. But what we do know is that, even if the $i$-th $F$-relevant experiment does occur, its probability of success is at most $\rho$.

Now we can prove the following claim. 

\begin{claim}
\label{clm:exactlyFinH}
For any set $F \subseteq E'$, the probability that every edge in $F$ is contained in $\tilde{H}^t$ is at most $(2e\lg t \cdot \rho)^{|F|}$. 
\end{claim}
\begin{proof}
Letting $X_1, X_2, \dots$ refer to stochastic process as before, we observe that 
\[
\{ F \subseteq \tilde{H}^t\} \subseteq \left\{\sum_i X_i \geq |F|\right\}
\]
because in order to have $|F|$ edges sampled there have to be at least $|F|$ successful $F$-relevant experiments.

Now we observe that, since each experiment that is $F$-relevant has to be $(u,v)$-relevant for $(u,v) \in F$ and there are at most $2|F|\lg t$ such experiments by \Cref{clm:fewEvents}, it holds  that $X_{i} = 0$ for $i > 2|F| \lg t$. 

To calculate the probability of the event $\left\{\sum_i X_i \geq |F|\right\}$ which gives and upper bound on the desired probability, we can consider all configurations of indices, i.e. all sets of at least $|F|$ indices that are distinct and in the range $[1, 2|F| \lg t]$ and calculate the probability that \emph{all} variables $X_i$ with the chosen indices realize to $1$. Formally,
\[
\mathbb{P}\left[\sum_i X_i \geq |F|\right] = \mathbb{P}\left[ \bigcup_{\substack{k_1 < k_2 < \dots < k_{|F|} \\ \subseteq [1, 2|F|\lg t]}} \bigcap_{\ell = 1}^{|F|} X_{k_\ell} = 1\right].
\]
By a union bound, it is straight-forward to obtain that 
\[
\mathbb{P}\left[ \bigcup_{\substack{k_1 < k_2 < \dots < k_{|F|} \\ \subseteq [1, 2|F|\lg t]}} \bigcap_{\ell = 1}^{|F|} X_{k_\ell} = 1\right] \leq  \sum_{\substack{k_1 < k_2 < \dots < k_{|F|} \\ \subseteq [1, 2|F|\lg t]}} \mathbb{P}\left[\bigcap_{\ell = 1}^{|F|} X_{k_\ell} = 1\right].
\]
Finally,     using the chain rule to expand the right-hand terms and  using \Cref{clm:probForXi} to the product terms (combined with a straight-forward application of the law of total probability), we obtain the final bound
\begin{align*}
\sum_{\substack{ k_1 < k_2 < \dots < k_{|F|} \\ \subseteq [1, 2|F|\lg t]}} \mathbb{P}\left[\bigcap_{\ell = 1}^{|F|} X_{k_\ell} = 1\right] &= \sum_{\substack{ k_1 < k_2 < \dots < k_{|F|} \\ \subseteq [1, 2|F|\lg t]}} \prod_{\ell = 1}^{|F|} \mathbb{P}\left[ X_{k_\ell} = 1 \;|\; X_{k_1}, X_{k_2}, \dots, X_{k_{\ell -1}}\right]\\
&\leq \sum_{\substack{ k_1 < k_2 < \dots < k_{|F|} \\ \subseteq [1, 2|F|\lg t]}} \prod_{\ell = 1}^{|F|} \rho = {2|F| \lg t \choose |F|} \cdot \rho^{|F|}
\end{align*}
The bound in our lemma is then obtained by using the Folklore \Cref{thm:folklore}.
\end{proof}

We can thus finally bound the probability that we were originally interested in by taking a union bound over all sets $F \subseteq E'$ of size $\ell \rho$ (since each edge is assigned weight $1/\rho$ in $\tilde{H}$) and obtain 
\[
    \mathbb{P}[|\tilde{w}^t(E') \geq \ell] \leq \mathbb{P}\left[\bigcup_{F \subseteq E', |F| \geq \ell \rho} F \subseteq \tilde{H}^t\right] \leq \left(\frac{2e |E'|}{\ell \rho}\right)^{\ell \rho} (e \lg t \cdot \rho )^{\ell \rho} = \left(\frac{2e^2 \lg t \cdot |E'|}{\ell}\right)^{\ell \rho}
\]
using a union bound, \Cref{clm:exactlyFinH} and again the folklore result ${n \choose k} \leq \left(\frac{e \cdot n}{k}\right)^k$.
\end{proof}

\subsubsection{An Improved Upper Bound on the Approximation}
\label{subsec:improvingTheUpperBoundForCutSparsifier}

Finally, we discuss how to refine our analysis to achieve an $O(\log n)$-approximation. We note that in the previous upper bound proof, we have at no point exploited that we resample neighbors once a vertex has lost $\zeta$ of its incident edges, as described in \Cref{lne:ifDegreeHigh} in \Cref{alg:update}. This is our main tool to improve the approximation ratio, since it gives us a way to ``restart'' a part of the cut sparsifier where the initial sample is no longer representative of the underlying edge set. The main result of this section is summarised as follows:

\begin{restatable}{theorem}{mainApproxUpperBound}
\label{thm:upperBoundCutSparsifier}
At all stages $t$, we have that $\tilde{H}$ is an $O(\log n)$-approximate cut sparsifier of $H$ with high probability, i.e. with probability $1 - n^{-\alpha}$.
\end{restatable}

We can derive the above theorem by refining the proof of \Cref{clm:simpleUpperBound} to obtain the following result and then by taking a union bound over all stages and cuts, similarly to the way we derived \Cref{clm:simpleApprox}.

\begin{lemma}\label{lma:keyLemmaUpperBoundCutSparsifier}
For any cut $(X,\overline{X})$, $k = |X| \leq |\overline{X}|$, at any stage $t$, we have that
\begin{equation}
\label{eq:OversampleRefinedWithLowProb}
    \mathbb{P}\left[\tilde{w}^t(E_{H^t}(X, \overline{X})) \geq 8e \lg n |E_{H^t}(X, \overline{X})|\right] \leq n^{-8k(\alpha+1)}.
\end{equation}
\end{lemma}
\begin{proof}
Now, in the analysis, we shift our perspective from analyzing events on edges toward analyzing events related to vertices in $X$. To formalize these events, for any vertex $c \in X$, we use the set $T_{DegreeUpdate}(c)$ that accounts for how often the degree of $c$ is decreased by a significant fraction, throughout the course of the algorithm.

More precisely, whenever the degree of $c$ has its degree drop by $\zeta$ edges, we enter the if-statement in \Cref{lne:ifDegreeHigh} of \Cref{alg:update}, add the current stage $t'$ to $T_{DegreeUpdate}(c)$ in  \Cref{lne:degreeUpdateT} and schedule a resampling event for every vertex $y$ in $\mathcal{N}_H(c)$ for $t', t'+2^0, t'+2^1, \dots$. Let us denote by $T_{DegreeUpdate}^{t'}(c)$ the set $T_{DegreeUpdate}(c)$ at the end of stage $t'$. We define the notion of phases.

\begin{definition}
We say that a vertex $c \in V$ is in \emph{phase} $j$ at stage $t''$, if $|T_{DegreeUpdate}^{t''}(c)| = j$. Informally, $c$ is in phase $j$ at stage $t''$ if up to the end of stage $t''$, there were exactly $j$ stages where a distinct index was added to $T_{DegreeUpdate}(c)$.
\end{definition}

Using this new concept, let us introduce new notions of relevance for an experiment that incorporates the perspective on vertices and the observation that the edge was sampled within the last phase of the vertex. We therefore introduce a new notion on relevance that focuses on a vertex $c$ and a phase $j$ that the vertex is in. We say an experiment is $(c,j)$-relevant, if it is concerned with an edge that is incident to $c$ {\bf and} in the cut $(X, \overline{X})$ {\bf while} the vertex $c$ is in the $j^{th}$ phase.

\begin{definition}[$(c, j)$-relevant] For any vertex $c \in X$, and integer $j$, we say that an experiment $Z_i$ conducted at stage $t'$ concerned with edge $(u,v)$ is $(c, j)$-relevant if
\begin{itemize}
    \item $(u,v) \in E_{H^{t'}}(\{c\}, \overline{X})$, \emph{and}
    \item $Z_{(u,v)}^{t'}$ is relevant for $(u,v)$, \emph{and}
    \item $|T_{DegreeUpdate}^{t'}(c)| = j$.
\end{itemize}
\end{definition}

We also generalize this notion to any subset $C$ of $X$ where we use function $\iota$ to encode the phase of concern for each vertex $c \in C$.

\begin{definition}[$(C, \iota)$-relevant] For any set $C \subseteq X$, and function $\iota : C \mapsto \mathbb{N}$ that maps any vertex $c \in C$ to an integer $\iota(c)$, we say that an experiment $Z_i$ is $(C, \iota)$-relevant if there is a vertex $c \in C$ such that $Z_i$ is $(c,\iota(c))$-relevant.
\end{definition}

As before in \Cref{subsec:simpleUpperBoundForCutSparsifier}, we also make some observations regarding each variable $Z_i$ in regard to the new notions of relevance. 
We first observe that for an experiment $Z_i$, we can test $(c,j)$-relevance using the information from experiments $Z_0, Z_1, \dots, Z_{i-1}$ and the adversarial updates up to the stage during which $Z_i$ is executed. To see this observe that from the adversarial updates up to the current stage, we can construct all schedules $\mathcal{T}_{Schedule}(u)$ for any $u \in V$ and therefore check which edge $(x,y)$ the experiment $Z_i$ is concerned with before the experiment is conducted. Thus, we can also check whether the edge is in $E_{H^{t'}}(\{c\}, \overline{X})$ straight-forwardly. By \Cref{obs:dependentOnPast}, we also have that $(u,v)$-relevance can be determined from this information. Finally, we point out that, since the experiment is conducted at the stage $t'$ after the schedule for stage $t'$ was constructed in \Cref{alg:update}, we can certainly determine whether $|T_{DegreeUpdate}^{t'}(c)| = j$. We summarize our result as follows.

\begin{observation}
\label{obs:dependentOnPastCJRelevant}
Whether $Z_i$ conducted at stage $t'$ is $(c,j)$-relevant depends only on the variables $Z_0, Z_1, \dots, Z_{i-1}$ and the updates issued by the adversary up to stage $t'$.
\end{observation}

Further, similarly to \Cref{obs:independentOfPast}, we can establish that an experiment conditioned on the fact that it is $(c,j)$-relevant has success probability $\rho$.  
This follows since the outcome, which is determined by  fresh random bits,  evaluates to be true with probability exactly $\rho$, and the information encoded by the conditioning statement, which concerns past events, cannot encode the outcome of the sampling step. 

\begin{observation}
\label{obs:independentOfPastCJRelevant}
We have for every experiment $Z_i$ that
\[
    \mathbb{P}[Z_i \;|\; Z_i \mbox{ is } (c,j) \mbox{-relevant}] = \rho
\]
\end{observation}

Next, let us consider any fixed function $\iota$ that maps every vertex in $X$ to some arbitrary integer in $[0, m]$. We then define a new stochastic process $Y_1^{\iota}, Y_2^{\iota}, \dots$ where $Y_i^{\iota}$ is the random indicator variable for the event that
\begin{itemize}
    \item the adversary performs at least $i$ $(X, \iota)$-relevant experiments, \emph{and}
    \item the $i$-th $(X, \iota)$-relevant experiment succeeds.
\end{itemize}

By the previous two observations, we immediately derive the following insight. 

\begin{observation}
\label{obs:probForXiCJRelevant}
We have for any $i$ that $\mathbb{P}[Y_i = 1 | Y_1, Y_2, \dots, Y_{i-1}] \leq \rho$.
\end{observation}

Using this set-up, let us define some events of particular interest. Let $\iota$ be an arbitrary function as defined above and $\ell$ an integer, then we denote by $\mathcal{E}^{\iota}_\ell$ the event that
\begin{itemize}
    \item $\sum_i |Y_i^{\iota}| \geq \ell$, \emph{and}
    \item for every vertex $c \in X$, we have that $c$ is in phase $\iota(c)$ at stage $t$.
\end{itemize}
Further, define $\mathcal{E}_\ell = \bigcup_{\iota} \mathcal{E}^{\iota}_\ell$, i.e. the event that there exists a function $\iota$ such that $\mathcal{E}^{\iota}_\ell$ occurs.

\begin{claim}\label{clm:numEdgesInSparsifierReducedToIndVar}
For any $\ell \geq 0$, the event $\{|E_{\tilde{H}^t}(X, \overline{X})| \geq \ell\}$ is a subset of the event $\mathcal{E}_\ell$.
\end{claim}
\begin{proof}
Observe first that at time $t$ there exists a unique function $\iota$ that satisfies the 
second requirement of event $\mathcal{E}^{\iota}_\ell$, namely, that for every vertex $c \in X$, we have that $c$ is in phase $\iota(c)$ at stage $t$. Further, we observe that when a vertex $c$ enters the last phase $\iota(c)$, then every edge in $E_{\tilde{H}^t}(\{c\}, \overline{X})$ is resampled as can be verified from inspecting \Cref{alg:update}. Thus, the experiment that sampled an edge $(u,v) \in E_{\tilde{H}^t}(\{c\}, \overline{X})$ into the cut-sparsifier is $(c, \iota(c))$-relevant. Thus, every edge sampled into the sparsifier in the cut $(X, \overline{X})$ is $(X, \iota)$-relevant. It is clear that in order to have at least $\ell$ edges in the cut $(X, \overline{X})$ sampled into the sparsifier $\tilde{H}^t$, we must have at least the same number of successful $(X, \iota)$-relevant experiments, and therefore at least that many indicator variables $Y_i^{\iota}$ that evaluate to $1$. 
\end{proof}

Next let us prove the following claim.

\begin{claim}\label{clm:upperBoundForTauEvent}
It holds for any function $\iota$ and integer $\eta$  that 
\[
\mathbb{P}[\mathcal{E}^{\iota}_\ell \cap \{ |E_{H^t}(X, \overline{X})| = \eta\}] \leq \left(\frac{4e |E_{H^t}(X, \overline{X})| \lg t}{\ell}\right)^\ell \rho^\ell.
\]
\end{claim}
\begin{proof}
We claim that if there are more than $\eta' = 4\eta \cdot \lg t$ experiments that are $(X, \iota)$-relevant, then we cannot have that every vertex $c$ is in phase $\iota(c)$ at stage $t$ and that $|E_{H^t}(X, \overline{X})| = \eta$. Thus, in this case, the probability of the event is $0$. 

To see this, we first observe that at any stage, the number of edges in the cut $(X, \overline{X})$ is at least $\phi \Delta_{\min}k$ since the volume of $X$ is at least $\Delta_{\min}k$. Now, let us consider for every vertex $c \in X$, the edges in $E_{H}(\{c\}, \overline{X})$ at the stage where the vertex enters phase $\iota(c)$. We observe that only edges that are present at that stage  can be sampled in an experiment that is $(c, \iota(c))$-relevant, by definition. Since we are still in phase $\iota(c)$ at stage $t$ and  the number of vertices in $X$ is at most $k$, we conclude that the number of such edges is at most $|E_{H^t}(X, \overline{X})| + k \zeta$.

By our previous observation and the fact that $\zeta = \phi \Delta_{\min}$, we thus have that at most $2|E_{H^t}(C_{\geq \epsilon}, \overline{X})|$ such edges exists. By \Cref{clm:fewEvents}, we further have that there can be at most $2\lg t$ experiments for every such an edge $(u,v)$ that are $(X, \iota)$-relevant, since every experiment that is $(X, \iota)$ relevant is $(u,v)$-relevant. 
Thus, the total number of $(X, \iota)$-relevant experiments can be at most $\eta' = 4\eta \cdot \lg t$,
since   otherwise the event   certainly does not occur in which case the upper bound is trivially given. Otherwise, we have that
\[
\mathbb{P}[\mathcal{E}^{\iota}_\ell \cap \{ |E_{H^t}(X, \overline{X})| = \eta\}] \leq \mathbb{P}[ \sum_{i=1}^{\eta'} |Y_i^{\iota}| \geq \ell].
\]
Following the analysis in \Cref{clm:exactlyFinH}, we obtain
\begin{align*}
    \mathbb{P}[ \sum_{i=1}^{\eta'} |Y_i^{\iota}| \geq \ell] &= \mathbb{P}\left[ \bigcup_{ \substack{k_1 < k_2 < \dots < k_{\ell} \\ \subseteq [1, \eta']}} \bigcap_{\ell = 1}^{\ell} Y_{k_\ell}^{\iota} = 1\right] \\
    &\leq \sum_{ \substack{k_1 < k_2 < \dots < k_{\ell} \\ \subseteq [1, \eta']}} \prod_{\ell = 1}^{\ell} \mathbb{P}[Y_{k_\ell}^{\iota} = 1 \; | Y^{\iota}_{k_1}, Y^{\iota}_{k_2}, \dots, Y^{\iota}_{k_{\ell-1}} \;]\\
    & \leq \left(\frac{e \eta'}{\ell}\right)^\ell \rho^\ell
\end{align*}
where we use the chain rule, a union bound and \Cref{obs:probForXiCJRelevant} in the first inequality and \Cref{thm:folklore} in the second. We derive the final bound by plugging in the value of $\eta'$.
\end{proof}

This claim enables us to finish the proof of the lemma. We set $\ell' = 16e \lg n |E_{H^t}(X, \overline{X})|$ and observe that 
\begin{align*}
    \mathbb{P}\left[\tilde{w}^t(E_{H^t}(X, \overline{X})) \geq \ell' \right] &= \mathbb{P}\left[|E_{\tilde{H}^t}(X, \overline{X})| \geq \ell' \rho \right]\\
    &\leq \mathbb{P}\left[\bigcup_{\eta \in [\phi k\Delta_{\min}, k\Delta_{\max}]} \mathcal{E}_{\ell' \rho} \; \cap \; \{|E_{H^t}(X, \overline{X})| = \eta\}\right]\\
    &= \mathbb{P}\left[\bigcup_{\eta \in [\phi k\Delta_{\min}, k\Delta_{\max}]} (\cup_{\iota} \mathcal{E}_{\ell' \rho}^{\iota}) \; \cap \; \{|E_{H^t}(X, \overline{X})| = \eta\}\right]\\
    &= \mathbb{P}\left[\bigcup_{\substack{\eta \in [\phi k\Delta_{\min}, k\Delta_{\max}] \\ \text{ any choice of } \iota}} \mathcal{E}_{\ell' \rho}^{\iota} \; \cap \; \{|E_{H^t}(X, \overline{X})| = \eta\}\right]
\end{align*}
where we first use that the edge weight in the cut sparsifier is $1/\rho$ and derive the inequality by \Cref{clm:numEdgesInSparsifierReducedToIndVar} and the insight that $|E_{H^t}(X, \overline{X})|$ has to be of size $\eta \in [\phi k\Delta_{\min}, k\Delta_{\max}]$ for some $\eta$ since each of the $k$ vertices in $X$ has maximum degree $\Delta_{\max}$ and minimum degree $\Delta_{\min}$ where a $\phi$-fraction of the total edges is in the cut $(X, \overline{X})$ since we have conductance at least $\phi$ and by \cref{clm:concentrationOfEdgesAtCertainVertices}. This last insight allows us to invoke the law of total probability to derive that summing over the values of $\eta$ we include every possible event. Finally, we expand along the definition of $\mathcal{E}_\ell$.

We then apply a union bound and use \Cref{clm:upperBoundForTauEvent} to derive
\begin{align*}
\mathbb{P}&\left[\bigcup_{\substack{\eta \in [\phi k\Delta_{\min}, k\Delta_{\max}] \\ \text{ any choice of } \iota}} \mathcal{E}_{\ell' \rho}^{\iota} \; \cap \; \{|E_{H^t}(X, \overline{X})| = \eta\}\right] 
\\
&\leq \sum_{\substack{\eta \in [\phi k\Delta_{\min}, k\Delta_{\max}] \\ \text{ any choice of } \iota}}  \mathbb{P}\left[\mathcal{E}_{\ell' \rho}^{\iota} \; \cap \; \{|E_{H^t}(X, \overline{X})| = \eta\}\right]\\
&\leq  \sum_{\substack{\eta \in [\phi k\Delta_{\min}, k\Delta_{\max}] \\ \text{ any choice of } \iota}} \left(\frac{4e |E_{H^t}(X, \overline{X})| \lg t}{\ell' \rho}\right)^{\ell' \rho} \rho^{\ell' \rho} \\
&\leq m^{k+1}  \left(\frac{4e  |E_{H^t}(X, \overline{X})| \lg t}{\ell' \rho}\right)^{\ell' \rho} \rho^{\ell' \rho} \leq m^{k+1}  \left(\frac{4e  |E_{H^t}(X, \overline{X})| \lg t}{\ell' }\right)^{\ell' \rho} 
\end{align*}
where we use in the second last step that there are at most $m$ stages, and therefore every vertex in $X$ can be in at most $m$ phases. There are at most $k$ vertices in $X$, so there are at most $m^k$ maps $\iota$. Further there are at most $m$ choices of $\eta$. 

We remind the reader that $\ell' = 16e \lg n |E_{H^t}(X, \overline{X})|$ and $\rho = \rhoValue$. Thus, we finally obtain
\begin{align*}
m^{k+1}  \left(\frac{4e  |E_{H^t}(X, \overline{X})| \lg t}{\ell' }\right)^{\ell' \rho} &\leq m^{k+1}  \left(\frac{8e  |E_{H^t}(X, \overline{X})| \lg n }{16e \lg n |E_{H^t}(X, \overline{X})|}\right)^{\ell' \rho}\\
&\leq m^{k+1} 2^{-\ell' \rho} \leq m^{k+1} 2^{-16e k \lg n (\alpha+2)} \\
&\leq m^{k+1} n^{-8e k(\alpha+2)} \leq n^{-16k(\alpha+2)}.
\end{align*}
and we finally use that $|E_{H^t}(X, \overline{X})| \geq k\Delta_{\min}\phi$, and since $\rho \geq 1/\Delta_{\min}$ we have $\ell' \rho \geq k \cdot 8e \lg t $. This implies the desired result.
\end{proof}

\subsubsection{Putting it all together}
\label{subsec:worstCaseForDecrCutSparsifier}

\paragraph{An Algorithm for Worst-Case Update Time.} Let us first consider to use \Cref{alg:update2} in place of \Cref{alg:update}. The difference between the two algorithms is constituted in the change that before, a vertex $z \in \{u,v\}$ would enter the for-each loop every $\zeta$ times an incident edge was deleted and schedule a vertex update for each of its neighbors (for the stages $t, t + 2^0, t + 2^1, t + 2^2, t + 2^3, \dots$). 

In the new algorithm, after every edge deletion we reschedule a few neighbors of $z$ only (again for the stages $t, t + 2^0, t + 2^1, t + 2^2, t + 2^3, \dots$). Thus, instead of spending all time for resampling at once, we spread updates in round-robin scheduling fashion. After $\zeta$ edges incident to a vertex are removed, it is clear that all its neighbors have  been rescheduled. Thus, since the last point in which all neighbors were resampled, at most $\zeta$ edges have been deleted incident to $z$. Carefully studying the proof of \Cref{clm:upperBoundForTauEvent}, the upper bound still applies to this algorithm and the same is true for the lower bound.

\begin{algorithm2e}
\caption{Handling an Edge Deletion.}
\label{alg:update2}
\SetKwProg{procedure}{Procedure}{}{}
\SetKwFunction{sample}{SampleVertex}
\SetKwFunction{update}{EdgeDeletion}
\procedure{\update{$(u,v), t$}}{
    \ForEach{$z \in \{u,v\}$}{
        Add $t, t + 2^0, t + 2^1, t + 2^2, t + 2^3, \dots$ to $T_{\schedule}(z)$\;
        $i \gets \deg_H(z) \mod \zeta$.\;
        \ForEach(\label{lne:scheduleNewUpdates}){$j \in [0, \lceil\Delta_{\max}/\zeta \rceil]$}
        {
            Let $y$ be the $(i \cdot \lfloor\Delta_{\max}/\zeta \rfloor + j)^{th}$ vertex in $\mathcal{N}(z)$\;
            Add $t, t + 2^0, t + 2^1, t + 2^2, t + 2^3, \dots$ to $T_{\schedule}(y)$ \label{lne2:neighborInducedStages}
        }
        
    }
    \ForEach(\label{lne2:foreachSchedule}){$v \in V$ where $t \in T_{\schedule}(v)$}{
        $S_v \gets$ \sample{$v$, $H$, $\rho$}\;
    }
}
\end{algorithm2e}

\paragraph{Analyzing Running Time and Size.} Further, we now have that at each stage there are at most $O(\frac{\Delta_{\max}\log n}{\zeta}) = O(\frac{\Delta_{\max}\log n}{\Delta_{\min}\phi})$ vertices scheduled for resampling, since at each stage $t$, the algorithm schedules at most $O(\frac{\Delta_{\max}}{\Delta_{\min}\phi})$ for any given stage $t, t+2^0, t+2^1, t+2^2, \dots$, one for each iteration of the foreach-loop in \Cref{lne:scheduleNewUpdates}. But since all updates scheduled at a stage $t'$ were issued during stages $t', t' - 2^0, t' - 2^1, t' - 2^2, \dots$ the upper bound applies. Finally, we observe that each invocation of $\textsc{SampleVertex}(v,H,\rho)$ can be implemented in worst-case time $O(\rho \Delta_{\max} \log n)$ by \Cref{thm:fast sampling} where $\rho = \rhoValue$. Thus, the total worst-case update time is 
\[
O\left(\frac{\Delta_{\max}\log n}{\Delta_{\min}\phi}\right) \cdot O(\rho \Delta_{\max} \log n) = O\left(\left(\frac{\Delta_{\max}\log n}{\Delta_{\min}\phi}\right)^3 \right).
\]

We also point out that it is straight-forward by a simple Chernoff bound to derive that for every vertex $u \in V(H)$, every set $S_u$ at any stage is of size $O(\rho \Delta_{\max} \log n)$ w.h.p. But since there are $|V(H)|$ vertices in $H$, the total size of the cut-sparsifier is bounded by $O\left(n\left(\frac{\Delta_{\max}\log n}{\Delta_{\min}\phi}\right)^2 \right)$ with high probability.

\paragraph{Approximation Ratio and Success Probability.} Since further upper and lower bounds from \Cref{lma:lowerBoundAdaptiveCutSparsifier} and \Cref{thm:upperBoundCutSparsifier} still  hold, we can multiply the weights of $\tilde{H}$ and derive \Cref{thm:mainDecrCutSparsifier}. We observe that the probability of correctness of at least $1-n^{-\alpha}$ can be derived by replacing $\alpha$ by $2\alpha$ in the value of $\rho$ which does not affect the asymptotic udpate time.

\paragraph{An Extension to Low Approximation Ratios.} Finally, we also point out that by resampling a vertex $u$ after an edge $(u,v)$ was deleted at stage $t$, at stages $t$ and $t+(1+\epsilon)^i$ for every $i$, where we can choose $\epsilon$ very small, we can improve the approximation ratio of the cut sparsifier since we can prove that there are $O(\frac{1}{\epsilon}\lg t)$ many $u$-relevant stages by following the proof of claim. This in turn gives an  $O(\log_{1/\epsilon} t)$-approximation ratio for our cut-sparsifier by following the remaining proofs. The worst-case update time of this modified algorithm is $O\left(\left(\frac{\Delta_{\max}\lg t}{\Delta_{\min}\phi \epsilon}\right)^3 \right)$. We thus derive the following theorem by setting $\epsilon = n^{-3/k}$ which generalizes  \Cref{thm:mainDecrCutSparsifier}.

\begin{theorem}\label{thm:generalStatementCutSparsifier}
Given a decremental unweighted graph $H$, and fixed values $\frac{1}{\phi} \in (1, n), \Delta_{\max} \geq \Delta_{\min} \geq \frac{80 \log n}{\phi}, 1 < k < \log n$, such that the initial graph $H$ has maximum degree $\Delta_{\max}$. Then, there exists an algorithm that maintains $\tilde{H} \subseteq H$ along with a weight function $\tilde{w}$ such that
\begin{itemize}
    \item At any stage, where $H$ has $\Delta_{min} \leq \min\deg(H)$ and $\Phi(H) \geq \phi$, we have that $\tilde{H}$ weighted by $\tilde{w}$ is a $O(k)$-approximate cut sparsifier of $H$, and
    \item the graph $\tilde{H}$ has at most $\tilde{O}\left(|V(H)| \left(\frac{\Delta_{max}\log n}{ \phi \Delta_{min}}\right)^2\right)$ edges.
\end{itemize}
The algorithm can be initialized in $O(m)$ time and has worst-case update time $\tilde{O}\left(\left(\frac{\Delta_{max}\log n}{\Delta_{min}\phi}\right)^3 n^{1/k} \right)$. The algorithm runs correctly with high probability, i.e. with probability $1 - n^{-\alpha}$ for any fixed constant $\alpha > 0$.
\end{theorem}

\subsubsection{Extending the Algorithm for ``pruned $\phi$-sub-expanders"}

So far, we have established \Cref{thm:mainDecrCutSparsifier}, i.e. we have given a decremental algorithm $\mathcal{A}$ which can maintain a cut-sparsifier on a almost-uniform-degree  $\phi$-expander. Thus, algorithm $\mathcal{A}$ satisfies \Cref{def:amortized:decremental_algorithm}, and it is therefore straight-forward to extend it to an amortized algorithm for general graphs which is part of the next section. 

However, this will not be sufficient to obtain an algorithm for general graphs with low worst-case update time. In order to use the reduction from \Cref{thm:reduction_worst_case}, the algorithm $\mathcal{A}$ has to work on pruned $\phi$-sub-expanders as described in
\Cref{def:worst-case:decremental_algorithm}. In this section, we show that our algorithm $\mathcal{A}$ already satisfies \Cref{def:worst-case:decremental_algorithm} if it just deletes the edge set $P_t$ at stage $t$ from $G$ and adds it to $\tilde{H}$ which only incurs a subpolynomial increase in the running time and a small additive term in the size of the sparsifier. We prove the following lemma.

\begin{lemma}\label{lma:WorstCaseBlackBoxCutSparsifier}
Given a decremental unweighted graph $H$ and an increasing edge set $P$ (where $P$ is initially empty and grows at stage $t$ by $P_t$ of size at most $2^{O(\sqrt{\log n})}$), and fixed values $\frac{1}{\phi} \in (1, n), \Delta_{\max} \geq \Delta_{\min} \geq \frac{80 \log n}{\phi}, 1 < k < \log n$, such that the initial graph $H$ has maximum degree $\Delta_{\max}$. Then, there exists an algorithm that maintains $\tilde{H} \subseteq H$ along with a weight function $\tilde{w}$ such that
\begin{itemize}
    \item At any stage, where $H$ has $\Delta_{\min} \leq \min\deg(H)$ and $\Phi(H) \geq \phi$, we have that $\tilde{H}$ weighted by $\tilde{w}$ is a $O(k)$-approximate cut sparsifier of $H$, and
    \item the graph $\tilde{H}$ has at most $\tilde{O}\left(|V(H)| \left(\frac{\Delta_{\max}\log n}{ \phi \Delta_{\min}}\right)^2 + |P|\right)$ edges.
\end{itemize}
The algorithm can be initialized in $O(m)$ time and has worst-case update time 
$$
2^{O(\sqrt{\log n})} \cdot \tilde{O}\left(\left(\frac{\Delta_{\max}\log n}{\Delta_{\min}\phi}\right)^3 n^{1/k} \right).
$$
The algorithm runs correctly with high probability, i.e. with probability $1 - n^{-\alpha}$ for any fixed constant $\alpha > 0$.
\end{lemma}
\begin{proof}
Let us consider that the algorithm from \Cref{thm:mainDecrCutSparsifier} is run with the above parameters and processes at each stage $t$ the adversarial edge deletion to $H$ \emph{and} the deletion of the edges in $P_t$ and thereafter adds all edges in $P_t$ to $\tilde{H}$ for the rest of the algorithm.  
It is straight-forward to obtain the size bound and running time stated by the proof in \Cref{subsec:worstCaseForDecrCutSparsifier}. 

However, lower and upper bound proofs are clearly affected by the fact that the algorithm runs on the graph $H \setminus P$. This is since $H \setminus P$ might not be an expander or satisfy the min-degree lower bound of $\Delta_{\min}$ at some times. However, the graph $H$ still satisfies these guarantees and this will in fact give even stronger lower and upper bounds. To see this, fix a specific cut $(X,\overline{X})$ and observe that lower and upper bound essentially use the concentration achieved by the edge sampling experiments of edges in the cut $(X,\overline{X})$. But in our new algorithm, we can now include the edges in $P$ that are in the cut $(X,\overline{X})$ as experiments with unit weight into $\tilde{H}$ that succeed with probability $1$. Since the amount of overestimation/underestimation is the same, but each of these edges in $P$ does not contribute to the error we must thus have that the remaining edges deviate even more from the average. 
More precisely, for the lower bound proven in \Cref{subsec:lowerBoundAdaptiveSparsifier}, the only claim affected is \Cref{clm:alphaUnderestimateSeldom} where we took a Chernoff bound to upper bound the probability that a specific vertex $c \in C_{\geq 1/20}$ at some stage $t$ is $3/4$-underestimating with probability at most $n^{-\frac{2^6(\alpha+1)\Delta_{\max}}{\Delta_{\min}\phi}\ell}$.

To derive this bound we used the scaled Chernoff bound from \Cref{thm:scaledChernoff} and the definition of $C_{\geq 1/20}$ which implies that there are at least $\phi\Delta_{\min}/20$ edges in the set $E_{H^t}(\{c\}, \overline{X})$. But we notice that, if we let each edge in $E_{H^t}(\{c\}, \overline{X}) \cap P$ be sampled in an experiment with probability $1$ (which also means these experiments are independent of the sampling step at the vertex),  we obtain the exact same upper bound. The rest of the proof is a straight-forward extension.

For the upper bound, considering the stochastic process $Z_0, Z_1, \dots$, we observe that there could only be less $(u,v)$-relevant experiments since edges in $P$ are now removed from the graph $H \setminus P$ that the algorithm operates on while the guarantees on $H$ remain unchanged. Thus, in the final proof of \Cref{clm:upperBoundForTauEvent} the probability upper bound still applies and everything goes through, giving the same upper bound on $\tilde{H} \setminus P$ for $H$. But the edges in $P$ are added with unit weight and can therefore at most add for each cut, the weight of the cut. Thus, we obtain again an $O(k)$ approximation.
\end{proof}

\subsection{Applying the Black Box Reduction}
\label{subsec:resultsViaBlackBoxFromCutSparsifier}

In this section, we first state that the graph problems of maintaining a cut-sparsifier fits the framework of a graph problem defined in \Cref{define_properties}.
The proof can be found in \Cref{sec:properties}

\begin{restatable}{lemma}{cutSparsifierIsAGraphProperty}\label{lma:cutSparsifierIsAGraphProblem}
Let $\H(\epsilon, G)$ be the set of all valid $e^{\epsilon}$-approximate cut-sparsifiers of $G$.
Then $\H$ satisfies \property{} from \Cref{define_properties}.
\end{restatable}

Using this lemma, we can now derive our first main result from \Cref{thm:amortizedCutSparsifierMainResult}: an amortized adaptive algorithm  that maintains an $O(k)$-approximate cut-sparsifier.

\paragraph{Proof of \Cref{thm:amortizedCutSparsifierMainResult}}
\begin{proof}
We invoke \Cref{thm:amortized:fully_dynamic_weighted} with $\phi = 1/\log^4 n$ and have by \Cref{lma:cutSparsifierIsAGraphProblem} that the required graph properties are given, and by \Cref{thm:generalStatementCutSparsifier} there is an algorithm $\mathcal{A}$ that maintains an $O(k)$-approximate cut-sparsifier on decremental $(1/\log^4 n)$-expanders with preprocessing time $O(m)$, size bounded by $\tilde{O}(n)$ and worst-case by $\tilde{O}(n^{1/k})$ (here we use that by \Cref{def:amortized:decremental_algorithm} the degrees in the graph that the algorithm $\mathcal{A}$ is run upon are at all times in $[c_1 \cdot \Delta\phi, c_2 \cdot \Delta/\phi]$ for some constants $c_1, c_2$ and $\Delta$). Thus, the preprocessing time of the final algorithm is $O(m)$, and the worst-case update time $\tilde{O}(n^{1/k})$ and the size of the output graph is $\tilde{O}(n \log W)$.
\end{proof}

Next, we show that  our algorithm from \Cref{subsec:worstCaseForDecrCutSparsifier} can be extended to  an algorithm for cut-sparsifiers with bounded worst-case update time for general weighted graphs which proves \Cref{thm:worstCaseCutSparsifierMainResult}

\paragraph{Proof of \Cref{thm:worstCaseCutSparsifierMainResult}}
\begin{proof}
We can now simply use the algorithm from \Cref{lma:WorstCaseBlackBoxCutSparsifier} with approximation $k^2 = e^{2 \log k}$ in the black box reduction~(c.f.~ \Cref{thm:reduction_worst_case}) where we set $N =n, d = n^{1/k}$, and $L = \lceil k \rceil$ and choose $\phi = 1/\log^4 n$.  

Thus, we obtain an algorithm to maintain cut-sparsifiers for general weighted graphs on $n$ nodes, 
and aspect ratio $W$, with approximation ratio $e^{O(\log k) \cdot L} = k^{O(k)}$
 with preprocessing time $\tilde{O}(m \log W)$ and size $\tilde{O}(n \log W)$, and worst-case update time
\begin{align*}
&\left(n^{O(1/\sqrt{\log n})} \cdot \tilde{O}\left(\left(\frac{\Delta_{\max}\log n}{\Delta_{\min}\phi}\right)^3 n^{1/k^2}\right) \right)^{O(L)}
\cdot (1 + \tilde{O}(n^{1/k} \log W)) \\
&= (\log n)^{O(k)} \cdot n^{O(1/\sqrt{\log n}} \cdot n^{1/k} \cdot n^{O(1/k)} \log W \\
&= n^{O(1/k)} \log W
\end{align*}
where we bound the recourse by our worst-case update time (each edge is update explicitly in $\tilde{H}$ so the bound on the running time must subsume the number of edge updates) and  the upper bound on $k$ to conclude that $(\log n)^{O(k)} = n^{O(1/k)}$ for all valid choices of $k$.
\end{proof}

\section{Spanners and Spectral Sparsifiers Against an Adaptive Adversary}
\label{sec:adaptive_more}

In this section, we show that adaptive algorithm for both spanners and spectral sparsifiers follows easily from our adaptive algorithms for cut sparsifiers. The results are formally summarized as follows:

\begin{theorem}\label{thm:SpannerMainResult}
There exists an algorithm that maintains an $\polylog(n)$-approximate spanner $\tilde{G}$ on any dynamic graph $G$. The algorithm  works against an adaptive adversary with high probability, has $O(m)$ initialization time, $O(\polylog n)$ amortized update time and contains $\tilde{O}(n \log W)$ edges.
\end{theorem}

\begin{theorem}
\label{thm:SpannerMainResult_worstcase}
For any $1 \leq k = O(\sqrt{\log n})$, there exists an algorithm that maintains an $\log^{O(k)} n$-approximate spanner $\tilde{G}$ on any dynamic graph $G$. The algorithm works against an adaptive adversary with high probability, has $\tilde{O}(m \log W)$ initialization time, $\tilde{O}(n^{1/k})$ worst-case update time and contains  $\tilde{O}(n \log W)$ edges.
\end{theorem}

\begin{theorem}\label{thm:SpectralMainResult}
There exists an algorithm that maintains an $O(\polylog n)$-approximate spectral sparsifier $\tilde{G}$ on any dynamic graph $G$. The algorithm  works against an adaptive adversary with high probability, has $O(m)$ initialization time, $O(\polylog n)$ amortized update time and contains $\tilde{O}(n \log W)$ edges.
\end{theorem}

\begin{theorem}
\label{thm:SpectralMainResult_worstcase}
There exists an algorithm that maintains, for any given parameter $k\ge 1$, an $\log^{O(k)}n$-approximate spanner $\tilde{G}$ on any dynamic graph $G$. The algorithm  works against an adaptive adversary with high probability, has $O(m)$ initialization time, $\tilde{O}(n^{1/k})$ worst-case update time and contains  $\tilde{O}(n \log W)$ edges.
\end{theorem}

The main idea of all the above results follows from the observation that any cut sparsifier of an expander is both a spanner and a spectral sparsifier of that expander.  

As the approximation ratio of \Cref{thm:SpectralMainResult,thm:SpectralMainResult_worstcase} is quite high, in \Cref{sec:spectral:query}, we also provide an additional result for spectral sparsifier with good approximation ratio.
The data structure in \Cref{sec:spectral:query} maintains a dynamic data structure with polylogarithmic amortized update which allows us to query for a $(1+\epsilon)$-approximate spectral sparsifier 
using $\tilde{O}(n / \epsilon^3)$ query time and this algorithm works against an adaptive adversary:

\begin{restatable}{theorem}{thmSpectralQuery}
\label{thm:query:spectral_sparsifier}
There exists a fully dynamic algorithm that maintains for any weighted graph 
an $e^\epsilon$-approximate spectral sparsifier against an \emph{adaptive} adversary. The algorithm's pre-processing time is bounded by $O(m)$, amortized update time is $
O\left(
	\log^{19} n
\right),
$
and query time is $O(n \log^{25} (n) \epsilon^{-3} \log W)$,
where $W$ is the ratio between the largest and the smallest edge weight.
The query operation returns an $e^\epsilon$-approximate spectral sparsifier of $G$.
\end{restatable}

\subsection{Adaptive Spanners}
\label{subsec:SpannerViaBlackBoxFromCutSparsifier}

Let us begin by proving that a cut-sparsifier on a $\phi$-expander is a spanner.

\begin{lemma}\label{lma:cutSparsifierIsSpannerOnExpander}
Let $G=(V,E)$ be an unweighted $\phi$-expander with $m$ edges.
Let $H$ is an $\alpha$-approximate cut-sparsifer of $G$ which is also a subgraph of $G$. Then, any
two vertices $u$ and $v$, there is a path in $H$ containing at
most $L=O(\frac{\alpha\log\alpha m}{\phi})$ edges connecting $u$
and $v$. In particular, let $H'$ be the unweighted graph with the same
edge set as $H$. Then, $H'$ is an $L$-spanner of $G$.
\end{lemma}
\begin{proof}
First, observe that $H$ a $(\phi/\alpha)$-expander. Next, suppose for contradiction that there are
two vertices $u$ and $v$ with \emph{unweighted} distance in $H$
greater than $L=\frac{100\alpha\log m}{\phi}$. For each $0\le i\le L/2$,
let $B_{i}$ contain all vertices reachable from $u$ using at most $i$
edges. That is, $B_{i}$ is a ball around $u$ of \emph{unweighted} radius $i$. Similarly, let $B'_{i}$ be a ball around $v$ of unweighted
radius $i$. Note that $B_{L/2}$ and $B'_{L/2}$ must be vertex disjoint.
So we assume w.l.o.g.~that $\vol_{H}(B_{L/2})\le\vol_{H}(B'_{L/2})$.
Now, we claim that there is $i\le L/2$ such that $B_{i}$ is a $\phi/\alpha$-sparse
cut in $H$. Otherwise, we have that 
\[
\vol_{H}(B_{L/2})\ge(1+\phi/\alpha)^{L/2}\cdot\vol_{H}(B_{0})\ge e^{10\log\alpha m} \cdot 1/\alpha \ge(\alpha m)^{5}
\]
This is a contradiction because $H$ is a $\alpha$-approximate
cut-sparsifer of $G$, and so $\vol_{H}(V)<\alpha\vol_{G}(V)\le2\alpha m$.

The above argument shows that $H'$ has diameter at most $L$. As $H'$ is a subgraph of $G$, $H'$ must be an $L$-spanner of $G$.
\end{proof}

Again, we also need in this section hat the problem of maintaining a spanner fits the framework from \Cref{define_properties}.
This is proven in \Cref{sec:properties}.

\begin{restatable}{lemma}{spannerIsAGraphProblem}\label{lma:spannerIsAGraphProblem}
Let $\H(G, \epsilon)$ be the set of all valid $e^\epsilon$-approximate spanners of $G$.
So $H \in \H(G, \epsilon)$, when $\dist_G(u,v) \le \dist_H(u,v) \le e^\epsilon \dist_G(u,v)$.
Then $\H$ satisfies properties \property{} from \Cref{define_properties}.
\end{restatable}

Using these two lemmas, it is now straight-forward to obtain an algorithm that maintains a spanner on general graphs.

\paragraph{Proof of \Cref{thm:SpannerMainResult}}
\begin{proof}
We can use the algorithm $\mathcal{A}$ from \Cref{thm:generalStatementCutSparsifier} for any approximation parameter $k \geq 1$, which gives an algorithm that maintains an $O(k \cdot \log^{5} n)$-spanner on a $(1/\log^4 n)$-expander with amortized update time $\tilde{O}(n^{1/k})$ of size $\tilde{O}(n)$ by \Cref{lma:cutSparsifierIsSpannerOnExpander}, where we assume that $\Delta_{min} = \Tilde{\Theta}(\Delta_{max})$.

Since,  by \Cref{lma:spannerIsAGraphProblem},  $\alpha$-spanner fits the framework in the black box reduction~(\Cref{thm:amortized:fully_dynamic_weighted}), we can apply the algorithm $\mathcal{A}$ which is then run on a graph where vertex degrees are at all times where a cut-sparsifier is requested is in $[c_1 \cdot \Delta\phi, c_2 \cdot \Delta/\phi]$ for some constants $c_1, c_2$ and $\Delta$). Thus, the preprocessing time of the final algorithm is $O(m)$, and the amortized update time $\tilde{O}(n^{1/k})$ and the size of the output graph is $\tilde{O}(n \log W)$ while the approximation becomes $k \cdot \log^{O(1)} n$. Setting $k = \log n$, we obtain our final result.
\end{proof}

\paragraph{Proof of \Cref{thm:SpannerMainResult_worstcase}}
\begin{proof}
Given approximation parameter $1 \leq k \leq \sqrt{\log n}$ in the theorem. Then, we first use the algorithm from \Cref{lma:WorstCaseBlackBoxCutSparsifier} with approximation $\log n$, which gives an algorithm that maintains a $(\polylog n)$-spanner on an $(1/\log^4 n)$-expander with worst-case update time $2^{O(\sqrt{\log n})}$ of size $\tilde{O}(n)$ by \Cref{lma:cutSparsifierIsSpannerOnExpander}.

Since $\alpha$-spanner fits the framework of the black box reduction~(c.f., \Cref{thm:reduction_worst_case}) by \Cref{lma:spannerIsAGraphProblem}, we can apply the algorithm above and set the parameters in the black box to  $\epsilon = k, N =n, d = n^{1/k}$ which implies $L = \lceil  k \rceil$ and choose again $\phi = 1/\log^4 n$.   
Hence, we obtain an algorithm to maintain a $(\polylog n)^{O(L)} = (\log n)^{O(k)}$-spanner for general weighted graphs on $n$ nodes, and aspect ratio $W$, with preprocessing time $\tilde{O}(m \log W)$, size $\tilde{O}(n \log W)$, and worst-case update time $\tilde{O}(2^{O(\sqrt{\log n})} \cdot n^{1/k}) = \tilde{O}(n^{O(1/k)})$. Choosing $k$ by a constant factor larger then the input parameter $k$, we can decrease the running time to $\tilde{O}(n^{1/k})$ at the cost of a slightly increased approximation ratio.
\end{proof}

\subsection{Adaptive Spectral Sparsifiers}
\label{sec:adaptive_spectral}

\global\long\def\xvec{\vec{x}}
\global\long\def\xx{\vec{x}}
\global\long\def\Lhat{\widehat{L}}
\global\long\def\Gtil{\tilde{G}}

In this section, we show that our result on adaptive algorithms for
cut sparsifiers immediately an adaptive algorithm for spectral sparsifier:

First, we need the following well-known lemma (see e.g.~Lemma 6.7 of \cite{ChuzhoyGLNPS19}) which says that any two expanders with the same degree profile are approximate
spectral sparsifiers of each other. We only give the proof here for
completeness (our proof is almost identical to the one in \cite{ChuzhoyGLNPS19} except that we rename some notations).

Below, the degree of a vertex $v$ is the total weight of edges incident to $v$ (i.e. the weighted degree where we count the self loops as well). 
\begin{lemma}
\label{lem:spectral exact degree}Let $G$ and $\tilde{G}$ be two
graph such that with the same set of vertices such that $\deg_{G}(v)=\deg_{\Gtil}(v)$
for all vertices $v$. If both $G$ and $\tilde{G}$ are $\phi$-expanders.
Then, for any $\xvec\in\mathbb{R}^{n}$, we have $\frac{\phi^{2}}{4}\xvec L_{G}\xvec\le\xvec L_{\Gtil}\xvec\le\frac{4}{\phi^{2}}\xvec L_{G}\xvec$.
\end{lemma}

\begin{proof}
Let $(\deg)$ denote the degree vector where $(\deg)_{v}=\deg_{G}(v)=\deg_{\Gtil}(v)$
and $D$ be the diagonal matrix where $(D)_{vv}=\deg_{v}$. For any
graph $H$, the normalized Laplacian $\Lhat_{H}$ of a weighted graph
$H$ is defined as $D^{-1/2}L_{H}D^{-1/2}$, where $L_{H}$ is the
Laplacian of $H$.

Let $\Lhat_{G}$ and $\Lhat_{\Gtil}$ be normalized Laplacians of
$G$ and $\Gtil$, respectively. It is well-known that eigenvalues
of normalized Laplacians are between $0$ and $2$. Also, observe
that, for any graph $H$, $L_{H}\vec{1}=0$. Therefore, $\Lhat_{G}(\deg)^{1/2}=\Lhat_{\Gtil}(\deg)^{1/2}=0$.
That is, $(\deg)^{1/2}$ is in the kernel of both $\Lhat_{G}$ and
$\Lhat_{\Gtil}$.
Let $\lambda$ be the second smallest eigenvalue of $\Lhat_{G}$.
Then for any vector $\xx'\perp\left(\deg_{G}\right)^{\frac{1}{2}}$,
we have: 
\[
\frac{\lambda}{2}\xx'^{\top}\Lhat_{\Gtil}\xx'\le\lambda\lVert\xx'\rVert^{2}\le\xx'^{\top}\Lhat_{G}\xx',
\]
since the largest eigenvalue of $\Lhat_{\Gtil}$ is at most $2$. This
implies that, for every vector $\xx\in\mathbb{R}^{n}$, $\xx^{\top}\Lhat_{G}\xx\ge\frac{\lambda}{2}\xx^{\top}\Lhat_{\Gtil}\xx$
holds. Indeed, we can write 
\[
\xx=\xx'+c\left(\deg\right)^{\frac{1}{2}}
\]
where $\xx'\perp(\deg_{G})^{\frac{1}{2}}$ and $c$ is a scalar. This
gives: 
\begin{align*}
\xx^{\top}\Lhat_{G}\xx & =\left(\xx'+c\left(\deg\right)^{\frac{1}{2}}\right)^{\top}\Lhat_{G}\left(\xx'+c\left(\deg\right)^{\frac{1}{2}}\right)\\
 & =\xx'^{\top}\Lhat_{G}\xx'\\
 & \ge\frac{\lambda}{2}\cdot\xx'^{\top}\Lhat_{\Gtil}\xx'\\
 & =\frac{\lambda}{2}\cdot\left(\xx'+c\left(\deg\right)^{\frac{1}{2}}\right)^{\top}\Lhat_{D}\left(\xx'+c\left(\deg\right)^{\frac{1}{2}}\right)\\
 & =\frac{\lambda}{2}\cdot\xx^{\top}\Lhat_{\Gtil}\xx.
\end{align*}
By Cheeger's inequality, we have $\lambda\ge\Phi(G)^{2}/2\ge\phi^{2}/2$.
Therefore, for any vector $\xx\in\mathbb{R}^{n}$: 
\begin{equation}
\xx^{\top}\Lhat_{G}\xx\ge\frac{\phi^{2}}{4}\xx^{\top}\Lhat_{\Gtil}\xx\label{eq:cheeger}
\end{equation}
We can now conclude that, for any vector $\xx\in\mathbb{R}^{n}$:
\begin{align*}
\xx^{\top}L_{G}\xx & =\xx^{\top}D^{1/2}\Lhat_{G}D^{1/2}\xx\\
 & \ge\frac{\phi^{2}}{4}\xx^{\top}D^{1/2}\Lhat_{\Gtil}D^{1/2}\xx\\
 & =\frac{\phi^{2}}{4}\xx^{\top}D^{1/2}D^{-1/2}L_{\Gtil}D^{-1/2}D^{1/2}\xx\\
 & =\frac{\phi^{2}}{4}\xx^{\top}L_{\Gtil}\xx
\end{align*}
where the inequality follows by applying \Cref{eq:cheeger} to vector
$D^{1/2}\xvec$. The proof that $\xx^{\top}L_{\Gtil}\xx\ge\frac{\phi^{2}}{4}\xx^{\top}L_{G}\xx$
is symmetric.
\end{proof}

Next, we relax the condition in the above lemma and show that it holds even when the degree
profiles of both graphs are approximately the same:
\begin{lemma}
\label{lem:spectral approx degree}Let $\alpha>1$ be an approximation
parameter. Let $G$ and $\Gtil$ be two graphs with the same set of
vertices such that $\frac{1}{\alpha}\deg_{G}(v)\le\deg_{\Gtil}(v)\le\alpha\deg_{G}(v)$
for all vertices $v$. Suppose that both $G$ and $\Gtil$ are $\phi$-expanders.
Then, $\Gtil$ is a $\poly(\alpha/\phi)$-approximate spectral sparsifier
of $G$.
\end{lemma}

\begin{proof}
Let $\deg'(v)=\max\{\deg_{G}(v),\deg_{\Gtil}(v)\}$ for all $v$.
Let $G'$ and $\Gtil'$ be obtained from $G$ and $\Gtil$ by adding
self-loops so that the degree of each vertex $v$ is $\deg'(v)$ in
both $G$ and $\Gtil$. Observe that (1) both $G$ and $\Gtil$ are
still $(\phi/\alpha)$-expanders because the degree of each vertex
is increased by at most $\alpha$ factor, and (2) $L_{G}=L_{G'}$
and $L_{\Gtil}=L_{\Gtil'}$ because self-loops do not contribute to
any entry in the Laplacian matrices. By applying \Cref{lem:spectral exact degree}
to $G'$ and $\Gtil'$, we are done. 
\end{proof}
\begin{corollary}
\label{cor: cut on expander is spectral}
Let $G$ be a $\phi$-expander. If $\Gtil$ is $\gamma$-approximate
cut sparsifier of $G$, then $\tilde{G}$ is a $\poly(\gamma/\phi)$-approximate
spectral sparsifier of $G$.
\end{corollary}

\begin{proof}
Observe that $\Gtil$ must be a $(\phi/\gamma)$-expander. So both
$G$ and $\Gtil$ are $(\phi/\gamma)$-expanders and each vertex $v$
is such that $\frac{1}{\gamma}\deg_{G}(v)\le\deg_{\Gtil}(v)\le\gamma\deg_{G}(v)$.
The claim follows from applying \Cref{lem:spectral approx degree}.
\end{proof}

\begin{proof}[Proof of \Cref{thm:SpectralMainResult,thm:SpectralMainResult_worstcase}]
Spectral sparsifiers satisfy all the required properties to be maintained by our reduction to expanders, 
as proven in \Cref{lem:spectralSparsifierIsAGraphProblem}.
By \Cref{cor: cut on expander is spectral} the algorithms of \Cref{thm:generalStatementCutSparsifier} 
also maintains a $O(k/\phi)$-approximate spectral sparsifier. 
So the same proof of \Cref{thm:amortizedCutSparsifierMainResult} (see \Cref{subsec:resultsViaBlackBoxFromCutSparsifier})
also yields a $\polylog(n)$-approximate spectral sparsifier 
by choosing $k = \log n$, $\phi = 1/\log^4 n$. 
Thus we obtain \Cref{thm:SpectralMainResult}.

Likewise, \Cref{lma:WorstCaseBlackBoxCutSparsifier} also maintains a $O(k/\phi)$-approximate spectral sparsifier,
so the proof of \Cref{thm:worstCaseCutSparsifierMainResult} (see \Cref{subsec:resultsViaBlackBoxFromCutSparsifier})
also yields \Cref{thm:SpectralMainResult}.
The only difference is the approximation guarantee.
The approximation of the spectral sparsifier is
$$
O(k^2/\phi)^{O(L)} \le \log^{O(k)}n,
$$
where $O(k^2/\phi)$ comes from running a $k^2$-approximate cut sparsifier,
and $O(L)$ comes from \Cref{thm:reduction_worst_case}.
The upper bound uses $k \le \sqrt{\log n}/\log \log n$ and $L = \lceil k \rceil$.
\end{proof}

\section{Spectral Sparsifiers Against an Oblivious Adversary}
\label{secss}
In this section, we prove \Cref{cor:main spectral}. That is, we give the first dynamic algorithm with \emph{worst-case} update time for maintaining a spectral sparsifier. 
The precise statement is summarized as follows:

\begin{restatable}{theorem}{thmSpectralWorstCase}
\label{thm:worst_case:spectral}
Fix some $n \ge 1$.
There exists a fully dynamic algorithm 
that maintains $e^\delta$-spectral sparsifiers of \emph{weighted} (up to) $n$ node graphs
whose ratio of largest to smallest weight is $W$
against an \emph{oblivious} adversary.

The algorithm maintains a sparsifier of size $2^{O(\sqrt{\log n})} \cdot O(n \delta^{-2} \log W)$.
The worst-case update time is
$$2^{O(\log^{0.75} n)}+\delta^{-O(\log^{0.25}n)} \cdot O(\log W).$$
(So for any constant $\delta$ or $\delta=1/\polylog n$, 
the update time is bounded by $2^{O(\log^{0.75} n))} \cdot O(\log W).$)
The preprocessing time is $\tilde{O}(m\delta^{-2} \log(W))$.
\end{restatable}

To prove this result, in \Cref{sec:spectral:sampling}
we explain how to obtain spectral sparsifiers via randomly sampling edges. Then, 
in \Cref{sec:spectral:worst-case} we use these results to
create a decremental algorithm on pruned $\phi$-expanders. Applying the expander reduction then results in \Cref{thm:worst_case:spectral}.

We first note that spectral sparsifiers indeed satisfy all the required properties listed in \Cref{define_properties}
in order to apply the expander reduction (see \Cref{sec:properties} for the proof).
\begin{restatable}{lemma}{spectralSparsifierIsAGraphProblem}\label{lem:spectralSparsifierIsAGraphProblem}
Let $\H(\epsilon, G)$ be the set of all valid $e^\epsilon$-approximate spectral sparsifiers of $G$.
Then $\H$ satisfies \property{} from \Cref{define_properties}.
\end{restatable}

\subsection{Simple Random Sampling}
\label{sec:spectral:sampling}

Now we study how random sampling results in  spectral sparsifiers. 
Our algorithm are based on the following result: 

\begin{lemma}[\cite{SpielmanT11}, Theorem~6.1]\label{lemma:ST-decomposition}
Let $\varepsilon\in(0,1/2)$, and $G=(V,E)$ be an unweighted graph 
whose smallest non-zero normalised Laplacian eigenvalue is at least $\lambda$. 
By sampling every edge $\{u,v\}$ with probability 
\[
p_{u,v} \ge \min\left\{1, \left(\frac{12\log n }{\varepsilon\lambda}\right)^2\frac{1}{\min\{\deg(u), \deg(v)\}} \right\}
\]
and setting the weights of a sampled edge as $1/p_{u,v}$, 
then with probability at least $1-1/n^3$ the resulting graph $\widetilde{G}$ 
is a $(1+\varepsilon)$-spectral sparsifier of $G$.
\end{lemma}

By combining the lemma above and the Cheeger inequality ($\lambda \ge \Phi_G^2/2$), we have the following corollary:
\begin{corollary}\label{cor:ST-decomposition}
Let $\varepsilon\in(0,1/2)$, and $G=(V,E)$ be an unweighted graph with conductance $\Phi_G$. 
By sampling every edge $\{u,v\}$ with probability 
\begin{equation}\label{eq:defpuv}
p_{u,v} \ge \min\left\{1, \left(\frac{24\log n }{\varepsilon\Phi_G^2}\right)^2\cdot\frac{1}{\min\{\deg(u), \deg(v)\}} \right\}
\end{equation}
and setting the weights of a sampled edge  as $1/p_{u,v}$, 
then with probability at least $1-1/n^3$ 
the resulting graph $\widetilde{G}$ is a $(1+\varepsilon)$-spectral sparsifier of $G$.
\end{corollary}

Intuitively, \Cref{cor:ST-decomposition} says that we obtain a spectral sparsifier 
by simply sampling every edge proportional to the degrees of the endpoints.
\Cref{def:worst-case:decremental_algorithm} guarantees us that the graph is of near uniform degree
and that the degrees of our graphs stay roughly the same throughout all updates.
Thus we obtain an algorithm by simply sampling the graph once during initialization.

Unfortunately, the graph does not stay an expander throughout all updates 
(see \Cref{def:worst-case:decremental_algorithm}).
Instead we are only given a graph $G$ and set $P \subset E(G)$ 
for which there exists some $2^{-O(\sqrt{\log n})}$-expander $W$
with the property $G \setminus P \subset W \subset G$.
We first show that it is sufficient to obtain a spectral sparsifier,
by only sampling the edges in $G \setminus P$.

\begin{lemma}\label{cor:sparsifyExpanderSubgraph}
Let $G=(V,E)$ be an unweighted graph and let $P\subset E$ with the property 
that there is a $\phi$-expander $W$ with $G \setminus P \subset W \subset G$
and the guarantee that $\deg_W(v) \ge \Delta/2$ for all $v \in V(G \setminus P)$ 
and some parameter $\Delta \ge 1$.
Assume that we sample every edge $\{u,v\}\in E\setminus P$ with probability
\begin{align}
p_{u,v} = \min\left\{1, \left(\frac{24\log n }{\varepsilon\Phi_G^2}\right)^2\cdot\frac{2}{\Delta} \right\} 
\label{eq:sampling_probability}
\end{align}
and let $\widetilde{H}$ be the resulting graph. 
Then, with high probability the graph 
$G'= \widetilde{H}\cup P$ is a $(1+\varepsilon)$-spectral sparsifier of $G$.
\end{lemma}

\begin{proof}
Consider the case where we sample every edge $\{u,v\} \in E(W) \setminus P$ 
with probability as in \eqref{eq:sampling_probability}
and let $H$ be the resulting graph.
Then $H \cup (E(W) \cap P)$ is a spectral sparsifier of $W$ by \Cref{cor:ST-decomposition},
because every edge in $E(W)$ is either sampled with probability
\begin{align*}
p_{u,v} &= \min\left\{1, \left(\frac{24\log n }{\varepsilon\Phi_G^2}\right)^2\cdot\frac{2}{\Delta} \right\} 
\ge \min\left\{1, \left(\frac{24\log n }{\varepsilon\Phi_G^2}\right)^2\cdot\frac{1}{\min\{\deg_W(u), \deg_W(v)\}} \right\} 
\end{align*}
if $\{u,w\} \notin P$ or with probability $1$ if $\{u,v\} \in P$.

Next, note that $E(W) \setminus P = E(G) \setminus P$, because of $G \setminus P \subset W$,
so the graph $H$ is exactly the graph $\widetilde{H}$.
This implies that $G' := H \cup P = H \cup (E(W) \cap P) \cup (P \setminus E(W))$
is a sparsifier of $W \cup (P \setminus E(W)) = G$.
\end{proof}

\subsection{Worst-Case}
\label{sec:spectral:worst-case}

\Cref{cor:sparsifyExpanderSubgraph} directly implies the following decremental algorithm.
During each update the set $P$ grows a bit (and thus $G \setminus P$ shrinks)
which means for maintaining $\widetilde{H} \cup P$
we simply need to remove edges from $\widetilde{H}$ and insert them into $P$.

\begin{lemma}\label{lem:worst_case:spectral_decremental}
For every $\phi$ there exists a decremental algorithm on $\phi$-sub-expanders 
(see \Cref{def:worst-case:decremental_algorithm}) 
that maintains a $e^\epsilon$-approximate spectral sparsifier
against \emph{oblivious} adversaries.

The pre-processing time is bounded by $O(m)$,
and the output size of the sparsifier returned after the pre-processing is bounded by $O(n 2^{O(\sqrt{\log n})}/\epsilon^{2} )$.
An edge deletion takes $2^{O(\sqrt{\log n})}$ worst-case time,
and the recourse is bounded by $2^{O(\sqrt{\log n})}$ as well.
The ratio between the  largest and the  smallest edge weight in the maintained sparsifier 
is bounded by $O(n)$.
\end{lemma}

\begin{proof}
During initialization we compute the minimum degree $\Delta$ of $G$
and sample every edge in $G$ by probability as in \eqref{eq:sampling_probability}
for $\phi = 2^{-O(\sqrt{\log n})}$.
Let $H$ be the resulting graph, when scaling every edge by the inverse of the sampling probability.
Let $P$ be the empty set initially, then the output of the algorithm is $H \cup P$.
With every update, we remove the deleted and pruned edges from $H$,
and insert the pruned edges into $P$.

\paragraph{Correctness}

The graph $H \cup P$ is a spectral sparsifier by \Cref{cor:sparsifyExpanderSubgraph},
because we are promised that there is a $2^{-O(\sqrt{\log n})}$-expander $W$ with
$G \setminus P \subset W \subset G$
and $\deg_W(v) \ge \Delta/2$ for $v \in V(G \setminus P)$.

\paragraph{Complexity}

The pre-processing time is $O(m)$ as we simply iterate over all edges.
The size of $H \cup P$ after the pre-processing is   bounded by
$O(n 2^{O(\sqrt{\log n}) / \epsilon^2})$ with high probability,
because the graph $G$ is of near uniform degree up to a factor of $O(1/\phi)$.
The update time and recourse is $2^{O(\sqrt{\log n})}$, 
because we remove at most that many edges from $H$ and insert them into $P$.
\end{proof}

By applying the expander reduction of \Cref{thm:reduction_worst_case} 
we now obtain the following fully dynamic spectral sparsifier algorithm
with worst-case update time.

\thmSpectralWorstCase*

\begin{proof}
The algorithm follows from \Cref{lem:worst_case:spectral_decremental} 
and the reduction of \Cref{thm:reduction_worst_case}.
We pick $N=n$ and $d = \lceil 2^{\log^{0.75}n} \rceil$ such that $\log d \ge \log^{0.75} N$, so then 
$$
L := \lceil \log(N)/ \log d \rceil \le  \lceil \log^{0.25} N\rceil = O(\log^{0.25} N).
$$
This results in worst-case update time
\begin{align*}
&~
(R(n)/\epsilon)^{O(L)}
\cdot 
\left(  T(n)+\frac{P(dS(n)\log (wW))}{S(n)}  \right)\\
=&~
 2^{O(\sqrt{\log n}) \cdot O(\log^{0.25} n)} \epsilon^{-O(\log^{0.25} n)} 
\left(2^{O(\sqrt{\log n})} + 2^{O(\log^{0.75 n})}\log(wW)\right) \\
=&~
2^{O(\log^{0.75} n)} \epsilon^{-O(\log^{0.25} n)} O(\log(W)),
\end{align*}
where we used $T(n) = 2^{O(\sqrt{\log n})}$, $T(m) = O(m)$ and $w = O(n)$.
In order to get $e^\delta$-sparsifier, we pick $\epsilon = \delta/L = O(\delta/\log^{0.25} n)$, 
so then the update time becomes
$$
2^{O(\log^{0.75} n)} \delta^{O(\log^{0.25} n)} O(\log(W))
$$
So for $\delta \ge 2^{-O(\log^{0.5} n)}$, the algorithm runs in $2^{O(\log^{0.75} n)} \cdot O(\log W)$ worst-case update time.

The preprocessing time is bounded by
$$
\tilde{O}((P(m)+m)\epsilon^{-2} L^2 \log(wW))
=
\tilde{O}(m \epsilon^{-2} \log W)
=
\tilde{O}(m \delta^{-2} \log W),
$$
because our spectral sparsifiers are subgraphs.
The size of the sparsifier is bounded by
$$
\tilde O(S(n) L \log(wW))
=
\tilde{O}(n 2^{O(\sqrt{\log n})} \epsilon^{-2} \log W)
=
2^{O(\sqrt{\log n})} \cdot O(n \delta^{-2} \log W).
$$
\end{proof}

\part{Applications\label{part:application}}
\section{Applications to Decremental Shortest Paths}
\label{sec:SSSP}

In this section, we show how to extend the currently fastest algorithms to maintain the $(1+\epsilon)$-approximate distances from a fixed source $s \in V$, in the decremental graph $G$ to report shortest paths in time $O(n)$. We therefore turn to the algorithms in \cite{Bernstein17} for weighted dense graphs and the algorithm \cite{gutenberg2020deterministic} for sparse graphs. Both algorithms are built around the framework introduced in \cite{BernsteinC16} that essentially partitions the graph $G$ into a \emph{light} graph $G^{light}$ and a \emph{heavy} graph $G^{heavy}$ (here we mean that the edge sets of the graphs $G^{light}, G^{heavy}$ partition the edge set of $G$). Both algorithms try to keep $G^{light}$ sparse and at the same time, they ensure that connected components of $G^{heavy}$ have small diameter. Thereby both algorithm satisfy the conditions stated below.

\begin{definition}[Heavy-light Algorithm]\label{def:heavyLightFramework}
We say an algorithm runs within the \emph{heavy-light framework} if given a decremental graph $G$, a fixed source $s \in V$, and an integer $i \leq \lg n$ and some real $\epsilon, \delta > 0$, it maintains the decremental graphs $G^{heavy}_i$ and $G^{light}_i$ such that at any stage
\begin{enumerate}
    \item $G^{heavy}_i$ and $G^{light}_i$ partition the edge set of graph $G$ while having the same vertex set, and
    \item $G^{heavy}_i$ is a decremental algorithm, and 
    \item given the connected components $\mathcal{C} = C_1, C_2, \dots, C_k$ in graph $G^{heavy}_i$, distances in the graph $G / \mathcal{C}$ ,i.e. the graph $G$ after contracting each connected component $C_i$ for $i \in [1,k]$, are only distorted by an additive error of $\delta 2^i$ compared to $G$, and
    \item for any $v \in V$, such that $\mathbf{dist}_G(s,v) \in [2^{i}, 2^{i+1})$, the algorithm can be queries for an $(1\pm \epsilon/2)$-approximate shortest paths $\pi_{s,v}$ in the graph $G^{light}_i / \mathcal{C}$.
\end{enumerate}
\end{definition}

We state the results obtained by the two algorithms below. While these theorems are not explicitly stated in their papers, it is not hard to verify that the theorems are straight-forward to obtain by closely inspecting the main result of both articles.

\begin{theorem}[see \cite{BernsteinC16, Bernstein17}]
For any decremental weighted graph $G=(V,E,w)$, fixed source $s$, any integer $i \leq \log n$, and reals $\epsilon, \delta > 0$, there is an algorithm $\mathcal{A}_i$ that runs within the \emph{heavy-light framework} and can return $(1\pm\epsilon/2)$-approximate shortest-paths in $G^{light}_i / \mathcal{C}$ in time linear in the number of edges. Algorithm $\mathcal{A}_i$ is deterministic and runs in total update time $\tilde{O}(n^2 / (\epsilon + \delta))$. 
\end{theorem}

\begin{theorem}[see \cite{gutenberg2020deterministic}]
For any decremental unweighted graph $G=(V,E)$, fixed source $s$, any integer $i \leq \log n$, constant $\epsilon > 0$ and real $\delta > 0$, there is an algorithm $\mathcal{B}_i$ that runs within the \emph{heavy-light framework} and can return $(1\pm\epsilon/2)$-approximate shortest-paths in $G^{light}_i / \mathcal{C}$ in time linear in the number of edges. Algorithm $\mathcal{B}_i$ is deterministic and runs in total update time total update time $\tilde{O}\left(\frac{mn^{0.5 + o(1)} }{ \delta}\right)$.
\end{theorem}

Now, for each $i \leq \log nW$ (where $W$ is $1$ for unweighted graphs), we maintain a data structure $\mathcal{A}_i$ (or $\mathcal{B}_i$) with $\delta = \epsilon/2\alpha$ such that $\alpha$ is the approximation factor in the algorithm to maintain an $\alpha$-spanner $\tilde{G}_i$ as described in \Cref{thm:SpannerMainResult} on the graph $G^{heavy}_i$, i.e. $\alpha = O(\polylog n)$. Then, for a path query from $s$ to any vertex $v \in V$ where $\mathbf{dist}_G(s,v) \in [2^i, 2^{i+1})$, we first query for a path $\pi_{s,v}$ in the graph $G^{light}_i / \mathcal{C}$. 

Then for every vertex $w$ on $\pi_{s,v}$ that corresponds to a connected component $C_i \in \mathcal{C}$ in $G^{heavy}_i$, we identify the two vertices $x$ and $y$ in $C_i$ such that the edges on $\pi_{s,v}$ ending in $w$ in $G^{light}_i / \mathcal{C}$ have endpoints in $G$ in $x$ and $y$. (We let $x$ be set to $s$ if $s \in C_i$, and $y$ be set to $v$ if $v \in C_i$.) 

Then, we run Dijkstra's algorithm on the $\alpha$-spanner $\tilde{G}_i$ from $x$ to the shortest-path in the spanner from $x$ to $y$. We do so for every connected component $C_i$ that is as vertex $w$ on the path and replace the vertex $w$ then with the path from $x$ to $y$ found in $\tilde{G}_i$. Thus, the final path $\pi_{s,v}'$ that we obtain is a path in $G$.

To see that $\pi_{s,v}'$ is a $(1+\epsilon)$-approximate shortest-path, we observe that there are no edges between distinct connected components in the graph $G^{heavy}_i$ by the definition of connected components, and since a spanner is a subgraph of $G^{heavy}_i$ we have the same set of connected components in $\tilde{G}_i$. Now, since contracting the connected components $\mathcal{C}$ in $G^{heavy}_i$ only distorts distances by $\delta 2^i$ by \Cref{def:heavyLightFramework} and distances in $G^{heavy}_i$ are preserved up to an $\alpha$-factor in $\tilde{G}_i$, we have that the final shortest path has weight at most $(1+\epsilon/2)\dist_G(s,t) + \alpha\delta2^i = (1+\epsilon/2)\dist_G(s,t) + \frac{\alpha\epsilon}{2\alpha} \cdot2^i =  (1+\epsilon)\dist_G(s,t)$, as required. We observe that since a path between two vertices in $G$ can only be longer than their shortest path, we have that all paths are overestimates and at least one level $i$, we obtain a $(1+\epsilon)$-approximate shortest path. 

To bound the running time, we observe that for every vertex $w$ on the shortest path $\pi_{s,v}$ that corresponds to a connected component $C_i$, we run at most one Dijkstra computation on the graph $\tilde{G}_i$ which also only explores the edge set in $\tilde{G}_i[C_i]$ by the very definition of a connected component. Thus, every edge in $\tilde{G}_i$ is explored during a shortest path query at most once. Thus, by the classic bound on Dijkstra's algorithm, and the upper bound on the size of $\tilde{G}_i$, we can upper bound the query time for a single level $i$ by $\tilde{O}(n \log W)$. We observe that using the distance estimate we can determine directly a constant number of levels at which we can run shortest path queries and can be certain to find a $(1+\epsilon)$-approximate shortest path. Thus, the total query time is $\tilde{O}(n \log W)$. The time to maintain the spanners $\tilde{G}_i$ is subsumed in the total update time of both data structures. Finally, since all data structures used are either deterministic or work against an adaptive adversary, the resulting algorithm also works against an adaptive adversary. The following theorems follow.

\bernsteinExtension*

\gutenbergExtension*

\section{Dynamic Effective Resistances with Worst-case Update Time}
\label{sec:effective_resistance}

Using the results of \cite{DurfeeGGP18-older}, 
we can use our dynamic sparsifier to obtain a dynamic $st$-effective resistance algorithm. 
The idea is to first run a vertex-sparsifier (Schur complement), 
then run our edge-sparsifier (spectral sparsifier) on top of it. 
At the end we run a Laplacian solver on the sparsified graph to compute the resistance.

$\textsc{SC}(G,T)$ is the Schur-Complement of $G$ with terminals $T$. 
If $M = \textsc{SC}(G,T)$, then for any $u,v \in T$ the effective resistance 
between $u$ and $v$ in $M$ and in $G$ is identical.

\begin{lemma}[{\cite[Lemma 6.4]{DurfeeGGP18-older}}]\label{lem:approxSC}
Given an undirected multi-graph $G = (V,E)$ a subset of vertices $T$ and a trade-off parameter $\beta$ 
such that $\beta n = \Omega(\log n)$, 
we can maintain with high probability a $(1+\varepsilon)$-sparsifier of $\textsc{SC}(G,T' \cup T)$ 
where $|T \cup T'| = \Theta(n \beta)$ and $T'$ is a random subset of $V$. 
The algorithm supports the following updates:
\begin{itemize}
\item $\textsc{Initialize}(G,T,\beta)$ in $O(m \beta^{-3} \log^5 n \varepsilon^{-2})$ expected time.
\item $\textsc{Insert}(u,v)$ in $O(\beta^{-6} \log^9 n \varepsilon^{-2})$ expected time.
\item $\textsc{Delete}(u,v)$ in $O(\beta^{-6} \log^9 n \varepsilon^{-2})$ expected time.
\end{itemize}
Furthermore, each of these operations leads to a number of changes in $H$ bounded by the corresponding costs.
\end{lemma}

\begin{theorem}
Given an undirected graph $G$ and two fixed nodes $s$ and $t$, 
we can with high-probability maintain 
the approximate effective resistance between $s$ and $t$ 
supporting edge updates in $O(n^{6/7+o(1)})$ expected worst-case time, 
if $\varepsilon = \Omega(1/\polylog n)$. 
The pre-processing time is $\tilde{O}(mn^{3/7} + n^{6/7})$.
\end{theorem}

\begin{proof}
We start with the high-level idea: 
Given a graph $G$ and nodes $s,t$, 
we set $T = \{ s, t \}$ and maintain a $(1+\varepsilon)$-sparsifier $H$ 
of the approximate Schur-complement $\textsc{SC}(G,T' \cup T)$ via \Cref{lem:approxSC}. 
In general, this graph $H$ is not very sparse, 
so we sparsify $H$ with our dynamic sparsifier \Cref{thm:worst_case:spectral}
to obtain a sparser graph $\tilde{H}$.
We can now maintain the effective resistance between $s$ and $t$ 
via a fast approximate Laplacian solver.

\paragraph{Pre-processing}

Initialize \Cref{lem:approxSC} on the graph $G$ for $T = \{s,t\}$, 
so we obtain a $(1+\varepsilon)$-sparsifier $H$ of $\textsc{SC}(G,T' \cup T)$. 
We then also initialize our dynamic sparsifier \Cref{thm:worst_case:spectral}
on top of $H$ to obtain a sparser $\tilde{H}$.

\paragraph{Update}

Let $\beta$ be the trade-off parameter of \Cref{lem:approxSC}.
When $G$ receives an edge update, 
we compute the change in $H$ in $\tilde{O}(\beta^{-6})$ time. 
The graph $H$ may change by upto $\tilde{O}(\beta^{-6})$ edges, 
so we have to perform this many edge updates to \Cref{thm:worst_case:spectral}
in order to maintain $\tilde{H}$. 
Note that $H$ has $\Theta(n\beta)$ nodes, 
so each update of \Cref{thm:worst_case:spectral} requires only $O((\beta n)^{o(1)})$ time 
for a total of $O(\beta^{o(1)-6}n^{o(1)})$ time.

Next, we must compute the effective resistance.
The $st$-effective resistance is given by $(\vec{e}_s - \vec{e}_t)^\top L_H^\dagger (\vec{e}_s - \vec{e}_t)$, 
where $L_G$ is the Laplacian matrix of the graph $H$, 
$\dagger$ denotes the Moore-Penrose pseudo inverse 
and $\vec{e}_i$ is the $i$th standard unit-vector.
Computing this resistance is done by computing 
$L_H^\dagger (\vec{e}_s - \vec{e}_t)$ 
via a Laplacian solver in $\tilde{O}((\beta n)^{1+o(1)})$ time \cite{SpielmanT04}, 
as $H$ has $O((n \beta)^{1+o(1)})$ many edges.

\paragraph{Balancing the cost}

We choose $\beta = n^{-1/7}$ in which case $\beta^{-6} = \beta n = n^{6/7}$, 
so the update time requires $O(n^{6/7+o(1)})$ time.
The preprocessing requires $\tilde{O}(m \beta^{-3} + (\beta n)^2)$ time, 
where the first term is the preprocessing of \Cref{lem:approxSC} 
and the second term is the preprocessing of \Cref{thm:worst_case:spectral}
on the $\tilde{O}((\beta n)^2)$ sized graph.
For $\beta = n^{-1/7}$ this is $\tilde{O}(mn^{3/7} + n^{6/7})$.
\end{proof}

\section{Congestion Minimization and Multi-commodity Flow}

\label{sec:multiflow}
In this section, we show new algorithms for computing multi-commodity
flows. To avoid confusion, we call the maximum multi-commodity flow
problem as \emph{maximum throughput flow} problem.
\begin{theorem}
\label{thm:multiflow}In undirected vertex-capacitated graphs with
$n$ vertices and $m$ edges, there are
\begin{enumerate}
\item a $\polylog(n)$-approximate algorithm for maximum throughput flow with
$k$ commodities in $\tilde O(n^2)$ time, and
\item \textup{a $\polylog(n)$}-approximate algorithm for maximum concurrent
flow with $k$ commodities in $\tilde O((n+k)n\log C)$ time where
$C$ is the ratio of largest finite capacity to smallest finite capacity. 
\end{enumerate}
\end{theorem}

In fact, the algorithm above for maximum concurrent flow with $k$
commodities gives us an explicit flow-path decomposition of the multi-commodity
flow. Therefore, we obtain the following using standard randomized
rounding technique (proven in \Cref{sec:multiflow_proof}).
\begin{corollary}
\label{cor:cong min}In undirected vertex-capacitated graphs with
$n$ vertices and $m$ edges, there is a \textup{$\polylog(n)$}-approximate
algorithm for congestion minimization with $k$ demand pairs in $\tilde O((n+k)n\log C)$
time where $C$ is the ratio of largest finite capacity to smallest
finite capacity. 
\end{corollary}

At the high-level, our algorithms are based on multiplicative weight
update framework for computing flow problems as used previously in
\cite{GargK07,Fleischer00,Karakostas08,Madry10}. Our approach is
essentially the same as the approximate max flow algorithm in vertex-capacitated
graphs by Chuzhoy and Khanna \cite{ChuzhoyK19}. We give the proof
of \Cref{thm:multiflow} in \Cref{sec:multiflow_proof} as we just follow
the known technique in literature.

\section*{Acknowledgement}
We thank Julia Chuzhoy and Gramoz Goranci for discussions. 

This project has received funding from the European Research Council (ERC) under the European
Unions Horizon 2020 research and innovation programme under grant agreement No 715672. Danupon
Nanongkai was also partially supported by the Swedish Research Council (Reg. No. 2015-04659 and 2019-05622). He Sun is supported by EPSRC Early Career Fellowship~(EP/T00729X/1). Aaron Bernstein is supported by NSF Award 1942010 and the Simon's Group for Algorithms \& Geometry. Maximilian Probst Gutenberg is supported by Basic Algorithms Research Copenhagen (BARC), supported by Thorup's Investigator Grant from the Villum Foundation under Grant No. 16582.

\newpage
\part{Appendices}
\appendix

\section{Verifying Sparsifier Properties}
\label{sec:properties}

Here we verify that the sparsifiers considered in this paper
satisfy the conditions \property{} defined in \Cref{define_properties}.

\subsection{Spanners}
We start by verifying spanners, as distances are a very intuitive graph property,
so that verification of the properties should be the easiest to understand.

\spannerIsAGraphProblem*

\begin{proof}
We prove properties \property{} in sequence.
\paragraph{Perturbation Property}
Let $G$ be a graph and $G'$ be the same graph 
where each edge $e$ has its cost multiplied by some factor $1 \le f_e \le e^\epsilon$.
Because of $f_e \ge 1$ the distances in $G'$ can not be shorter than the distances in $G$.
Further, the distances in $G'$ can at most be larger by a factor of $e^\epsilon$ than the distances in $G$.
Thus we have $G' \in \H(G, \epsilon)$ and $e^\epsilon \cdot G \in \H(G', \epsilon)$,
so property \eqref{con:identity} is satisfied.

\paragraph{Union Property}
Let $G_1,...,G_k$ be some graphs, $ s_1,...,s_k \in \R_{\geq}\setminus \{0\}$, and $G = \bigcup_{i=1}^k s_i \cdot G_i$. 
Let $H_i \in \H(G_i, \epsilon)$ for all $i=1,...,k$. 
For any $s,t \in V$ and a shortest path in $G$ connecting them,
we can decompose the path into segments $p_1,p_2,\ldots$,  
where each $p_i$ is contained in some $s_{G_{j_i}} \cdot G_{j_i}$.
Let $v_i, v_{i+1}$ be the first and last node of $p_i$,
then $$\dist_G(s,t) 
= 
\sum_i s_{j_i} \dist_{G_{i_j}}(v_i, v_{i+1})
\le
\sum_i s_{j_i} \dist_{H_{i_j}}(v_i, v_{i+1})
\le
\sum_i s_{j_i} e^\epsilon \dist_{G_{i_j}}(v_i, v_{i+1})
=
e^\epsilon \dist_G(s,t),
$$
so we have $\bigcup_i s_i H_i \in \H(\bigcup_i s_i G_i, \epsilon)$ 
and property \eqref{con:union} is satisfied.

\paragraph{Transition Property}
Let $H_1,H_2 \in \H(G, \epsilon)$ and $H \subset H_1$,
then $H \cup H_2 \in \H(G, \epsilon)$,
because $\dist_{H \cup H_2}(s,t) \le \dist_{H_2}(s,t) \le e^\epsilon \dist_{G}(s,t)$
and further $\dist_G(s,t) \le \dist_{H_1 \cup H_2}(s,t) \le \dist_{H \cup H_2}(s,t)$,
where we use the previous property \eqref{con:union} 
and the fact that distances can only increase when removing edges.
So property \eqref{con:increment} is satisfied.

\paragraph{Transitivity Property}
Let $H \in \H(G, \epsilon)$ and $H' \in \H(H, \delta)$,
then for any pairs $s,t$ we have
$$
\dist_{G}(s,t)
\le
\dist_{H}(s,t)
\le
\dist_{H'}(s,t)
\le
e^\delta \dist_{H}(s,t)
\le
e^{\delta+\epsilon} \dist_{G}(s,t).
$$
thus property \eqref{con:nested} is satisfied.

\paragraph{Contraction Property}
At last, consider property \eqref{con:contraction}.
Let $H \in \H(G, \epsilon)$ and $W \subset V$. 
Let $G'$ be the graph $G$ after contracting the set $W$,
and likewise let $H'$ be the graph $H$ after contracting the same set of nodes.
For any $s,t \in V$ and some shortest path in $G$ connecting them,
let $u \in W$ be the first node in $W$ that is visited and $v$ be the last visited node in $W$.
Then 
\begin{align*}
\dist_{G'}(s,t) 
&= 
\dist_G(s,u) + \dist_G(v,t)
\le
\dist_H(s,u) + \dist_H(v,t)
=
\dist_{H'}(s, t) \\
&=
\dist_H(s,u) + \dist_H(v,t)
\le
e^\epsilon(\dist_G(s,u) + \dist_G(v,t))
=
e^\epsilon \dist_{G'}(s,t),
\end{align*}
so $H' \in \H(G', \epsilon)$ which means property \eqref{con:contraction} is satisfied.
\end{proof}

\subsection{Spectral Sparsifiers}

Next, we verify that spectral sparsifiers satisfy the properties \property.
Unlike the previous lemma, this proof is based on linear algebra.

\spectralSparsifierIsAGraphProblem*

\begin{proof}
For any vertex $u$, let $\vec{e}_u$ be the $u$-th standard unit vector, i.e., $\vec{e}_u(v)=1$ if $v=u$, and $\vec{e}_u(v)=0$ otherwise. Then, the Laplacian matrix $L_G$ of graph $G=(V,E)$ can be written as 
$$L_G = \sum_{\{u,v\} \in E} w_G(u,v) (\vec{e}_u - \vec{e}_v) (\vec{e}_u - \vec{e}_v)^\top,$$
where $w_G(u,v)$ is the weight of edge $\{u,v\}$.
A spectral sparsifier of $G$, is a subgraph $H$ with edge weights $w_H$ such that for all vectors $\vec f \in \R^{|V|}$ we have
$$
e^{-\epsilon} (\vec f)^\top L_G \vec f \le (\vec f)^\top L_H \vec f \le e^{\epsilon} (\vec f)^\top L_G \vec f.
$$

\paragraph{Pertubation Property}

Property \eqref{con:identity} is satisfied, because for any $G'$ with the same edges as $G$, but scaled by up to $e^{\pm \epsilon}$ satisfies
$$
\vec f^\top L_G \vec f
= 
\sum_{\{u,v\} \in E} w_G(u,v) (\vec{f}_u - \vec{f}_v)^2
\le
e^{\epsilon}
\sum_{\{u,v\} \in E} w_{G'}(u,v) (\vec{f}_u - \vec{f}_v)^2
=
e^{\epsilon}
\vec f^\top L_{G'} \vec f,
$$
$$
\vec f^\top L_{G'} \vec f
= 
\sum_{\{u,v\} \in E} w_{G'}(u,v) (\vec{f}_u - \vec{f}_v)^2
\le
e^{\epsilon}
\sum_{\{u,v\} \in E} w_G(u,v) (\vec{f}_u - \vec{f}_v)^2
=
e^{\epsilon}
\vec f^\top L_{G} \vec f.
$$

\paragraph{Transitivity Property}

For $H \in \H(G,\epsilon)$ and $H' \in \H(H,\delta)$,
it holds that $H' \in \H(G,\epsilon+\delta)$, since
$$
\vec f^\top L_{H'} \vec f
\le
e^{\delta} \vec f^\top L_{H} \vec f
=
e^{\epsilon}e^{\delta}
\vec f^\top L_{G} \vec f,
$$
and similarly it holds that 
  $\vec f^\top L_{H'} \vec f \ge e^{-(\epsilon+\delta)} \vec f^\top L_{G} \vec f$.
This implies that Property~\eqref{con:nested} is satisfied.
  
\paragraph{Union Property}
 
To study Property~\eqref{con:union}, notice that, 
for graphs $G_1,...,G_d$ and scalars $s_1,...,s_d \in \R$, 
we could define $G := \bigcup_i G_i$ and have that $L_{G} = \sum_i s_i L_{G'}$. 
Therefore Property~\eqref{con:union} is satisfied.

\paragraph{Transition Property}
  
For $H_1,H_2 \in \H(G, \epsilon)$ and $H \subset H_1$ 
we have 
\begin{align*}
x^\top L((e^{\delta}-1) H \cup H_2) x  &
= 
(e^{\delta}-1) x^\top L_H x  + x^\top L_{H_2}x 
\le 
(e^{\delta}-1) x^\top L_{H_1} x + x^\top L_{H_2} x\\
& 
\le 
e^\epsilon ((e^{\delta}-1) x^\top L_G x + x^\top L_G x)
=
e^{\epsilon + \delta} x^\top L_Gx
\end{align*}
where we used that graph Laplacians are PSD, so $x^\top L_H x \le x^\top L_H x + x^\top L_{H_1 \setminus H} x = x^\top L_{H_1} x$.
Conversely, we have
$$
x^\top L((e^{\delta}-1) H \cup H_2) x 
=
(e^{\delta}-1) x^\top L_H x  + x^\top L_{H_2}x 
\ge
x^\top L_{H_2}x 
\ge
x^\top L_{G}x.
$$
In summary, we obtain that $(e^\delta - 1) H \cup H_2 \in \H(G, \epsilon+\delta)$, 
so property \eqref{con:increment} is true.

\paragraph{Contraction Property}

And lastly for \eqref{con:contraction} consider the following.
Let $G = (V,E)$ be a graph and $H \in \H(G, \epsilon)$.
Let $G' = (V', E')$ be a graph obtained from $G$ by contracting some set of nodes $X \subset V$
(so $V' = (V \setminus X) \cup \{ X \}$),
and let $H'$ be the graph obtained from $H$ when contracting the same set of nodes.
For any $f' \in \R^{V'}$ let $f \in \R^V$ by setting $f_v = f'_v$ for $v \in V \setminus X$ and $f_v = f'_X$ for $v \in X$.
Further let $w_G(u,v) = 0$ if edge $\{u,v\}$ does not exist. Then
\begin{align*}
\lefteqn{f'^\top L_{G'} f'}\\
&=
\sum_{u',v' \in V'} w_{G'}(u',v') (\vec{f'}_{u'} - \vec{f'}_{v'})^2 \\
&=
\sum_{u',v' \in V\setminus X} w_{G'}(u',v') (\vec{f'}_{u'} - \vec{f'}_{v'})^2
+
\sum_{v' \in V\setminus X} w_{G'}(X,v') (\vec{f'}_{X} - \vec{f'}_{v'})^2
+
w_{G'}(X,X) (\vec{f'}_{X} - \vec{f'}_{X})^2 \\
&=
\sum_{u,v \in V\setminus X} w_{G}(u,v) (\vec{f}_{u} - \vec{f}_{v})^2
+
\sum_{u \in X,v \in V\setminus X} w_{G}(u,v) (\vec{f}_{u} - \vec{f}_{v})^2
+
\sum_{u,v \in X} w_{G}(u,v) (\vec{f}_{u} - \vec{f}_{v})^2 \\
&=
\sum_{u,v \in V} w_{G}(u,v) (\vec{f}_{u} - \vec{f}_{v})^2 \\
&=
f^\top L_G f
\end{align*}
Hence we have
$$
f'^\top L_{G'} f' = f^\top L_G f \le e^\epsilon f^\top L_H f = f'^\top L_{H'} f'
$$
and likewise $f'^\top L_{G'} f' \ge e^{-\epsilon} f'^\top L_{H'} f'$, so $H' \in \H(G', \epsilon)$.
Thus \eqref{con:contraction} is satisfied.
\end{proof}

\subsection{Cut Sparsifiers}

The last remaining sparsifier, that is considered in this paper,
are cut-sparsifiers.
The proof for cut-sparsifiers is based on the previous proof of spectral-sparsifiers.

\cutSparsifierIsAGraphProperty*

\begin{proof}
A cut sparsifier $H$ of $G$ is a graph with the property that for any subset $W \subset V$
we have
$$\delta_G(W) \le \delta_H(W) \le e^\epsilon \delta_H(W).$$
If we define for a set $W \subset V$ a vector $f_v = 1$ for $v \in W$ and $f_v = 0$ otherwise,
then $f^\top L_G f$ is exactly the size of the cut in $G$, when $L_G$ is the Laplacian of $G$.
Hence we can restate the cut sparsifier property as
$$f^\top L_H f \le f^\top L_G f \le e^\epsilon f^\top L_H f$$
for all vector $f \in \{0,1\}^V$.

The proof of \Cref{lma:cutSparsifierIsAGraphProblem}
is now identical to the proof of \Cref{lem:spectralSparsifierIsAGraphProblem},
except that we restrict the vector to come from $\{0,1\}^V$.
\end{proof}

\section{Implications of Problem Properties}

Here we prove some facts by applying certain combinations of the properties above. 
Our first observation shows that edge insertions are easy to handle. 
Specifically, assume that we have a sparsifier $H$ of some graph $G$ 
and some edge $e$ is inserted into $G$, 
then $H\cup \{e\}$ is a valid sparsifier of $G\cup \{e\}$.

\begin{lemma}[Insertion Lemma]\label{lem:insertion}
If $\H$ satisfies \eqref{con:identity} and \eqref{con:union},
then for any $H \in \H(G,\epsilon)$ and $G'$ we have $H \cup G' \in \H(G \cup G', \epsilon)$.
In particular, it holds for any edge $e$ that  $H \cup e \in \H(H \cup e, \epsilon)$.
\end{lemma}

\begin{proof}
By \eqref{con:identity} we have $G' \in \H(G', \epsilon)$, 
so by \eqref{con:union} $H \cup G' \in \H(G \cup G', \epsilon)$.
The second claim follows by interpreting edge $e$ as a graph on two nodes.
\end{proof}

The next observation states that,
given a union of $k \approx 1/\epsilon$ many $\epsilon$-accurate sparsifiers $H_1,...,H_k$ of some graph $G$,
one can replace $H_1$   with any subgraph, 
and then the union of all $H_i$ is still a valid sparsifier of $G$. 
This observation allows us to interpolate between different sparsifiers of $G$, 
by slowly removing or inserting edges of of some sparsifier.

\begin{lemma}[Interpolation Lemma]\label{lem:interpolation}
Assume $\H$ satisfies \eqref{con:identity}, \eqref{con:union}, \eqref{con:increment}, and \eqref{con:nested}.
For $3+(e^{\epsilon/2}-1)^{-1}$ many $H_i \in \H(G,\epsilon)$ 
and some $H \subset H' \in \H(G,\epsilon)$, we have that
$$
\frac{1}{3+(e^{\epsilon/2}-1)^{-1}} \cdot \left(H \cup \bigcup_{i=1}^c H_i \right) \in \H(G, \epsilon).
$$
\end{lemma}

\begin{proof}
Define $c = 3+(e^{\epsilon/2}-1)^{-1}$, $\Delta = 1/c$,
and let $\delta$ be the parameter  such that $e^{\delta}-1 = \Delta$.
Then we have
$$
\Delta \cdot \bigcup_{i=1}^c H_i \in \H(c \cdot \Delta \cdot G, \epsilon) = \H(G, \epsilon/2)
$$
by property \eqref{con:union}.
This then also implies
$$
\Delta \cdot \left(H \cup \bigcup_{i=1}^c H_i\right) \in \H(G, \epsilon/2+\delta)
$$
by property \eqref{con:increment} and the definition of $\delta$. 
As $\Delta < e^{\epsilon/2}-1$ we also know that $\delta \le \epsilon / 2$,
so $\H(G, \epsilon/2+\delta) \subset \H(G, \epsilon)$ by \eqref{con:nested}.
In summary, we obtain
$$
\frac{1}{3+(e^{\epsilon/2}-1)^{-1}} \cdot \left(H \cup \bigcup_{i=1}^c H_i \right) \in \H(G, \epsilon).
$$
\end{proof}

\section{Omitted Proofs about Expanders}
\label{sec:omit}

\subsection{Proof of \Cref{prop:explicit expander}}
\begin{proof}
The expander construction by Margulis, Gabber and Galil is as follows.
For any number $k$, the graph $H'_{k^{2}}$ has the vertex set $\mathbb{Z}_{k}\times\mathbb{Z}_{k}$
where $\mathbb{Z}_{k}=\mathbb{Z}/k\mathbb{Z}$. For each vertex $(x,y)\in\mathbb{Z}_{k}\times\mathbb{Z}_{k}$,
its eight adjacent vertices are $(x\pm2y,y),(x\pm(2y+1),y),(x,y\pm2x),(x,y\pm(2x+1))$.
In \cite{GabberG81}, it is shown that $\Phi_{H'_{k^{2}}}=\Omega(1)$. 

Let $k$ be such that $(k-1)^{2}<n\le k^{2}$. Note that $n\ge10$,
so $k\ge4$, and so $(k-1)^{2}\ge k^{2}/2$. So we can contract disjoint
pairs of vertices in $H'_{k^{2}}$ and obtain a graph $H''_{n}$ with
$n$ vertices where each vertex has degree between $8$ and $16$. Note
that $\Phi_{H''_{n}}= \Omega\left(\Phi_{H'_{k^{2}}}\right)$. 
Let $t=\left\lfloor d/8\right\rfloor $
then $H_{n,d}$ is just a union of $t$ many copies of $H''_{n}$. So each
node in $H_{n,d}$ has degree at least $8t\ge d-8$ and at most $16t\le2d$.
Note that $\Phi_{H_{n,d}}=\Omega\left(\Phi_{H''_{n}}\right)$. It is clear that the
construction takes $O(nd)$ total time.
\end{proof}
 
\subsection{Proof of \Cref{prop:delta-reduction prop}}

\begin{proof}
All properties except the last one are clear. It remains to prove
the last property about the conductance of a graph. We first show
that $\Phi_{G'}=O(\Phi_{G})$. Let $(S,V \setminus S)$ be a minimum conductance
cut in $G$ with $\Phi_{G}(S)=\Phi_{G}$. Let $S'=\bigcup_{u\in S}X_{u}$.
We have that $\delta_{G}(S)=\delta_{G'}(S')$. Also, $\vol_{G}(S)=\Theta(\vol_{G'}(S'))$
and $\vol_{G}(V \setminus S)=\Theta(\vol_{G'}(V' \setminus S'))$. So 
\[
\Phi_{G'}\le\Phi_{G'}(S')=\Theta(\Phi_{G}(S))=\Theta(\Phi_{G}).
\]
Next, consider any cut $(S',V' \setminus S')$ in $G'$. We will show that either
$\Phi_{G'}(S')=\Omega(\Phi_{G})$ or there is a cut $(T,V \setminus T)$
in $G$ such that $\Phi_{G'}(S')=\Omega(\Phi_{G}(T))$. This will
prove that $\Phi_{G'}=\Omega(\Phi_{G})$ which will conclude the proof. 
Below
we treat $X_{u}$ as an expander itself and sometimes as a set of
nodes in $G'$. Assume without loss of generality  that $\vol_{G'}(S')\le\vol_{G'}(V' \setminus S')$.
Let $A=\{u\mid0<\vol_{G'}(S'\cap X_{u})\le2\vol_{G'}(X_{u}-S')\}$
be the set of nodes $u$ in $G$ where $S'$ intersects with $X_{u}$
but the overlap is at most $2/3$ with respect to the volume in $G'$. 

Below, we let $a\lesssim b$ to denote $a=O(b)$. There are two cases.
Let $\epsilon$ be a small constant to be chosen later. 

In the first case, suppose $\sum_{u\in A}\vol_{X_{u}}(S'\cap X_{u})\ge\epsilon\Phi_{G}\vol_{G'}(S')$.
Observe that $\delta_{X_{u}}(S'\cap X_{u})=\Omega(\vol_{X_{u}}(S'\cap X_{u}))$
because 
\[
\vol_{X_{u}}(S'\cap X_{u})\le\vol_{G'}(S'\cap X_{u})\underset{u\in A}{\le}2\vol_{G'}(X_{u} \setminus S')=O(\vol_{X_{u}}(X_{u} \setminus S'))
\]
 and $\Phi_{X_{u}}=\Omega(1)$ by \Cref{prop:explicit expander}. Therefore,
$\Phi_{G'}(S')=\Omega(\epsilon\Phi_{G})$ as 

\[
\delta_{G'}(S')\ge\sum_{u\in A}\delta_{X_{u}}(S'\cap X_{u})\ge\sum_{u\in A}\Omega(\vol_{X_{u}}(S'\cap X_{u}))=\Omega(\epsilon\Phi_{G}\vol_{G'}(S')).
\]

For the second case, suppose $\sum_{u\in A}\vol_{X_{u}}(S'\cap X_{u})\le\epsilon\Phi_{G}\vol_{G'}(S')$.
Let $T=\{u\mid\vol_{G'}(S'\cap X_{u})>2\vol_{G'}(X_{u}-S')\}$, $T'=\bigcup_{u\in T}X_{u}$,
$\overline{T'}=V' \setminus T'$ and $\overline{S'}=V' \setminus S'$. We will show that
(1) $\delta_{G'}(T')=O(\delta_{G'}(S')+\epsilon\Phi_{G}\vol_{G'}(S'))$,
(2) $\vol_{G'}(S')=O(\vol_{G'}(T'))$, and (3) $\vol_{G'}(\overline{S'})=O(\vol_{G'}(\overline{T'}))$.
This is enough because 
\begin{align*}
\Phi_{G'}(T') & =\frac{\delta_{G'}(T')}{\min\{\vol_{G'}(T'),\vol_{G'}(V'\setminus T')}\\
 & \lesssim\frac{\delta_{G'}(S')+\epsilon\Phi_{G}\vol_{G'}(S')}{\min\{\vol_{G'}(S'),\vol_{G'}(V'\setminus S')\}}\\
 & =\frac{\delta_{G'}(S')+\epsilon\Phi_{G}\vol_{G'}(S')}{\vol_{G'}(S')}\\
 & =\Phi_{G'}(S)+\epsilon\Phi_{G}.
\end{align*}
As $\Phi_{G}(T)=\Theta(\Phi_{G'}(T'))$, there is some constant $C$
such that  $\Phi_{G}(T)\le C(\Phi_{G'}(S')+\epsilon\Phi_{G}(T))$. By choosing
$\epsilon=1/2C$, we have that $\Phi_{G}(T)\le2C\cdot\Phi_{G'}(S')$
and so $\Phi_{G'}(S')=\Omega(\Phi_{G})$. In both cases, we have that
$\Phi_{G'}(S')=\Omega(\Phi_{G})$. Now, it remains to prove the three
claims above. 
\begin{claim}
We have the following: 
\end{claim}

\begin{itemize}
\item $\delta_{G'}(T')=O(\delta_{G'}(S')+\epsilon\Phi_{G}\vol_{G'}(S'))$,
\item $\vol_{G'}(S')=O(\vol_{G'}(T'))$, and 
\item $\vol_{G'}(\overline{S'})=O(\vol_{G'}(\overline{T'}))$.
\end{itemize}
\begin{proof}
It is convenient to bound $\vol_{G'}(T'\setminus S')$ and $\vol_{G'}(S' \setminus T')$
first. We have 
\begin{align*}
\vol_{G'}(T'\setminus S') & =\sum_{u\in T}\vol_{G'}(X_{u} \setminus S')\\
 & \lesssim\sum_{u\in T}\vol_{X_{u}}(X_{u} \setminus S')\\
 & \lesssim\sum_{u\in T}\delta_{X_{u}}(X_{u} \setminus S')\\
 & \le\delta_{G'}(\overline{S'})=\delta_{G'}(S').
\end{align*}
We also have 
\begin{align*}
\vol_{G'}(T' \setminus S') & =\sum_{u\in T}\vol_{G'}(X_{u} \setminus S')\\
 & <\sum_{u\in T}\vol_{G'}(X_{u}\cap S')/2\\
 & =\vol_{G'}(T'\cap S')/2.
\end{align*}
Next, 
\begin{align*}
\vol_{G'}(S'\setminus T') & =\sum_{u\in A}\vol_{G'}(S'\cap X_{u})\\
 & \lesssim\sum_{u\in A}\vol_{X_{u}}(S'\cap X_{u})\\
 & \le\epsilon\Phi_{G}\vol_{G'}(S').
\end{align*}
\end{proof}
So we have 
\begin{align*}
\delta_{G'}(T') & \le\delta_{G'}(S')+\vol_{G'}(T'\setminus S')+\vol_{G'}(S'\setminus T')\\
 & =O\left(\delta_{G'}(S')+\epsilon\Phi_{G}\vol_{G'}(S')\right).
\end{align*}
Next, we have, 
\begin{align*}
\vol_{G'}(S') & \le\vol_{G'}(T')+\vol_{G'}(S'\setminus T')\\
 & \le\vol_{G'}(T')+O(\epsilon\Phi_{G}\vol_{G'}(S'))\\
 & \le\vol_{G'}(T')+\vol_{G'}(S')/2,
\end{align*}
and so $\vol_{G'}(S')=O(\vol_{G'}(T'))$. Lastly, we have
\begin{align*}
\vol_{G'}(\overline{S'}) & \le\vol_{G'}(\overline{T'})+\vol_{G'}(T' \setminus S')\\
 & <\vol_{G'}(\overline{T'})+\vol_{G'}(T'\cap S')/2\\
 & \le\vol_{G'}(\overline{T'})+\vol_{G'}(S')/2\\
 & \le\vol_{G'}(\overline{T'})+\vol_{G'}(\overline{S'})/2,
\end{align*}
and so $\vol_{G'}(\overline{S'})=O(\vol_{G'}(\overline{T'}))$.
\end{proof}

\section{Amortized Spectral Sparsifier against Adaptive Adversaries}
\label{sec:spectral:amortized}
\label{sec:spectral:query}
 
Here we prove our dynamic spectral sparsifier result against adaptive adversaries.

\thmSpectralQuery*

While the query time is slow,
note that this data structure has some interesting applications.
For example one can create a dynamic Laplacian system solver against adaptive adversaries,
by running a Laplacian solver on top of the sparsifier.
Then any update to the Laplacian system takes $\tilde{O}(1)$ amortized update time,
and one can then solve the system for any vector $b \in \R^n$ in $\tilde{O}(n)$ time.
As this dynamic algorithm holds against adaptive adversaries,
it can be used inside iterative algorithms such as \cite{CohenMSV17}.

\paragraph{Proof of \Cref{thm:query:spectral_sparsifier}}
 
Recall that \Cref{cor:ST-decomposition} says that we obtain a spectral sparsifier 
by simply sampling every edge proportional to the degrees of the endpoints.
\Cref{def:amortized:decremental_algorithm} guarantees us that the graph is of near uniform degree
and that the degrees of our graphs stay roughly the same throughout all updates.
Thus we obtain a queray-algorithm by simply sampling the graph
whenever a query is performed.

\begin{lemma}\label{lem:query:spectral_decremental}
For every $\phi$ there exists a decremental \emph{query}-algorithm on almost-uniform-degree $\phi$-expanders 
(see \Cref{def:amortized:decremental_algorithm}) 
that maintains a $e^\epsilon$-approximate spectral sparsifier
against \emph{adaptive} adversaries.  
The algorithm's pre-processing time is bounded by $O(m)$, 
every edge deletion takes $O(1)$ time, 
and it takes the algorithm 
$O(n \epsilon^{-2} \phi^{-5} \log^2 n)$ time 
to output a spectral sparsifier.
\end{lemma}

\begin{proof}
The algorithm is very simple.
During the preprocessing we check for the minimum degree of the given graph $G$.
Let $\Delta$ be the degree.

When a query is performed we sample each edge of $G$ by some probability 
$p = \Theta(\epsilon^{-2} \phi^{-5} \Delta^{-1} \log^2 n)$
and when an edge is included, it is scaled by $1/p$.
Let $H$ be the resulting graph, then we return $H$ as the spectral sparsifier.

\paragraph{Correctness}

Based on \Cref{def:amortized:decremental_algorithm}
we are guaranteed that $G$ is always a $\phi$-expander and
that the minimum degree is always bounded by $O(\phi \Delta)$.
By \Cref{cor:ST-decomposition}, with high probability the graph $H$ is an $e^\epsilon$-accurate spectral sparsifier. 
Note that by independently resampling the edges with each query,
the sample graph is independent of the updates, 
so the algorithm only holds against an adaptive adversary.

\paragraph{Complexity}

By Markov's inequality, the size of the resulting sparsifier is bounded by
$O(n \epsilon^{-2} \phi^{-5} \log^2 n)$,
which also bounds the complexity for sampling the graph $H$.
The pre-processing takes only $O(m)$ time
and an edge deletion takes only $O(1)$ time.
\end{proof}

Using the reduction from \Cref{sec:uniform_degree_reduction},
we then obtain our dynamic spectral sparsifier with amortized update time. 

\thmSpectralQuery*

\begin{proof}
The result follows directly from \Cref{thm:amortized:fully_dynamic_weighted}
and \Cref{lem:query:spectral_decremental}.
We choose $\phi = \log^4 n$
so then the preprocessing time is bounded by 
$O(P(m)) = O(m)$, and the update time is bounded by
$$
O\left(
	\phi^{-3} \log^7 n \cdot T(n) 
	+ \frac{P(m)}{m \phi^{3}} \log^{7} n
\right)
=
O\left(
	\log^{19} n 
\right).
$$
The time per query is bounded by
$$
O(Q(n \log^3 n) \epsilon^{-1} \log W)
=
O((n \log^3 (n) \epsilon^{-2} \phi^{-5} \log^2 (n))\cdot( \epsilon^{-1} \log W))
=
O(n \log^{25} (n) \epsilon^{-3} \log W).
$$
\end{proof}

\section{Proofs of \Cref{thm:multiflow} and \Cref{cor:cong min}}

\label{sec:multiflow_proof}

Below, we define some notations and then discuss the previous works on these problems.

\paragraph{Flow.}

We first define some basic notions. Although our algorithms are for
undirected vertex-capacitated graphs, we will define the problems
also on directed graphs and also edge capacitated graphs. Let $G=(V,E)$
be a directed graph,  and $s,t\in V$. An $s$-$t$ flow
$f\in\mathbb{R}_{\ge0}^{E}$ is such that, for any $v\in V\setminus\{s,t\}$,
the amount of flow into $v$ equals the amount of flow out of $v$,
i.e., $\sum_{(u,v)\in E}f(u,v)=\sum_{(v,u)\in E}f(v,u)$. Let $f(v)=\sum_{(u,v)\in E}f(u,v)$
be the \emph{amount of flow at $v$}. The \emph{value} $|f|$ of $f$
is $\sum_{(s,v)\in E}f(s,v)-\sum_{(v,s)\in E}f(v,s)$. Let $c\in(\mathbb{R}_{>0}\cup\{\infty\})^{E}$
denote edge capacities. A flow $f$ is \emph{edge-capacity-feasible}
if $f(e)\le c(e)$ for all $e\in E$. We can define capacities of
vertices as well. Let $c\in(\mathbb{R}_{>0}\cup\{\infty\})^{V}$ be
vertex capacities. We say that $f$ is \emph{vertex-capacity-feasible} if $f(v)\le c(v)$
for all $v\in V$. If $G$ is undirected, one way to define an $s$-$t$
flow is by treating $G$ as a directed graph where there are two directed
edge $(u,v)$ and $(v,u)$ for each undirected edge $\{u,v\}$. We
will assume that a flow only goes through an edge in one direction,
i.e., for each edge $\{u,v\}\in E$, either $f(u,v)=0$ or $f(v,u)=0$. 

For $1\le i\le k$, let $f_{i}$ be an $s_{i}$-$t_{i}$ flow. We
call $\F=\{f_{1},\dots,f_{k}\}$ a \emph{multi-commodity flow} or
\emph{$k$-commodity flow}. We call the $k$ pairs of vertices $(s_{1},t_{1}),\dots,(s_{k},t_{k})$
the \emph{demand pairs} of $\F$ and $\{s_{1},t_{1},\dots,s_{k},t_{k}\}$
a set of \emph{terminal vertices }of $\F$. We say that $\F$ is \emph{edge-capacity-feasible
}if $\sum_{i}f_{i}(e)\le c(e)$ for each $e\in E$, and is \emph{vertex-capacity-feasible}
if $\sum_{i}f_{i}(v)\le c(v)$ for all $v\in V$. 

\paragraph{Flow problems.}

In the \emph{maximum concurrent flow }problem with vertex capacities,
we are given a graph $G$ with capacities on vertices, and for $1\le j\le k$,
a demand pair $(s_{j},t_{j})$ and a \emph{target demand} $d_{i}$.
We need to find $\F=\{f_{1},\dots,f_{k}\}$ w.r.t. the demand pairs
$(s_{1},t_{1}),\dots,(s_{k},t_{k})$ such that $\F$ is vertex-capacity-feasible
and the value $|f_{j}|=\lambda d_{j}$ for each $j$. The goal is
to maximize $\lambda$. The congestion minimization problem is closely
related to this problem. In the \emph{congestion minimization} problem
with vertex capacities, we are given a graph $G$ with capacities
on vertices, and for $1\le j\le k$, a demand pair $(s_{j},t_{j})$.
We need to find a collection of \emph{paths }$\P=\{P_{1},\dots,P_{k}\}$
where $P_{i}$ is an $s_{i}$-$t_{i}$ path. The goal is to minimize
$\max_{v}\frac{|\{i\mid v\in P_{i}\}|}{c(v)}$, i.e. the maximum congestion
over all vertices $v$.

In the\emph{ maximum throughput flow }problem with vertex capacities,
we are given a graph $G$ with capacities on vertices, and a demand
pair $(s_{j},t_{j})$ for $1\le j\le k$. We need to find, for all
$j\in\{1,\dots,k\}$, an $s_{j}$-$t_{j}$ flow $f_{j}$ such that
$\F=\{f_{1},\dots,f_{k}\}$ is vertex-capacity-feasible. The goal
is to maximize the total value $\sum_{j}|f_{j}|$ of $\F$.

For convenience, we will assume that each terminal vertex has capacity
$\infty$. (In fact, this is a standard definition of the (single-commodity)
max flow problem.) If we want to also enforce capacity bounds at terminal
vertices, we can just add dummy terminal vertices $\{s'_{1},t'_{1},\dots,s'_{k},t'_{k}\}$,
add edges between $(s'_{i},s_{i})$ and $(t_{i},t'_{i})$ for each
$i$. The capacity at dummy terminal vertices are $\infty$, but we
can now enforce capacity bounds at terminal vertices.

The problems with edge capacities are defined analogously. 

\paragraph{Previous works.}

To see the context of the results, we need to discuss the result in
edge-capacitated graphs as well. 
Below, $n$ is the number of vertices, $m$ is the number of edges,
and $k$ is the number of demand pairs. We hide $\poly\log(nC)$ factors
where $C$ is the ratio between largest and smallest capacity.

In directed edge-capacitated graphs, Garg and K\"{o}nemann \cite{GargK07}
present $(1+\epsilon)$-approximation algorithms based on the multiplicative
weight update framework for maximum throughput and concurrent flow,
and many other problems. Then, Fleischer \cite{Fleischer00} gives
an improved maximum throughput flow algorithms with running time $\tilde{O}(m^{2}/\epsilon^{2})$.
She also improved maximum concurrent flow algorithms of \cite{GargK07},
but this is further improved by Karakostas \cite{Karakostas08} to
running time $\tilde{O}(m^{2}+kn)/\epsilon^{2})$ and only $\tilde{O}(m^{2}/\epsilon^{2})$
if only the value of optimal solution is needed. All these algorithm
are deterministic. Madry \cite{Madry10} observes that this framework
can be sped up nicely using dynamic algorithms. In \cite{Madry10},
he exploits a variant of dynamic all-pairs shortest paths algorithms,
and shows that a factor $m$ in previous running time can be replaced
by $n$. More precisely, he shows an $\tilde{O}(mn/\epsilon^{2})$-time
maximum throughput flow algorithm, and an $\tilde{O}((m+k)n/\epsilon^{2})$-time
maximum concurrent algorithm. His algorithms are randomized. When
$k$ is small, there are faster algorithms on undirected graphs: an
$\tilde{O}(mk/\epsilon)$-time max-concurrent flow algorithm by Sherman
\cite{Sherman17}, and an $\tilde{O}(m^{4/3}\poly(k/\epsilon))$-time
max-throughput flow algorithm by Kelner, Miller, and Peng \cite{KelnerMP12}.
Also in undirected graphs, by paying a large $O(\log^{4}n)$-approximation
factor, we can reduce both problems to the same problems on trees
by using the \emph{R\"{a}cke tree} which is from the context of oblivious
routing scheme. This reduction can be done in $\tilde{O}(m)$ \cite{RackeST14,Peng16}.
If the graph is a tree, it is not hard to see that both problems can
be solved in $\tilde{O}(m+k)$ time using, for example, the link-cut
tree \cite{SleatorT83}. Therefore, we obtain $O(\log^{4}n)$-approximation
algorithms in $\tilde{O}(m+k)$ time which is optimal up to poly-logarithmic factors.

Now, we turn to vertex-capacitated graphs. It turns out that the algorithms
by \cite{GargK07,Fleischer00,Karakostas08} can be sped up easily.
From the framework, the number of \emph{saturated augmentations} which
is $\tilde{O}(m/\epsilon^{2})$ in edge-capacitated graphs can be
reduced to $\tilde{O}(n/\epsilon^{2})$ in vertex-capacitated graphs.
Because of this, we can obtain a maximum throughput flow algorithm
with running time $\tilde{O}(mn/\epsilon^{2})$ from Fleischer's \cite{Fleischer00}
and a maximum concurrent flow algorithm with running time $\tilde{O}((m+k)n/\epsilon^{2})$.
These are deterministic, work in directed graphs, and, at the same
time, match the running time of algorithms by Madry \cite{Madry10}.
The recent algorithm \cite{ChuzhoyK19} for computing $(1+\epsilon)$-approximate
(single-commodity) max flow in vertex capacitated graph in $O(n^{2+o(1)}/\epsilon^{O(1)})$
expected time also exploits this.

In our algorithms, we reduce another $m$ factor to $\tilde O(n)$
by working on dynamic spanners of size $\tilde O(n)$ instead. Hence,
the time bound we obtain is $\tilde O(n^2/\epsilon^{2})$ and $\tilde O(n(n+k)/\epsilon^{2})$
for max-throughput and max-concurrent flow respectively. 

It is not clear how to obtain an  $o(n^{2})$-bound, even when the graph
is sparse, $k=1$, and we allow large approximation factor, say, $n^{o(1)}$.
For example, algorithms by Sherman \cite{Sherman17} and Kelner, Miller,
and Peng \cite{KelnerMP12} are quite specific to edge-capacitated
graphs. Also, there is provably no good quality oblivious routing
in vertex capacitated graphs \cite{HajiaghayiKRL07} like the R\"{a}cke
tree in edge-capacitated case.

\paragraph{Overview of the Approach.}
To avoid reproving the properties about algorithms based on multiplicative
weight update framework, we try to use the algorithms of Fleischer
\cite{Fleischer00} as a black-box as much as possible. However, her
algorithms are stated for directed edge-capacitated graphs. So, our
algorithms have three main steps as follows. 
\begin{enumerate}
\item We first show how to reduce the problem in undirected vertex-capacitated
graphs $G$ to directed edge-capacitated graphs $G''$. The reduction
is a very standard one.
\item Fleischer's algorithms mainly involve computing approximate shortest
paths in $G''$ where $G''$ undergoes edge weight updates. We show
how assign the weight to \emph{edges} of $G$ so that an approximate
shortest path in $G$ corresponds to an approximate shortest path in 
$G''$. Moreover, whenever when $G''$ is updated, the weight $G$
can be updated accordingly.
\item We will maintain a data structure from \Cref{thm:SpannerMainResult}  on $G$ that allows us to query a spanner $H$ of $G$ using our dynamic algorithm.
To find an approximate shortest path in $G$, we just first query for a spanner $H$ and then run Dijkstra's
in $H$, which is sparse.
\end{enumerate}

\subsection{Fleischer's algorithms}

We need some notations for describing Fleischer's algorithms. Let
$G=(V,E)$ be a directed graph with edge capacities $c\in(\mathbb{R}_{>0}\cup\{\infty\})^{E}$.
For any $s,t\in V$, let $\P_{j}$ be a set of all $s$-$t$ paths.
For any flow $f$ and path $P$, let $f(P)$ denote the amount of
flow in $f$ through path $P$. By writing $f(P)\gets f(P)+c$, this
means $f(e)\gets f(e)+c$ $\forall e\in P$. The algorithms maintain
lengths $\ell\in\mathbb{R}_{\ge0}^{E}$ on edges. For any path $P$,
let $\ell(P)=\sum_{e\in P}\ell(e)$ be the length of $P$. A $\ell$-shortest
$s$-$t$ path is a path $P^{*}=(s,\dots,t)$ which minimize $\ell(P)$
over all $\P_{s,t}$. An $\alpha$-approximate $\ell$-shortest path
$\tilde{P}$ is such that $\ell(\tilde{P})\le\alpha\cdot\ell(P^{*})$.

\subsubsection{Maximum throughput flow}

See \Cref{alg:throughput} for Fleischer's multi-commodity max-throughput
flow algorithm. We note that it is important to set $\ell(e)=0$ when
$c(e)=\infty$. This is not stated in  \cite{Fleischer00} because
they did not consider infinite capacity.

\begin{algorithm}
\textbf{Input: }A directed graph $G=(V,E)$ with edge capacities $c\in(\mathbb{R}_{>0}\cup\{\infty\})^{E}$,
the demand pair $(s_{1},t_{1}),\dots,(s_{k},t_{k})$ and an accuracy
parameter $0<\epsilon<1$.

\textbf{Output: }A multi-commodity flow $\F$ for $(s_{1},t_{1}),\dots,(s_{k},t_{k})$
where $\F$ is edge-capacity-feasible.

Let $\S=\{s_{1},\dots,s_{k}\}$ be the set of sources. 
\begin{enumerate}
\item Set $\delta=\frac{(1+\epsilon)}{((1+\epsilon)n)^{1/\epsilon}}$.
\item Set $\ell(e)=\delta$ if $c(e)$ is finite; otherwise $\ell(e)=0$.
Set $f\equiv0$.
\item \textbf{for} $r=1$ to $\left\lfloor \log_{1+\epsilon}\frac{1+\epsilon}{\delta}\right\rfloor $
\textbf{do}
\begin{enumerate}
\item \textbf{for each $s\in\S$ do}
\begin{enumerate}
\item $T\gets$$\alpha$-approximate $s$-source $\ell$-shortest paths
tree. Let $t$ be such that the distance from $s$ to $t$ in $T$
is minimized among all demand pair where $s$ is a source.
\item $P\gets$the path from $s$ to $t$ in $T$.
\item \textbf{while $\ell(P)<\min\{1,\delta(1+\epsilon)^{r}\}$}
\begin{enumerate}
\item $c\gets\min_{e\in P}c(e)$
\item $f_{j}(P)\gets f_{j}(P)+c$ where $(s_{j},t_{j})=(s,t)$
\item $\ell(e)\gets\ell(e)(1+\frac{\epsilon c}{c(e)})$ $\forall e\in P$
\item $T\gets$$\alpha$-approximate $s$-source $\ell$-shortest paths
tree. Let $t$ be such that $(s,t)$ is a demand pair and the length
from $s$ to $t$ in $T$ is minimized.
\item $P\gets$the path from $s$ to $t$ in $T$.
\end{enumerate}
\end{enumerate}
\end{enumerate}
\item \textbf{return} $\F=\{f_{j}/\log_{1+\epsilon}\frac{1+\epsilon}{\delta}\}_{j=1}^{k}$.
\end{enumerate}
\caption{Fleischer's algorithm \cite{Fleischer00} for computing multi-commodoity
max-throughput flow in directed edge-capacitated graphs.\label{alg:throughput}}
\end{algorithm}

Now, we state some properties of \Cref{alg:throughput}. Observe the
following:
\begin{lemma}
For every edge $e$ where $\ell(e)\neq0$, $\ell(e)$ only increases,
and $\delta\le\ell(e)\le(1+\epsilon)$. \label{lem:throughput weight}
\end{lemma}

\begin{lemma}
[Lemma 2.2 of \cite{Fleischer00}]The returned $\F$ is edge-capacity-feasible.
\label{lem:throughput feasible}
\end{lemma}

In \cite{Fleischer00}, $\alpha=(1+\epsilon)$. We observe the analysis
of Theorem 2.4 in \cite{Fleischer00} shows that the following holds
for general $\alpha$:
\begin{lemma}
The returned $\F$ is an $(\alpha(1+O(\epsilon))$-approximate solution.
\label{lem:throughput optimal}
\end{lemma}

\begin{lemma}
Suppose $G$ has at most $\tau$ edges with finite capacity. Before
returning $\F$, then are at most $\tilde{O}((\tau+\min\{k,n\})\epsilon^{-2})$
approximate single-source shortest path computations. \label{lem:throughput time}
\end{lemma}

\begin{proof}
[Proof sketch]A straight-forward extension of Lemma 2.1 of \cite{Fleischer00}
shows that the total number of times that we iterate in the while
loop is at most $\tilde{O}(\tau/\epsilon^{2})$, each iteration we
compute one single-source shortest path. In \cite{Fleischer00}, it
is also shown that the number of time we compute single-source shortest
paths outside the while loop is at most $\tilde{O}(|\S|/\epsilon^{2})=O(\min\{k,n\})\epsilon^{-2})$.
\end{proof}

\subsubsection{Maximum concurrent flow}

See \Cref{alg:concurent} for Fleischer's maximum concurrent flow algorithm. 

\begin{algorithm}
\textbf{Input: }A directed graph $G=(V,E)$ with edge capacities $c\in(\mathbb{R}_{>0}\cup\{\infty\})^{E}$,
demand pairs $(s_{1},t_{1}),\dots,(s_{k},t_{k})$, target demands
$d_{1},\dots,d_{k}$ and an accuracy parameter $0<\epsilon<1$.

\textbf{Output: }A multi-commodity flow $\F=\{f_{1},\dots,f_{k}\}$
where $f_{i}$ is an $s_{i}$-$t_{i}$ flow and $\F$ is edge-capacity-feasible.

\textbf{Assumption: }The optimal value of the solution $\beta$  satisfies that $\beta\geq1$.
\begin{enumerate}
\item Set $\delta=(2m)^{-1/\epsilon}$
\item Set $\ell(e)=\delta/c(e)$ $\forall c\in E$. Set $f\equiv0$.
\item \textbf{while} $\sum_{e\in E}c(e)\ell(e)<1$ \textbf{do}
\begin{enumerate}
\item \textbf{for} $j=1$ to $k$ \textbf{do}
\begin{enumerate}
\item $d_{j}'\gets d_{j}$
\item \textbf{while} $\sum_{e\in E}c(e)\ell(e)<1$ and \textbf{$d'_{j}>0$
do}
\begin{enumerate}
\item $P\gets\alpha$-approximate $\ell$-shortest $s_{j}$-$t_{j}$ path 
\item $c\gets\min\{\min_{e\in P}c(e),d'_{j}\}$
\item $d'_{j}\gets d'_{j}-c$
\item $f_{j}(P)\gets f_{j}(P)+c$
\item $\ell(e)\gets\ell(e)(1+\frac{\epsilon c}{c(e)})$ $\forall e\in P$
\end{enumerate}
\end{enumerate}
\end{enumerate}
\item \textbf{return} $\F=\{f_{j}/\log_{1+\epsilon}\frac{1}{\delta}\}_{j=1}^{k}$.
\end{enumerate}
\caption{Fleischer's algorithm \cite{Fleischer00} for computing a  maximum concurrent
flow in directed edge-capacitated graphs.\label{alg:concurent}}
\end{algorithm}

Next, we state some properties of \Cref{alg:throughput}. Observe the
following:
\begin{lemma}
For every edge $e$ where $\ell(e)\neq0$, $\ell(e)$ only increases,
and $\frac{\delta}{c(e)}\le\ell(e)\le\frac{1+\epsilon}{c(e)}$.\label{lem:concurrent weight}
\end{lemma}

Let $\beta$ be the value of optimal solution (i.e. the maximum $\lambda$).
\Cref{alg:concurent} assumes that $\beta\ge1$. This can be done easily
if we know the approximation of $\beta$. For example, a $\Delta$-approximation
$\tilde{\beta}$ of $\beta$ where $\beta/\Delta\le\tilde{\beta}\le\beta$
is given. Then, the problem where the target demands are $d_{1}/\Delta,\dots,d_{k}/\Delta$
must have the optimal value at least 1. 
\begin{lemma}
[Lemma 3.2 of \cite{Fleischer00}]The returned $\F$ is edge-capacity-feasible.\label{lem:concurrent feasible}
\end{lemma}

The analysis about approximation ratio in \cite{Fleischer00} is referred
to the analysis by Garg and K\"{o}nemann \cite{GargK07}. In \cite{GargK07},
$\alpha=1$. We observe for general $\alpha$, the analysis in Sections
5.1 and 5.2 \cite{GargK07} can be extended straightforwardly as follows:
\begin{lemma}
Suppose $\beta\ge1$. $\F$ is an $(\alpha(1+O(\epsilon))$-approximate
solution.\label{lem:concurrent optimal}
\end{lemma}

\begin{lemma}
[Lemma 5.2 of \cite{GargK07}]Suppose that $1\le\beta\le2$. Suppose
that $G$ has at most $\tau$ edges with finite capacity. Then there
are at most $\tilde{O}((\tau+k)/\epsilon^{2})$ approximate shortest
path computations in \Cref{alg:throughput}. 
\end{lemma}

In particular, if we have computed $\tilde{O}((\tau+k)/\epsilon^{2})$
many shortest paths and \Cref{alg:throughput} does not stop, then
$\beta\ge2$. In this case, we can scale the target demands up by
a factor of $2$ so that $\beta$ is reduced by a factor of $2$.
Since we  need to repeat this only $\log\Delta$ times,  we  can $\Delta$-approximate
$\beta$.   That is, we have the following:
\begin{lemma}
Given a $\Delta$-approximation $\tilde{\beta}$ of $\beta$, we can
compute an $(\alpha(1+O(\epsilon))$-approximate solution using at
most $\tilde{O}((\tau+k)\epsilon^{-2}\log\Delta)$ approximate shortest
path computations.\label{lem:concurrent time}
\end{lemma}

\subsection{Simulating Fleischer's algorithms}

Let $G=(V,E)$ be an  undirected $n$-vertex $m$-edge graph with
vertex capacities $c\in(\mathbb{R}_{>0}\cup\{\infty\})^{V}$. Let
$(s_{1},t_{1}),\dots,(s_{k},t_{k})$ be the demand pairs. Recall that
we assume $c(v)=\infty$ for all $v\in\{s_{1},t_{1},\dots,s_{k},t_{k}\}$.
Next, we create a directed graph $G''=(V'',E'')$ with edge capacities
$c''$ as follows. For each node $v\in V$, we create $v^{in},v^{out}\in V''$
and a directed edge $(v^{in},v^{out})$ with capacity $c''(v^{in},v^{out})=c(v)$.
For each undirected edge $(u,v)\in E$, we create $(v^{out},u^{in}),(u^{out},v^{in})\in E''$
with capacity $\infty$. This is a standard reduction from flow problems
with vertex capacities to flow problems with edge capacities on directed
graphs. Given a multi-commodity flow $\F=\{f_{1},\dots,f_{k}\}$ with
demand pairs $(s_{1},t_{1}),\dots,(s_{k},t_{k})$, there is a corresponding
$\F''=\{f''_{1},\dots,f_{k}''\}$ with demand pairs $(s_{1}^{out},t_{1}^{in}),\dots,(s_{k}^{out},t_{k}^{in})$
such that, for each $j$, $|f_{j}|=|f''_{j}|$ and $f_{j}$ capacity-feasible
iff $f_{j}''$ is capacity-feasible.

Now, we can run \Cref{alg:throughput} and \Cref{alg:concurent} on
$G''$. Both algorithms maintain lengths $\ell''\in\mathbb{R}_{\ge0}^{E''}$
of edges. Observe that $\ell''(e'')=0$ unless $e''=(v^{in},v^{out})$
and $v\notin\{s_{1},t_{1},\dots,s_{k},t_{k}\}$. To be able to compute
approximate shortest paths in $G''$ quickly, we assign the weights
$w'\in\mathbb{R}_{\ge0}^{V}$ to vertices in $G$ as follows. For
each $v$, $w'(v)=\ell''(v^{in},v^{out})$. Observe that, for every
demand pair $(s_{j},t_{j})$, there is a correspondence between $s_{j}$-$t_{j}$
paths $P$ in $G$ and $s_{j}^{out}$-$t_{j}^{in}$ paths $P''$ in
$G''$ where $w'(P)=\ell''(P'')$. Note that this is true because
$w'(v)=0$ for $v\in\{s_{1},t_{1},\dots,s_{k},t_{k}\}$. However,
the weights of $G$ are now in vertices. To assign the weight $w\in\mathbb{R}_{\ge0}^{E}$
on edges, we set $w(u,v)=(w'(u)+w'(v))/2$. As $w'(v)=0$ for $v\in\{s_{1},t_{1},\dots,s_{k},t_{k}\}$,
we have that $w(P)=w'(P)$ for any $s_{j}$-$t_{j}$ path $P$ in
$G$.

Now, we are ready to show how to exploit dynamic spanner algorithms
for fast flow computations.
\begin{theorem}
Suppose there is a decremental spanner against an adaptive adversary
that can maintain $t$-spanner of size $s$ using $u$ amortized update
time and $q$ query time. Then, in undirected vertex-capacitated graphs with $n$ vertices
and $m$ edges, we have 
\begin{enumerate}
\item  a $(t\cdot(1+\epsilon))$-approximate algorithm for maximum throughput
flow with $k$ commodities in $\tilde{O}( (mu+n(s+q))\epsilon^{-2} )$ time,
and
\item  a $(t\cdot(1+\epsilon))$-approximate algorithm for maximum concurrent
flow with $k$ commodities in $\tilde{O}( (mu+(n+k)(s+q))\epsilon^{-2}\log C)$
time where $C$ is the ratio of largest finite capacity to smallest
finite capacity.
\end{enumerate}
\end{theorem}

\begin{proof}
\textbf{(1):} By \Cref{lem:throughput weight}, as the lengths only
increase, we can use decremental spanner algorithms on $G$ for maintaining
a $t$-spanner $H$. We will update the weight $w(e)$ of each edge
only when it has been increased by more than a factor of $2$. Again,
by \Cref{lem:throughput weight}, we only need to update each edge
at most $\log\frac{1+\epsilon}{\delta}=O(\epsilon^{-2}\log m)$ times.
Therefore the total update time to maintain $H$ is $O(mu\epsilon^{-2}\log m)$.
Next, by \Cref{lem:throughput time}, we need to compute single-source
shortest path $\tilde{O}(n/\epsilon^{-2})$ times. Each time, we first query for $H$ in $q$ time, and then compute
the shortest path on $H$ to get a $t$-approximate single source shortest path
in $\tilde{O}(s)$ time. Therefore the total running time is $\tilde{O}((mu+n(s+q))/\epsilon^{2})$.
By \Cref{lem:throughput feasible} and \Cref{lem:throughput optimal},
we obtain a $(t\cdot(1+\epsilon))$-approximation  from \Cref{alg:throughput}.

\textbf{(2):} We will first show how to get an  $O(kn)$-approximation
$\tilde{\beta}$ of the optimal value $\beta$ in time $O((m+k)\log C)$.
Let $\zeta_{j}$ denote the maximum $s_{j}$-$t_{j}$ flow value in
$G$. Let $\zeta=\min_{j}\zeta_{j}/d_{j}$. As argued in \cite{GargK07,Fleischer00},
$\zeta$ is a $k$ approximation of $\beta$. So now it suffices to
$O(n)$-approximate $\zeta_{j}$ for each $j$. As any $s_{j}$-$t_{j}$
flow in vertex-capacitated graphs can be decomposed into $n$ paths,
we have that $\hat{\zeta}_{j}=\max_{P\in\P_{s_{j},t_{j}}}\min_{u\in P}c(u)$,
i.e. the amount of flow that we can send along maximum min-capacity
path, is an $n$-approximation of $\zeta_{j}$. In \cite{Fleischer00},
it is shown how to compute all $\hat{\zeta}_{j}$ in $O(\min\{k,n\}m\log m)$
total time which is too slow  for us. Here, we will show how to $2$-approximate
all $\hat{\zeta}_{j}$ in $O((m+k)\log C)$ time. 

Let $c_{\max}=\max_{v}\{c(v)\mid c(v)\neq\infty\}$, $c_{\min}=\min_{v}c(v)$,
and $C=c_{\max}/c_{\min}$. Let $V_{i}=\{v\mid2^{i}c_{\min}\le c(v)<2^{i+1}c_{\min}\}$.
Let $V_{\infty}=\{v\mid c(v)=\infty\}$. Let $G_{i}=G[\bigcup_{j\ge i}V_{j}]$.
If $(s_{j},t_{j})$ is connected in $G_{i}$ but not in $G_{i+1}$,
then we know $2^{i}c_{\min}\le\hat{\zeta}_{j}<2^{i+1}c_{\min}$. For
each $i$, we compute connected component of $G_{i}$ in $O(m)$ time.
Then we can check quickly, for all $j$, if $s_{j}$ and $t_{j}$
are connected in $G_{i}$ in $O(k)$. Summing over all $i$, the total
time is $O((m+k)\log C)$. Therefore, we obtain $O(kn)$-approximation
$\tilde{\beta}$ of $\beta$.

Now, the proof is very similar to (1). By \Cref{lem:concurrent weight},
as the lengths only increase, we use a decremental spanner algorithms
on $G$ for maintaining a $t$-spanner $H$, and we will update the
weight $w(e)$ of each edge only when it has been increased by more
than a factor of $2$. By \Cref{lem:concurrent weight}, we only need
to update each edge at most $\log\frac{1+\epsilon}{\delta}=O(\epsilon^{-2}\log m)$
times. Therefore the total update time to maintain $H$ is $O(mu\epsilon^{-2}\log m)$.
Next, by \Cref{lem:concurrent time}, we need to compute single-source
shortest paths  $\tilde{O}((n+k)\epsilon^{-2}\log(nk))$ times. 
Each time, we first query for $H$ in $q$ time, and then compute
the shortest path on $H$ to get a $t$-approximate single source
shortest path in $\tilde{O}(s)$ time. Therefore the total running
time is $\tilde{O}((mu+(n+k)(s+q))/\epsilon^{2})$. By \Cref{lem:concurrent feasible}
and \Cref{lem:concurrent optimal}, we obtain a $(t\cdot(1+\epsilon))$-approximate
from \Cref{alg:concurent}. 
\end{proof}

To prove \Cref{thm:multiflow}, we set $\epsilon=\Omega(1)$. Then we use the dynamic data structure for spanners from  \Cref{thm:SpannerMainResult}  with parameters $t=\polylog(n),s=\tilde O(n),u=\tilde O(1),q=\tilde O(n)$. This
concludes \Cref{thm:multiflow}.

\subsection{Proof of \Cref{cor:cong min}}
\begin{proof}
Given an instance of the congestion minimization problem with $k$
demand pairs $\{(s_{j},t_{j})\}_{j=1}^{k}$. Let $\mathsf{OPT}$ be
the minimum congestion over all solutions. We solve the same instance
of maximum concurrent flow when $d_{j}=1$ for all $j$. Let $\lambda$
be the value of our maximum concurrent flow solution $\F=\{f{}_{1},\dots,f{}_{k}\}$ from
\Cref{thm:multiflow}. Let $\F'=\{f'_{1},\dots,f'_{k}\}$ be such that
$f'_{i}=f_{i}/\lambda$. Note that $f'_{i}$ is a $s_{i}$-$t_{i}$
flow of value $1$. Let the \emph{congestion} of $\F'$ denote $\max_{v}\frac{\sum_{i}f_{i}'(v)}{c(v)}$.
Note that $\max_{v}\frac{\sum_{i}f_{i}'(v)}{c(v)}\le \polylog(n)\mathsf{OPT}$.

The standard rounding technique by Raghavan and Thompson \cite{RaghavanT87}
works as follows. For each $i$, recall that \Cref{thm:multiflow}
gives flow-path decomposition of $f_{i}$ as well, and hence of $f'_{i}$.
Let $P_{i}^{1},\dots,P_{i}^{z_{i}}$ be the paths in the given flow-path
decomposition of $f'_{i}$. Let $f_{i}'(P)$ be the amount of flow
$f'_{i}$ on $P$. Note that $\sum_{j}f'_{i}(P_{i}^{j})=1$ and $f'_{i}(P_{i}^{j})\ge0$.
We choose $P_{i}^{j}$ with probability $f'_{i}(P_{i}^{j})$ into
the solution for congestion minimization. Let $\P$ be the collection
of such paths. By Chernoff's bound, we have the congestion of $\P$
is at most \[
O\left(\frac{\log n}{\log\log n}\right)\cdot\max_{v}\frac{\sum_{i}f_{i}'(v)}{c(v)}=\polylog(n)\mathsf{OPT}.\qedhere
\]
\end{proof}

\section{Tables}
\label{sec:tables}

\begin{table}[H]
	\begin{centering}
		\begin{tabular}{|c|c|c|c|c|c|}
			\hline 
			\textbf{Reference} & \textbf{Year} & \textbf{Stretch} & \textbf{Size} & \textbf{Time} & \textbf{Adaptive?}\tabularnewline
			\hline 
			\hline 
			\multicolumn{6}{|l|}{\textbf{Incremental}}\tabularnewline
			\hline 
			\cite{FeigenbaumKMSZ05} & 2005 & $2k-1$ & $O(kn^{1+1/(k-1)})$ & $\tilde{O}(kn^{1/(k-1)})$  & \tabularnewline
			\hline 
			\cite{Baswana08} & 2008 & $2k-1$ & $O(kn^{1+1/k})$ & $O(1)$ & \tabularnewline
			\hline 
			\cite{Elkin11} & 2007 & $2k-1$ & $O(kn^{1+1/k})$ & $O(1)$  & \tabularnewline
			\hline 
			\hline 
			\multicolumn{6}{|l|}{\textbf{Fully dynamic (amortized)}}\tabularnewline
			\hline 
			\multirow{2}{*}{\cite{AusielloFI06}} & \multirow{2}{*}{2005} & 3 & $O(n^{1+1/2})$ & $O(\Delta)$ & adaptive \tabularnewline
			\cline{3-6} 
			&  & 5 & $O(n^{1+1/3})$ & $O(\Delta)$ & adaptive \tabularnewline
			\hline 
			\multirow{2}{*}{\cite{BaswanaKS12}} & \multirow{2}{*}{2008} & $2k-1$ & $\tilde{O}(k^{9}n^{1+1/k})$ & $O(7^{k})$ & \tabularnewline
			\cline{3-6} 
			&  & $2k-1$ & $\tilde{O}(kn^{1+1/k})$ & $O(k^{2}\log^{2}n)$ & \tabularnewline
			\hline 
			\multirow{2}{*}{\cite{ForsterG19}} & \multirow{2}{*}{2019} & $2k-1$ & $\tilde{O}(n^{1+1/k})$ & $O(k\log^{2}n)$ & \tabularnewline
			\cline{3-6} 
			&  & $\boldsymbol{2\log n-1}$ & $\boldsymbol{\tilde{O}(n)}$ & $\boldsymbol{\tilde{O}(1)}$ & \tabularnewline
			\hline 
			\textbf{This paper} &  & $\boldsymbol{\Otil(1)}$ & $\boldsymbol{\Otil(n \log W)}$ & $\boldsymbol{\Otil(1)}$ & adaptive\tabularnewline
			\hline 
			\hline 
			\multicolumn{6}{|l|}{\textbf{Fully dynamic (worst-case)}}\tabularnewline
			\hline 
			\cite{Elkin11} & 2007 & $2k-1$ & $O(n^{1+1/k})$ & $\tilde{O}(mn^{-1/k})$ & \tabularnewline
			\hline 
			\multirow{3}{*}{\cite{BodwinK16}} & \multirow{3}{*}{2016} & 3 & $\tilde{O}(n^{1+1/2})$ & $\tilde{O}(n^{3/4})$ & \tabularnewline
			\cline{3-6} 
			&  & 5 & $\tilde{O}(n^{1+1/3})$ & $\tilde{O}(n^{5/9})$ & \tabularnewline
			\cline{3-6} 
			&  & 5 & $\tilde{O}(n^{1+1/2})$ & $\tilde{O}(n^{1/2})$ & \tabularnewline
			\hline 
			\multirow{2}{*}{\cite{BernsteinFH19}} & \multirow{2}{*}{2019} & $2k-1$ & $\tilde{O}(n^{1+1/k})$ & $O(1)^{k}\log^{3}n$ &
			\tabularnewline
			\cline{3-6} 
			&  & $\boldsymbol{2\sqrt{\log n}-1}$ & $\boldsymbol{n2^{O(\sqrt{\log n})}}$ & $\boldsymbol{2^{O(\sqrt{\log n})}}$ & \tabularnewline
			\hline 
			\textbf{This paper} &  & $\boldsymbol{n^{o(1)}}$ & $\boldsymbol{\Otil(n \log W)}$ & $\boldsymbol{n^{o(1)}}$ & adaptive\tabularnewline
			\hline 
		\end{tabular}
		\par\end{centering}
	\caption{A list of dynamic spanner algorithms. Above, $n$ is a number of vertices,
		$m$ is the number of initial edges, $\Delta$ is the maximum degree, and $k$ is any positive integer.
		In the ``Year'' column, the year is according to the conference
		version of the paper. If an algorithm assumes an oblivious adversary,
		then we leave blank in the column ``Adaptive?''. All previous algorithms
		that can handle adaptive updates are deterministic algorithms, while
		our algorithms are randomized Monte Carlo. All incremental algorithms
		have worst-case update time except \cite{Baswana08}. We highlight
		\textbf{(in bold)} the results with spanner size at most $n^{1+o(1)}$
		in $n^{o(1)}$ time.\label{tab:spanner}
		}
\end{table}

\bibliographystyle{alpha}
\bibliography{bibliography}

\end{document}